%% file: AntipodalOverloadedIEEEOpenAccess02.tex
\let\texyear\year
\RequirePackage{expl3}
\documentclass[a4paper]{ieeeaccess}

\IEEEoverridecommandlockouts
\normalsize
\usepackage{amsmath,amssymb,amsfonts}
\usepackage{algorithmic}
\usepackage{graphicx}
\usepackage{textcomp}
\usepackage{upgreek}

\usepackage{float}
\usepackage{mathtools}
\usepackage{algorithm}
\usepackage{amsthm}
\usepackage{tablefootnote}
\usepackage{sansmath}
\usepackage{isomath}
\usepackage{threeparttable}
\usepackage[mathscr]{euscript}
\usepackage{footmisc}

\usepackage{multirow}

\usepackage{enumitem}
\usepackage{adjustbox}
\usepackage{longtable}
\usepackage{soul}

\usepackage[noadjust]{cite}
\usepackage{url}
\usepackage{dblfloatfix}    

\floatname{algorithm}{Algorithm}

\newtheorem{claim}{Claim}

\newcommand{\be}{\begin{equation}}
\newcommand{\ee}{\end{equation}}
\newcommand{\ba}{\begin{array}}
	\newcommand{\ea}{\end{array}}
\newcommand{\nid}{\noindent}

\newcommand{\m}{\hspace{-.05cm}}
\renewcommand{\thesubsubsection}{\thesubsection.\arabic{subsubsection}}
\newtheorem{theorem}{Theorem}

\theoremstyle{definition}

\newcommand{\argmin}{\operatornamewithlimits{argmin}}

\renewcommand*{\thefootnote}{\fnsymbol{footnote}}
\usepackage{enumitem}

 \include{nomenclature}
\usepackage[utf8]{inputenc}
\usepackage[T1]{fontenc}

\let\ieeeaccessyear\year
\let\year\texyear

\include{timeline}
  
\let\year\ieeeaccessyear
\definecolor{accessblue}{RGB}{0,105,154}
\definecolor{greycolor}{cmyk}{0,0,0,.8}
\definecolor{grey}{cmyk}{0,0,0,.1}
\definecolor{black}{cmyk}{0,0,0,1}

\def\thevol{5}
\def\myyear{2020}
\makeatletter
\patchcmd{\@evenfoot}{2020}{\myyear}{}{}
\patchcmd{\@oddfoot}{2020}{\myyear}{}{}
\makeatother
\def\BibTeX{{\rm B\kern-.05em{\sc i\kern-.025em b}\kern-.08em
    T\kern-.1667em\lower.7ex\hbox{E}\kern-.125emX}}

\begin{document}
\history{Date of publication xxxx 00, 0000, date of current version xxxx 00, 0000.}
\doi{11.1109/ACCESS.2021.DOI}

		\title { Low-Complexity Decoder for Overloaded Uniquely Decodable Synchronous CDMA}

\author{\uppercase{Michel~Kulhandjian}\authorrefmark{1}, \IEEEmembership{Senior Member, IEEE},
\uppercase{Hovannes~Kulhandjian}\authorrefmark{2}, \IEEEmembership{Senior Member, IEEE},
\uppercase{Claude~D'Amours}\authorrefmark{1}, \IEEEmembership{Member, IEEE}, \uppercase{Halim Yanikomeroglu}\authorrefmark{3},
\IEEEmembership{Fellow, IEEE}, \uppercase{Dimitris A. Pados }\authorrefmark{4},
\IEEEmembership{Senior Member, IEEE} and \uppercase{Gurgen~Khachatrian}\authorrefmark{5},
\IEEEmembership{Member, IEEE}}
\address[1]{School of Electrical Engineering and Computer Science, University of Ottawa, Ottawa, Ontario, K1N 6N5, Canada (e-mail: mkk6@buffalo.edu,cdamours@uottawa.ca)}
\address[2]{Department of Electrical and Computer Engineering, California State University, Fresno, Fresno, CA 93740, U.S.A. (e-mail: hkulhandjian@csufresno.edu)}
\address[3]{Department of Systems \& Computer Engineering,
	Carleton University, Ottawa, Canada (e-mail: halim@sce.carleton.ca)}
\address[4]{Computer \& Electrical Engineering \& Computer Science \& I-SENSE Center, Florida Atlantic University, U.S.A., (e-mail: dpados@fau.edu)}
\address[5]{College of Science \& Engineering,
		American University of Armenia, Armenia, (e-mail: gurgenkh@aua.am)}

\corresp{Corresponding author: Michel Kulhandjian (e-mail: mkk6@buffalo.edu).}

	
	 \begin{abstract}
	We consider the problem of designing a low-complexity decoder for antipodal uniquely decodable (UD) /errorless code sets for overloaded synchronous code-division multiple access (CDMA) systems, where the number of signals $K_{\rm{max}}^a$ is the largest known for the given code length $L$. In our complexity analysis, we illustrate that compared to maximum-likelihood (ML) decoder, which has an exponential computational complexity for even moderate code lengths, the proposed decoder has a quasi-quadratic computational complexity. Simulation results in terms of bit-error-rate (BER) demonstrate that the performance of the proposed decoder has only a $1-2$ dB degradation in signal-to-noise ratio (SNR) at a BER of $10^{-3}$ when compared to ML. Moreover, we derive the proof of the minimum Manhattan distance of such UD codes and we provide the proofs for the propositions; these proofs constitute the foundation of the formal proof for the maximum number users $K_{\rm{max}}^a$ for $L=8$.  \end{abstract}


     \maketitle
     \printacronyms[include=nomencl,name=NOMENCLATURE] 

\newpage

	\section{{Introduction}}
	\renewcommand*{\thefootnote}{\arabic{footnote}}
	\renewcommand{\thefootnote}{\roman{footnote}}
	 \IEEEPARstart{I}{n} the last decade, wireless communication services have experienced explosive growth while communication technologies have progressed generation by generation. In the previous generations spanning from 1G to 4G, the multiple access schemes were mostly characterized by orthogonal multiple access (\acs{OMA}) techniques, where users are assigned orthogonal resources in either frequency, (frequency-division multiple access (\acs{FDMA})), time, (time-division multiple access (\acs{TDMA})) or code, (code-division multiple access (\acs{CDMA})). CDMA \cite{Viterbi1995} was the basic technology for 3G and for some 2G (IS-95) networks. High spectral- and power-efficiency, massive connectivity and low latency are among the requirements for next generation communications and these requirements are expected to increase in the future, as researchers turn their efforts towards sixth generation (\acs{6G}) wireless communications. Enhanced mobile broadband (\acs{eMBB}), ultra-reliable low-latency communications (\acs{uRLLC}) and massive machine-type communication (\acs{MMTC}) support a suite of compelling applications driving these requirements. Massive multiple-input multiple-output (\acs{MIMO}), non-orthogonal multiple access (\acs{NOMA}) and millimeter-wave (\acs{mmWave}) communications constitute promising techniques of addressing these stringent requirements \cite{Dai2018}.
	 
	 Supporting a large number of users communicating over a common channel may not be readily achievable by OMA techniques due to the multiple-access interference (\acs{MAI}) in rank-deficient systems, where the number of users is greater than that of the resource blocks. To meet the demand of increased bandwidth  efficiency in synchronous CDMA, a CDMA concept was introduced in \cite{HSari2000}, which can support many more users for a given code length compared to traditional CDMA. A number of signature design schemes have been studied for dense spreading in conventional CDMA, where low cross-correlation sequence sets are designed to minimize the overall MAI, which allows more users to simultaneously access the common channel. This in turn results in increased spectral efficiency. 
  Finding suitable spreading codes and decoding schemes for such overloaded systems is a challenging optimization problem. To address these challenges, numerous non-uniquely decodable (non-\acs{UD}) \cite{ Grant1998, Verdu1999, Vanhaverbeke2000, Sari2000, HSari2000, Vanhaverbeke2002, Vanhaverbeke2004, Ming2009}, and UD \cite{Erdos1963, Shapiro1963, Cantor1964, Lindstrom1964, Lindstrom1965, Cantor1966, Erdos1970,  Chang1979, Ferguson1982, Chang1984, Khachatrian1987, Khachatrian1989, Khachatrian1995, Khachatrian1998, Wu1997, Mow2009, Marvasti2009, ming2016, Marvasti2012, Marvasti2016, michel2012, michel2018, michelHindawi2018, michelWCNC2019, michelTernary2021} construction based code sets have been proposed. Examples of such non-UD code sets are pseudo-noise spreading (\acs{PN}) \cite{Grant1998, Verdu1999}, orthogonal/orthogonal CDMA (O/O), \cite{Vanhaverbeke2000, Sari2000}, PN/orthogonal CDMA (PN/O) \cite{HSari2000}, multiple-orthogonal CDMA (\acs{MO}) \cite{Vanhaverbeke2002}, improved O/O CDMA \cite{Vanhaverbeke2004}. Those codes employ two or more sets of orthogonal signal waveforms, which allows the system to accommodate more users than the signature length $L$. As a consequence of this, a significant level of MAI exists at the output of each user's matched filter due to the non-zero cross-correlation of different signatures. 
 
 Low cross-correlation sequence sets might not be the best criterion for very high rank-deficient systems. One important criterion in such rank-deficient systems is for the code set to be UD. By definition the UD codes are those in which the data of different users can be unambiguously decoded in a noiseless channel using linear recursive decoders \cite{michel2012}. Low-complexity linear decoders were introduced for these UD code sets using either binary $\{0,1\}$, or antipodal $\{\pm 1\}$, or alternatively ternary $\{0,\pm 1\}$ chips in \cite{Marvasti2009, ming2016,  Marvasti2012, Marvasti2016, michel2018, michelHindawi2018, michelWCNC2019, michelTernary2021}. On the other hand, Lu \emph{et al.} \cite{Kamabe2019} proposed $M$-ary code sets for the multiple-access adder channel. Various applications have been implementing UD code sets framework such as in multi-way physical-layer network coding conceived in \cite{michelHindawi2018,Xiang2019}.
 
 All of these multiple access concepts were introduced in order to serve a number of excess users beyond the available resources. These multiple access schemes are characterized by NOMA techniques \cite{Dai2015} in \acs{5G} and beyond wireless communications. Recently, several NOMA solutions have been actively investigated \cite{Dai2018}, which can be basically divided into two main categories, namely power-domain and code-domain NOMA. A few of the strong contenders of code-domain NOMA are low-density spreading aided CDMA (\acs{LDS}-CDMA) \cite{Hoshyar2008,MichelHanzo2021} and sparse code multiple access (\acs{SCMA}) \cite{Nikopour2013, michel2017}.
LDS-CDMA \cite{Hoshyar2008,MichelHanzo2021} must generally guarantee to be UD code set \cite{JVan2009}, which means non-zero Euclidean distance. The design of LDS type matrices offers flexible resource allocation, performs better in terms of handling the MAI that exists in rank-deficient systems and has low-complexity receivers compared to conventional CDMA. 

Unlike the dense UD code set, the LDS-CDMA structure can be represented by a factor graph, the classic message passing algorithm (MPA) can be employed for its detection. Due to the reason that UD code sets were originally designed for adder channels, generally a simple noiseless detection is developed. However in practice, the wireless transmission channel exhibits, among other things, selective fading, multipath and the near-far problem, which leads to unequal received power among users. Consequently, if synchronization, channel equalization are compensated, low-complexity detectors can be applied to wireless channels. 

Inspired by these attractive features of UD code sets, this paper investigates UD codes for synchronous uplink NOMA. In the case of dispersive fading channels, we can potentially equalize the channel effect by using channel precoding. Therefore, we consider developing a low-complexity detector for UD code sets that are proposed in \cite{michel2012} for overloaded synchronous CDMA systems. It is widely recognized that the complexity of an optimal detector is exponentially proportional to the number of users, which prohibits its practical implementation. Various suboptimal low-complexity detection techniques have been already proposed. These suboptimal approaches can be classified into two categories: linear and non-linear multiuser detectors. Linear multiuser detectors include among others, matched filter (\acs{MF}), minimum mean-square error (\acs{MMSE}), and zero-forcing (\acs{ZF}), etc. In a non-linear subtractive interference cancellation detector the interference is first estimated and then it is subtracted from the received signal before detection. The cancellation process can be carried out either successively (\acs{SIC}) \cite{Kobayashi2001}, or in parallel (\acs{PIC}) \cite{Varanas1990,Xue1999, Guo2000}. In non-linear iterative detectors \cite{Reed1998,Wang1999,  Morosi2007, Kumar2010,  Sasipriya2014},  probabilistic data association (\acs{PDA}) \cite{Romano2005} aims to suppress the MAI in each iteration in order to improve the overall error performance. Suboptimal polynomial time detectors that are based on the geometric approach are studied in \cite{Najkha2005,Manglani2006}.  

In general, linear as well as non-linear detectors cannot separate users in overloaded systems even in the case of asymptotically vanishing noise. Therefore, the spreading codes must have property such that decoding can achieve asymptotically zero probability of error multiuser detection when the signal-to-noise (\acs{SNR}) ratio becomes arbitrary large. The UD class of codes that guarantee ``errorless'' communication in an ideal (noiseless) synchronous CDMA/code-division multiplexing (\acs{CDM}) also shows a good performance in the presence of noise. 

Finding the overloaded UD class of codes for noiseless channel is directly related to coin-weighing problem, one of the problems that is discussed by Erd\H{o}s and R\'enyi in \cite{Erdos1963}. It can be considered as a special case of a general problem where authors in  \cite{Shapiro1963, Cantor1964, Cantor1966, Wu1997, Mow2009} refer to them as detecting matrices. Lindstr\"{o}m in \cite{Lindstrom1965} defines the same problem as the detecting set of vectors. Given an integer $q \geq 2$ and a finite set alphabet $\mathcal{M}$ of rational integers, let $\mathbf{v}_k$ for $1\leq k \leq K$ be $L$-dimensional (column) vectors with all components from $\mathcal{M}$ such that the $q^K$ sums
\begin{equation}
\sum_{k = 1}^K   \mathbf{v}_k \epsilon_k  \: \: (\epsilon_k = 0,1,2,\dots, q-1)
\end{equation}
\nid are all distinctly unique, then $\{ \mathbf{v}_1, \dots, \mathbf{v}_K \}$ are detecting set of vectors. Let $\mathsf{F}_q(L)$ be the maximal number of $L$-dimensional vectors and $\mathsf{f}_q(K)$ be the minimal vector length to form a detecting matrix for a given length $L$ and a number of vectors $K$. The problem of determining $\mathsf{f}_q(K)$ as a special case when $q=2$, $\mathcal{M} = \{0,1\}$ that can be equivalently expressed as a coin-weighing problem: what is the minimal number of weighings on an accurate scale to determine all false coins in a set of $K$ coins. The choice of coins for a weighing must not depend on results of previous weighings. This problem was first introduced by S\"oderberg and Shapiro \cite{Shapiro1963} for $K=5$. The minimal number of weighings, $L$, has only been found for a few different values of $K$ in \cite{Erdos1970}. However, Lindstr\"{o}m gives an explicit construction of $L \times  \upgamma(L+1)$ binary (alphabet $\{0, 1\}$) and $L \times \upgamma(L)+1$ antipodal (alphabet $\{\pm 1\}$) detecting matrices \cite{Lindstrom1964}, where $\upgamma(L)$\footnote{As an example, $\upgamma(8) = 12$.} is the number of ones in the binary expansion of all positive integers less than $L$. 
He also proved that the lower bound in the case of $\mathcal{M}=\{0,1\}$ or $\{\pm1\}$ is
\begin{equation}
\lim_{K\to\infty} \frac{\mathsf{f}_2(K) \log{K}}{K} =2.
\end{equation}
Cantor and Mills \cite{Cantor1966} constructed a class of $2^i \times (i+2)2^{(i-1)}$ ternary (alphabet $\{0, \pm 1\}$) detecting matrices for $i \in \mathbb{Z}^+$, which implies that in the case of $\mathcal{M}=\{0,\pm1\}$ the lower bound is
\begin{equation}
\lim_{K\to\infty} \frac{\mathsf{f}_3(K) \log{K}}{K} \leq 2.
\end{equation}
In the literature, most of the explicit construction of UD code sets are recursive \cite{Chang1979, Ferguson1982, Chang1984, Khachatrian1987, Khachatrian1989, Khachatrian1995, Khachatrian1998, Wu1997, Mow2009, Marvasti2009, ming2016, Marvasti2012, Marvasti2016, michel2012, michel2018, michelHindawi2018, michelWCNC2019, michelTernary2021}. To the best of our knowledge, it is worth mentioning that the maximum number 
of vectors of the explicit constructions of binary, antipodal and ternary code sets are $K_{\rm{max}}^b = \upgamma(L+1)$, $K_{\rm{max}}^a = \upgamma(L) +1$ and $K_{\rm{max}}^t = (i+2)2^{(i-1)}$,
as shown in Table \ref{table:binary}, Table \ref{table:antipodal} and Table \ref{table:ternary}, respectively. Several authors have proposed decoders with linear complexity in noiseless scenarios for the explicit construction, where the detecting matrix has the known-to-us maximum number of vectors, $K_{\rm{max}}$.  

\begin{table*}[h]
	\caption{Binary Codes} 
	\centering 
	\begin{threeparttable}
		\begin{tabular}{l l c c c c} 
			\hline\hline  
			\multicolumn{1}{c}{\multirow{2}{*}[-1.5pt]{\bf{Year}}} &
			\multicolumn{1}{l}{\multirow{2}{*}[-1.5pt]{\bf{Authors and Publications}}}  & \multicolumn{1}{c}{\multirow{2}{*}[-1.5pt]{\bf{Rows}}} & \multicolumn{1}{c}{\multirow{2}{*}[-1.5pt]{\bf{Columns}}} & \multicolumn{2}{c}{\multirow{1}{*}[-1.5pt]{\bf{Decoder}}} \\[0.5ex]  \cline{5-6}
			&  &  &  & \multirow{1}{*}[-1.5pt]{\bfseries{Noiseless}} & \multirow{1}{*}[-1.5pt]{\bfseries{AWGN}}\\ [1.0ex]
			\hline   \rule{-3pt}{2.5ex} 
			1963 & S\"{o}derberg and Shapiro \cite{Shapiro1963} & $L$  & $<\upgamma(L+1)$ & No  & No\\[0.6ex]
			1964 & Lindstr\"{o}m  \cite{Lindstrom1964}  & $L$ & $\bf{\boldsymbol{\upgamma}(L+1)}$\tnote{\dag} & No & No \\[0.6ex]
			1966 & Cantor and Mills \cite{Cantor1966}  & $2^i-1$ & ${i2^{(i-1)}}$ & No & No \\[0.6ex]
			1989 & Martirossian and Khachatrian \cite{Khachatrian1989}  & $L$ & $\bf{\boldsymbol{\upgamma}(L+1)}$ & Yes & No \\[0.6ex]
			2019 & Kulhandjian \textit{et al.} \cite{michelWCNC2019}  & $L$ & $\bf{\boldsymbol{\upgamma}(L+1)}$ & Yes & Yes \\[0.6ex]
			\hline 
		\end{tabular}
		\footnotesize
		\begin{tablenotes}
			\item[\dag] Code set constructions that achieve the maximum number of vectors $\mathbf{K}_{\rm{max}}$ are presented in bold.
		\end{tablenotes}
	\end{threeparttable}
	
	\label{table:binary}
\end{table*}

\begin{table*}[h]
	\caption{Antipodal Codes} 
	\centering 
	\begin{tabular}{l l c c c c} 
		\hline\hline  
		\multicolumn{1}{c}{\multirow{2}{*}[-1.5pt]{\bf{Year}}} &
		\multicolumn{1}{l}{\multirow{2}{*}[-1.5pt]{\bf{Authors and Publications}}}  & \multicolumn{1}{c}{\multirow{2}{*}[-1.5pt]{\bf{Rows}}} & \multicolumn{1}{c}{\multirow{2}{*}[-1.5pt]{\bf{Columns}}} & \multicolumn{2}{c}{\multirow{1}{*}[-1.5pt]{\bf{Decoder}}} \\[0.5ex]  \cline{5-6}
		&  &  &  & \multirow{1}{*}[-1.5pt]{\bfseries{Noiseless}} & \multirow{1}{*}[-1.5pt]{\bfseries{AWGN}}\\ [1.0ex]
		\hline   \rule{-3pt}{2.5ex} 
		1964 & Lindstr\"{o}m \cite{Lindstrom1964} & $L$  & $\bf{\boldsymbol{\upgamma}(L)+1}$ & No  & No\\[0.6ex]
		1987 & Khachatrian and Martirossian \cite{Khachatrian1987}  & $L$ & $\bf{\boldsymbol{\upgamma}(L)+1}$ & No & No \\[0.6ex]
		1995 & Khachatrian and Martirossian \cite{Khachatrian1995}  & $2^i$ & $\bf{i2^{(i-1)}+1}$ & Yes & No \\[0.6ex]
		2012 & Kulhandjian and Pados \cite{michel2012}  & $2^i$ & $\bf{i2^{(i-1)}+1}$ & Yes & No \\[0.6ex]
		\hline 
	\end{tabular}
	\label{table:antipodal}
\end{table*}

\begin{table*}[h]
	\caption{Ternary Codes} 
	\centering 
	\begin{tabular}{l l c c c c} 
		\hline\hline  
		\multicolumn{1}{c}{\multirow{2}{*}[-1.5pt]{\bf{Year}}} &
		\multicolumn{1}{l}{\multirow{2}{*}[-1.5pt]{\bf{Authors and Publications}}}  & \multicolumn{1}{c}{\multirow{2}{*}[-1.5pt]{\bf{Rows}}} & \multicolumn{1}{c}{\multirow{2}{*}[-1.5pt]{\bf{Columns}}} & \multicolumn{2}{c}{\multirow{1}{*}[-1.5pt]{\bf{Decoder}}} \\[0.5ex]  \cline{5-6}
		&  &  &  & \multirow{1}{*}[-1.5pt]{\bfseries{Noiseless}} & \multirow{1}{*}[-1.5pt]{\bfseries{AWGN}}\\ [1.0ex]
		\hline   \rule{-3pt}{2.5ex} 
		1966 & Cantor and Mills \cite{Cantor1966} & $2^i$  & $\bf{(i+2)2^{(i-1)}}$ & No  & No\\[0.6ex]
		1979 & Chang and Weldon \cite{Chang1979}  & $2^i$ & $\bf{(i+2)2^{(i-1)}}$ & Yes & No \\[0.6ex]
		1982 & Ferguson \cite{Ferguson1982}   & $2^i$ & $\bf{(i+2)2^{(i-1)}}$ & Yes & No \\[0.6ex]
		1984 & Chang \cite{Chang1984}  & $2^i$ & $\bf{(i+2)2^{(i-1)}}$ & No & No \\[0.6ex]
		1998 & Khachatrian and Martirossian \cite{Khachatrian1998}  & $2^i$ & $\bf{(i+2)2^{(i-1)}}$ & Yes & No \\[0.6ex]
		2012 & Mashayekhi and Marvasti \cite{Marvasti2012}  & $2^i$ & $2^{(i+1)}-1$ & Yes & Yes \\[0.6ex]
		2016 & Singh \textit{et al.} \cite{Marvasti2016}   & $2^i$ & $2^{(i+1)}-2$ & Yes & Yes \\[0.6ex]
		2018 & Kulhandjian \textit{et al.} \cite{michel2018}   & $2^i$ & $2^{(i+1)}+2^{(i-2)}-1$ & Yes & Yes \\[0.6ex]
		2021 & Kulhandjian \textit{et al.} \cite{michelTernary2021}   & $2^i$ & ${i2^{(i-1)}+1}$ & Yes & Yes \\[0.6ex]
		\hline 
	\end{tabular}
	\label{table:ternary}
\end{table*}


For noisy channels, recently in \cite{ming2016}, a class of antipodal code sequences, for overloaded CDM systems with simplified two-stage maximum-likelihood (\acs{ML}) detection, has been proposed. In addition to that, other overloaded matrices over the ternary alphabet are introduced in \cite{Marvasti2012} with a low-complexity
decoding algorithm. Similarly, in \cite{Marvasti2016} the authors propose overloaded code sets over the ternary alphabet that has a twin tree structured cross-correlation 
hierarchy that can be decoded with a simple multi-stage detector. Yet another construction of ternary codes that increases the number of columns, $K$, of UD codes for a such length compared to those proposed in \cite{Marvasti2012} and \cite{Marvasti2016} with a low-complexity polynomial time decoder is proposed in \cite{michel2018}.
The primary reason for such low-complexity decoders is that the code sets are constructed with a certain criteria, which entails lowering the maximum number of users $K < K_{\rm{max}}$, as shown in Table \ref{table:ternary}.

Apart from, binary, antipodal and ternary UD spreading codes higher alphabet $k$-ary spreading codes were studied by Lu \emph{et al.} \cite{LuShan2018} for the multiple-access adder channels.

In this work, for the first time we consider the problem of designing a low-complexity decoder that has a complexity of $\mathcal{O}(LK \:\mathsf{log}_2(K))$ for UD code sets in \cite{michel2012} having the maximum number of users $K_{\rm{max}}^a$. The code sets presented in \cite{michel2012} are also recursive, which make use of a linear map between vector spaces to Galois field extensions. These UD code sets are one possible construction of all possible distinct UD code sets, shown in Table \ref{table:antipodal}. Simulation results in terms of bit-error-rate (\acs{BER}) demonstrate that the proposed decoder has a degradation of only $1-2$ dB in SNR compared to the ML decoder at a BER of $10^{-3}$.

 Our contributions are summarized as follows:
      \vspace{-0.0cm}
      \begin{enumerate}[label=(\arabic*)]
       \item We present the proofs for the important \textit{\bf{Propositions 1-4}} that are presented for the first time in \cite{michel2012}.  Those new prepositions are actually the broken down versions of the unique decodability (UD) property.
       \item  Moreover, for the first time, based on \textit{\bf{Propositions 1-4}}, we formally prove that the maximum number of users for the given $L=8$ is $K_{\rm{max}}^a=13$. 
       \item The minimum Manhattan distance is proved to be $4$ for the recursive UD code sets in \cite{michel2012}.
      \item We develop a low-complexity decoder that has a complexity of $\mathcal{O}(LK\:\mathsf{log}_2(K))$ for our developed UD code sets in \cite{michel2012} having the maximum number of users $K_{\rm{max}}^a$.
       \item We compute the complexity of the deterministic noiseless detection algorithm (\acs{NDA}) presented in \cite{michel2012} and perform some complexity analysis for the proposed fast (low-complexity) detection algorithm (\acs{FDA}).
       
      \end{enumerate}

The rest of the paper is organized as follows. The minimum Manhattan distance of code sets in \cite{michel2012} is presented in Section \ref{minDist}, followed by the proofs of the \textit{\bf{Propositions 1-4}} and the maximum number of $K^a_{max}$ for the case of $L=8$ in Sections \ref{ProofOfProp} and \ref{Proof8_13}, respectively. Detailed discussion of the FDA is presented in Section \ref{fastDecoder}. The complexity analysis for both NDA and FDA algorithms is presented in Section \ref{performanceAnalysis}. After illustrating simulation results in Section \ref{simulation}, a few conclusions are drawn in Section \ref{conclusion}.

The following notations are used in this paper. All boldface lower case letters indicate column vectors and upper case letters indicate matrices, $()^T$ denotes transpose operation, \!\!\!\!\!$\mod$\!\! denotes the modulo operation,  $\mathsf{sgn}$ denotes the sign function, and $|\cdot|$ denotes cardinality of the set.

   \section{Proof of Propositions}
   \label{ProofOfProp}
   In order to facilitate the development of the proof it is beneficial to present the UD code set in \cite{michel2012} which are constructed as $\mathbf{C} = [\mathbf{H}_L \mathbf{V}_L] \in \{ \pm 1\}^{L \times K}$, where $\mathbf{V}_L \in \{\pm 1\}^{L \times (K-L)}$. We recall that the Sylvester-Hadamard matrix of order $2$ is  $\mathbf{H}_2 = \renewcommand{\baselinestretch}{0.7}
	{\normalsize 
		\begin{bmatrix}
		1&\m \m1\\
		1&\m-1
		\end{bmatrix}}$ and of order $2^{p+1}$ for $p = 1,2,...$ is $\mathbf{H}_{2^{p+1}} = \renewcommand{\baselinestretch}{0.7}
	{\normalsize 
		\begin{bmatrix}
		\mathbf{H}_{2^p}&\m \mathbf{H}_{2^p}\\
		\mathbf{H}_{2^p}&\m-\mathbf{H}_{2^p}
		\end{bmatrix}}$. Then, for any $p=1,2,...$, $\mathbf{H}_{2^p}\mathbf{H}_{2^p} = 2^p \mathbf{I}_{2^p \times 2^p}$, where $\mathbf{I}_{N \times N}$ is the $N \times N$ identity matrix. We introduce the notation $\mathbf{H}_4 = \left[\mathbf{h}_0, \mathbf{h}_1,\mathbf{h}_2,\mathbf{h}_3\right]$, $\left[\mathbf{a}_0, \mathbf{a}_1,\mathbf{a}_2,\mathbf{a}_3\right]\triangleq \left[\left[-1,1,1,1\right]^T,\left[1,-1,1,1\right]^T,\left[1,1,-1,1\right]^T,\left[1,1,1,-1\right]^T\right]$, and the negation function $\mathbf{x}^- \triangleq -\mathbf{x}$. We can see that set of vectors $G = \{\mathbf{h}_0,\mathbf{h}_1,\mathbf{h}_2,\mathbf{h}_3, \mathbf{a}_0,\mathbf{a}_1,\mathbf{a}_2,\mathbf{a}_3 \}$ together with $\odot$ operator form a finite group $(G, \odot)$. There exists an isomorphism $\varphi$, shown in Table \ref{isomorphism2} from $G$ to finite additive Abelian group $(\mathbb{F}_{2^4}, +)$ of extended Galois field $\mathbb{F}_{2^4}$, in other words $G$ is isomorphic to $(\mathbb{F}_{2^4}, +)$, $(G, \odot) \cong (\mathbb{F}_{2^4}, +)$. From linear algebra we know that there is an isomorphism from finite additive groups $(\mathbb{F}_{p^n}, +)$ to vector fields $(\mathbb{F}_{p}^n, +)$ and to $\mathbb{Z}_{p}^n$, that is $(\mathbb{F}_{p^n}, +) \cong (\mathbb{F}_{p}^n, +) \cong\mathbb{Z}_{p}^n$, \cite{David1991}.
	
	Table \ref{isomorphism2} shows the mapping of the vectors $\mathbf{h}_0,...,\mathbf{h}_3, \mathbf{a}_0, .., \mathbf{a}_3$ and its negated forms to elements in $\mathbb{F}_{2^4}$ with primitive polynomial $\alpha^4 + \alpha + 1 =0$, where $\alpha$ is the primitive element in extended Galois field GF$(2^4)$. Notice carefully that operation of the finite group $G$ is $\odot$, whereas the finite additive group $\mathbb{F}_{2^4}$ is $+$.
	
	\begin{table}[ht]
		
		\caption{Isomorphism $\varphi:G\mapsto \mathbb{F}_{2^4}$} 
		\centering  
		\begin{tabular}{c c c c} 
			\hline\hline                        
			Antipodal & Polynomial & Power  \\ [0.5ex] 
			\hline                  
			$\mathbf{h}_0$ & 0  & 0   \\ 
			$\mathbf{h}_2$ & $1$  & $1$ \\
			$\mathbf{h}_1$ & $\alpha$  & $\alpha$  \\
			$\mathbf{h}_0^-$ & $\alpha^2$  & $\alpha^2$   \\
			$\mathbf{a}_1$ & $\alpha^3$  & $\alpha^3$  \\
			$\mathbf{h}_3$ & $\alpha + 1$  & $\alpha^4$  \\
			$\mathbf{h}_1^-$ & $\alpha^2+\alpha$  & $\alpha^5$  \\
			$\mathbf{a}_1^-$ & $\alpha^3+\alpha^2$  & $\alpha^6$  \\
			$\mathbf{a}_2$ & $\alpha^3+\alpha +1$  & $\alpha^7$  \\
			$\mathbf{h}_2^-$ & $\alpha^2+1$  & $\alpha^8$  \\
			$\mathbf{a}_3$ & $\alpha^3+\alpha$  & $\alpha^9$  \\
			$\mathbf{h}_3^-$ & $\alpha^2+\alpha+1$  & $\alpha^{10}$  \\
			$\mathbf{a}_3^-$ & $\alpha^3+\alpha^2+\alpha$  & $\alpha^{11}$  \\
			$\mathbf{a}_2^-$ & $\alpha^3+\alpha^2+\alpha+1$  & $\alpha^{12}$  \\
			$\mathbf{a}_0$ & $\alpha^3+\alpha^2+1$  & $\alpha^{13}$  \\
			$\mathbf{a}_0^-$ & $\alpha^3+1$  & $\alpha^{14}$   \\ [1ex]
			\hline 
		\end{tabular}
		\label{isomorphism2} 
	\end{table}
	
	With the above formulation, the columns of $\mathbf{V}_L$, $\mathsf{F}_{\varphi} : \{\pm 1 \}^{L \times 1} \mapsto \mathbb{F}_{2^4}^{L' \times 1} $,  $\mathbf{f}_i = \mathsf{F}_{\varphi}(\mathbf{v}_i)$ maps the $i$-th column $\mathbf{v}_i$ to $\mathbf{f}_i$, for $i = 0,..., K-L-1$, where $\mathbf{f}_i = [f_{i,0} \ f_{i,1}, \dots, \ f_{i,L'}]^T$, $f_{i,j} \in \{0,1,\alpha,...,\alpha^{14} \}$, $0 \geq j \geq L'$, $L' = L/4 = 2^{p-2}$ for $p \geq 2$. As an example for $\mathbf{V}_8$ the mapped vectors are as two-dimensional vectors $\mathbf{f}_i = [f_{i,0} \ f_{i,1}]^T$.
	
   In order to prove \textit{\bf{Proposition 1-4}} that was first presented in \cite{michel2012}, we will expose some very interesting claims. One property of the Hadamard matrix is the following, if we replace $1$'s in Hadamard matrix with $0$'s, and replace the $-1$'s with $1$'s we create $L-1$ Hadamard binary channel codes, since the first column results in zero vector \cite{Todd2005}. Let the Hadamard columns $\mathbf{H}_L = [\mathbf{h}_0,...,\mathbf{h}_{L-1}]$ be mapped into linear binary codes $\mathbf{B}_L = [\mathbf{b}_0,...,\mathbf{b}_{L-1}]$. There is an isomorphism between binary addition and multiplication of $\{\pm 1 \}$ elements. Consequently, since binary addition of Hadamard linear codes is a Hadamard code itself then it is equivalent to element-wise multiplication of any Hadamard codes $\mathbf{H}_L$ is also in $\mathbf{H}_L$. Therefore, some properties can be derived. We denote the mapping of $\mathbf{h}_i$ for $ i=0,...,L-1$ vectors into binary $\mathbf{b}_i \in \{0,1\}^{L\times 1}$ vectors by $\mathbf{b}_i = f(\mathbf{h}_i)$ and inverse $\mathbf{h}_i = f^{-1}(\mathbf{b}_i)$, where $\mathbf{H}_L = [\mathbf{h}_0,...,\mathbf{h}_{L-1}]$ and $\mathbf{B}_L = [\mathbf{b}_0,...,\mathbf{b}_{L-1}]$. Define the $f(\mathbf{x}) = (\boldsymbol{1} - \mathbf{x})/2$ and the inverse $f^{-1}(\mathbf{y}) = [(-1)^{y_0},...,(-1)^{y_{L-1}}]$ functions. Then, the $j$th column corresponds to linear combination of $\mathbf{b}_j = j_1 \mathbf{b}_{1} \oplus j_2 \mathbf{b}_{2} \oplus j_4 \mathbf{b}_{4} \oplus ...\oplus j_{2^p} \mathbf{b}_{2^p}$, where $j$ has the binary representation $j = j_1 + 2j_2 + 4j_4 +...+2^pj_{2^p}$ \cite{Todd2005}. Here $\oplus$ is modular $2$ summation and $j_l \in \{0,1\}$ for $l = \{1,2,4,...,2^p\}$. Therefore, the resulting code is linear code $\mathbf{b}_k = \epsilon_0 \mathbf{b}_0 \oplus \epsilon_1 \mathbf{b}_1 \oplus ... \oplus \epsilon_{L-1} \mathbf{b}_{L-1}$ for all $ \epsilon_{i} \in \{0,1\}$, where $\mathbf{b}_k \in \mathbf{B}_L$.
\begin{table}[H]
	
	\caption{Binary addition and $\{\pm 1\}$ multiplication} 
	\centering  
	\begin{tabular}{c c c c} 
		\hline\hline                        
		Binary & Bipolor  \\ [0.5ex] 
		\hline                  
		0 + 0 = 0 & 1 $\times$ 1 = 1   \\ 
		0 + 1 = 1 & 1 $\times$ -1 = -1   \\
		1 + 0 = 1 & -1 $\times$ 1 = -1   \\
		1 + 1 = 0 & -1 $\times$ -1 = 1   \\ [1ex]
		\hline 
	\end{tabular}
	\label{table1} 
\end{table}
Let us look at the problem and assume that $\mathsf{Null}(\mathbf{C}) \in \mathcal{Z}$, where $\mathcal{Z} \in \{0,\pm 1\}^{K\times 1}$ excluding trivial case $\{0\}^{K\times 1}$ and let $\mathbf{z} \in \mathcal{Z} $. Then the nullspace of $\mathbf{C}$ can be formulated as such
\begin{subequations} \label{null02}
	\begin{eqnarray}
	\mathbf{C} \mathbf{z} &=& \boldsymbol{0} \label{first}\\
	\left [ \mathbf{H}_L \mathbf{V} \right ] \mathbf{z}  &=& \boldsymbol{0}_L \label{second}\\
	\mathbf{H}_L \mathbf{z}_1 + \mathbf{V}_L \mathbf{z}_2    &=& \boldsymbol{0} \label{third}\\
	\mathbf{H}_L \mathbf{z}_1     &=& -\mathbf{V}_L \mathbf{z}_2  \label{fourth}.
	\end{eqnarray}
\end{subequations}
\nid where $\mathbf{z} = \left[ \mathbf{z}_1^T   \mathbf{z}_2^T \right ]^T $, $\mathbf{z}_1 \in \{0,\pm 1 \}^{L\times 1}$, $\mathbf{z}_2 \in \{0,\pm 1 \}^{(K-L) \times 1}$ and $\boldsymbol{0} = \{ 0\}^{L\times 1}$.
If we assume $\mathbf{z}_1 = \boldsymbol{0}$ and for some $\mathbf{z}_2$ (\ref{null02}) is true, then (\ref{null02}) can be expressed in terms of $\mathbf{V}_L$ only as
\begin{eqnarray}\label{dep03}
\mathbf{V}_L \mathbf{z}_2    &=& \boldsymbol{0}. 
\end{eqnarray}

\begin{claim}
	\label{claim01}
	Element-wise multiplication of $\mathbf{H}_L$'s columns $\mathbf{h}_j = \epsilon_0 \mathbf{h}_0 \odot \epsilon_1 \mathbf{h}_1 \odot ... \odot \epsilon_{L-1} \mathbf{h}_{L-1}$  for all possible $0\leq i\leq L-1$, $ \epsilon_{i} \in \{0,1\} $ is true, where $\mathbf{h}_j \in \mathbf{H}_L$.
\end{claim}

\begin{claim}
	\label{claim02}
	If $\mathbf{V}_L$ does not intersect $\boldsymbol{0}$ over the non-trivial $\{0,\pm 1\}$ combinations and (\ref{null02}) is true, then  $\beta_0 \mathbf{v}_0 \odot \beta_1 \mathbf{v}_1 \odot ... \odot \beta_{K-L-1} \mathbf{v}_{K-L-1}= \alpha \mathbf{h}_j $ must be true for some $\beta_i \in \{0, 1\}$, $0\leq i \leq K-L-1$, $\alpha \in \{\pm 1\}$. However, if $\beta_0 \mathbf{v}_0 \odot \beta_1 \mathbf{v}_1 \odot ... \odot \beta_{K-L-1} \mathbf{v}_{K-L-1}= \alpha \mathbf{h}_j $ for some $\beta_i \in \{0, 1\}$, $0\leq i \leq K-L-1$, $\alpha \in \{\pm 1\}$ is true then it is not necessary that (\ref{null02}) is also true.
\end{claim}

\begin{claim}
	\label{claim}
	If (\ref{dep03}) is true, then  $\beta_0 \mathbf{v}_0 \odot \beta_1 \mathbf{v}_1 \odot ... \odot \beta_{K-L-1} \mathbf{v}_{K-L-1} = \pm \boldsymbol{1} $ must be true for some $\beta_i \in \{0, 1\}$, $0\leq i \leq K-L-1$ and $\boldsymbol{1} $ which is all one vector of dimension $L\times 1$. However, if $\beta_0 \mathbf{v}_0 \odot \beta_1 \mathbf{v}_1 \odot ... \odot \beta_{K-L-1} \mathbf{v}_{K-L-1} = \pm \boldsymbol{1} $ for some $\beta_i \in \{0, 1\}$, $0\leq i \leq K-L-1$ is true then it is not necessary that (\ref{dep03}) is also true.
\end{claim}

\begin{claim}
	\label{claim04}
	If $\mathbf{V}_L$ does not satisfy (\ref{null02}) and (\ref{dep03}) then multiplying any column of $\mathbf{V}_L$ by $-1$ will still not satisfy (\ref{null02}) and (\ref{dep03}).
\end{claim}
Authors in \cite{michel2012} showed by exhaustive search that in case when $L=4$ then $K_{\rm{max}}^a = 5$ and all the possible candidates are $[-1, 1, 1, 1]^T$, $[1,-1, 1, 1]^T$, $[1, 1, -1, 1]^T$, $[1, 1, 1, -1]^T$ and their negatives according to the Claim \ref{claim04}. Furthermore, we continue to present additional claims. Let $\mathbf{v}_i \in \{\pm 1\}^4$ for $0\leq i \leq N-1$ then,
\begin{claim}
	\label{claim05}
	If $\beta_0 \mathbf{v}_0 \odot \beta_1 \mathbf{v}_1 \odot ... \odot \beta_{N-1} \mathbf{v}_{N-1} = \alpha \mathbf{h}_j$ and $\alpha \in \{\pm 1\}$, $\mathbf{h}_j \in \mathbf{H}_4$ then $\forall N \in \mathbb{N}$, $[\mathbf{v}_0, ...,\mathbf{v}_{N-1} ] \mathbf{A}_v \boldsymbol{1} = \mathbf{H}_4 \mathbf{A}_H \boldsymbol{1}$ is true for some $\beta_i \in \{0, 1\}$, $0\leq i \leq N-1$, $N' \in \mathbb{N}$ is not necessarily equal to $N$, $\mathbf{A}_v \in \{0,\pm 1\}^{N \times N}$, $\mathbf{A}_H \in \{0,\pm 1\}^{4 \times N'}$, which have maximum one $\pm 1$ entry in each column and $0$'s elsewhere.
\end{claim}

\begin{claim}
	\label{claim06}
	If $0 = N \mod{2}$ and $\beta_0 \mathbf{v}_0 \odot \beta_1 \mathbf{v}_1 \odot ... \odot \beta_{N-1} \mathbf{v}_{N-1} = \pm \boldsymbol{1}$ then $[\mathbf{v}_0, ...,\mathbf{v}_{N-1} ] \mathbf{A}_v \boldsymbol{1} = \boldsymbol{0}$ is true for some $\beta_i \in \{0, 1\}$, $0\leq i \leq N-1$, $\mathbf{A}_v \in \{0,\pm 1\}^{N \times N}$, which have maximum one $\pm 1$ entry in each column and $0$'s elsewhere, except when $N = 4$ and $\mathbf{v}_0 = [-1, 1, 1, 1]^T$, $\mathbf{v}_1 = [1, -1, 1, 1]^T$, $\mathbf{v}_2 = [1, 1, -1, 1]^T$, $\mathbf{v}_3 = [1, 1, 1, -1]^T$.
\end{claim}

The generalization of \textit{\bf{Claims 1-6}} formalizes the \textit{\bf{Propositions 1-4}} and hence we presented the proofs.

\section{Proof $K_{\rm{max}}^a = 13$ for $L=8$ }
\label{Proof8_13}
\vspace{-0.0 cm}
For the case when $L =8$, we prove that the maximum number of columns we can append to $\mathbf{H}_8$ is actually $K^a_{\rm{max}} -L=5$. Note that all proposed UD or another words one-to-one matrix constructions in literature are $\mathbf{C} \subset \mathcal{C}$, where $\mathcal{C}$ is all possible antipodal UD code sets for a given $L$. In order to prove for the maximum number of possible vectors $K_{\rm{max}}^a$, we should look at all possible $\mathbf{V} $ and count how many structure of $\mathbf{V} $ hits any of the forbidden lattice points $\mathbf{H}_m \mathbf{z}_1$ and how many does not. If for a given $k$ column, we count the number of $\mathbf{V}$ that hits any forbidden lattice points is equal to the total number of possible $\mathbf{V}$ vector set, then we know that maximum number of columns of $\mathbf{V}$ that does not hit any forbidden lattice points should be smaller than $k$.

First, we transform antipodal vectors into polynomials with integer coefficients $\mathbb{Z}[x]$, $\mathsf{F}: \{\pm 1\}^{m\times 1} \mapsto \mathbb{Z}[x]$. Those polynomials represent the row location and number of $-1$s or $+1$s in any antipodal $\mathbf{v} \in \{\pm 1\}^{m\times 1}$ with dimension $m$. Let the polynomial be
\begin{equation}
\label{poly01}
G(x) = a_0x^0 + a_1x^1 + a_2 x^2 + ... + a_{m-1} x^{m-1},
\end{equation}
where $a_i \in \mathbb{Z}$. Additions of vector $\mathbf{v}$ in vector space is equivalent to the addition of $\mathbb{Z}[x]$ in polynomial space. Each antipodal vectors, $\mathbf{v}_j \in \{\pm1 \}^{m\times 1}$, where $0 \leq j \leq 2^m-1$, is mapped to the corresponding polynomials $G_n^j(x)$ and $G_p^j(x)$ to represent the $+1$ and $-1$ of the $\mathbf{v}_j$. As an example, for $m=4$, the antipodal vector, $\mathbf{v}_{10} = [1,-1,1,-1]^T$, is mapped to $G_n^{10}(x) = x^1 + x^3$ and $G_p^{10}(x) = x^0 + x^2$ polynomials. Observe that for any antipodal vector $
\mathbf{v}_j$, the addition of polynomials, $G_n^j(x) + G_p^j(x) = x^0 + x^1 + ... + x^{m-1}$.
In order to visualize polynomial additions in $2$-dimensional Euclidean space we can further transform $\mathbb{Z}[x]$ vector space into $\Lambda  \subseteq \mathbb{Z}^2$ integer lattice, $\mathsf{H}: \mathbb{Z}[x] \mapsto \mathbb{Z}^2$. Let us define functions $\sigma_n^j(m) = G_n^j(m)$, $\sigma_p^j(m) = G_p^j(m)$, which are the evaluations of polynomials $G_n^j(x)$ and $G_p^j(x)$ at $m$, where $m$ is the dimension of vector $\mathbf{v}_j$. 
By setting the x-axis and y-axis to be $G_n^j(x)$ and $G_p^j(x)$, we can build $\Lambda  \subseteq \mathbb{Z}^2$ lattice points, since evaluations $\sigma_n^j(m)$ and $\sigma_p^j(m)$ for each antipodal $\mathbf{v}_j$ vectors are integers. Taking the above example of antipodal vectors having dimension of $m=4$ the equivalent integer lattice points is shown in Table \ref{Lattice01}.
\begin{center}
	\begin{figure}[H]
	\hspace{-0.6cm}
		\centering
		\includegraphics[width=4.0 in]{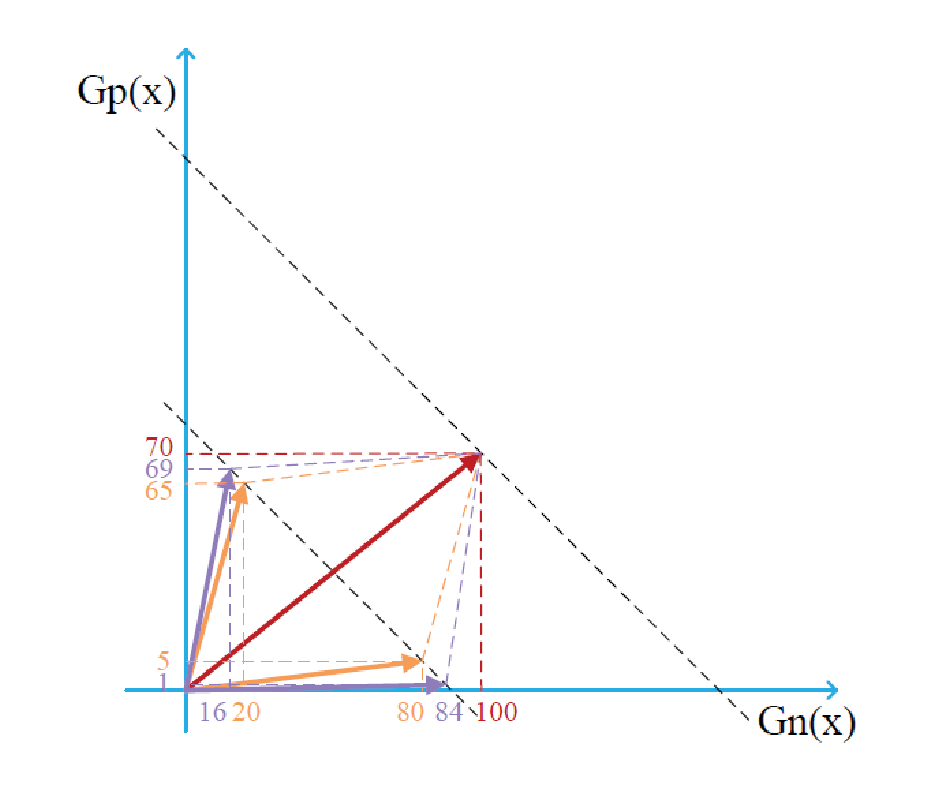}
		\centering \caption{Addition of lattice points in $\Lambda  \subseteq \mathbb{Z}^2$, $m = 4$.}\label{2DLattice}
	\end{figure}
\end{center}
\vspace{-0.4cm}
Let us define \textit{sublattices} $\Lambda_H  \subseteq \Lambda $ and $\Lambda_{V'} \subseteq \Lambda$ constructed by vectors of $\left[\mathbf{H}_4 -\mathbf{H}_4 \right]$ and $\mathbf{V}'$, where $\mathbf{V}' \in \{\pm 1\}^{4 \times 8}$, $\mathbf{v}'_i \notin \left[\mathbf{H}_4 -\mathbf{H}_4 \right] $ for $1\leq i \leq 8$. As a reminder, for the antipodal UD code set the goal is to construct $\mathbf{V}'$ with the maximum number of columns $K_{\rm{max}}^a - L$ such that $\Lambda_H \cap \Lambda_{V'} = \emptyset$.

\begin{table*}[h]
	\caption{Lattice points for all antipodal vectors $\mathbf{v}$ when $m=4$} 
	\centering  
	\begin{tabular}{ccccccccccccccccc} 
		\hline\hline                        
		$j$  & 0 & 1 & 2 & 3 & 4 & 5 & 6 & 7 & 8 & 9 & 10& 11& 12& 13& 14& 15 \\ [0.5ex] 
		\hline                 
		$G_n^j(m)$ & 0 & 1 & 4 & 5 & 16& 17& 20& 21& 64& 65& 68& 69& 80& 81& 84& 85  \\ 
		$G_p^j(m)$ & 85& 84& 81& 80& 69& 68& 65& 64& 21& 20& 17& 16& 5 & 4 & 1 & 0   \\
		\hline 
	\end{tabular}
	\label{Lattice01} 
\end{table*}
Fig. \ref{2DLattice} demonstrates a hit, in other words $\Lambda_H \cap \Lambda_{V'} \neq \emptyset$ for $k=2$ and $\mathbf{V}' =[\mathbf{v}_{4}, \mathbf{v}_{14} ] = \left[\left[1,1,-1,1\right]^T \left[1,-1,-1,-1\right]^T \right]$, which is equivalent to $\left[16,69\right]^T$ and $ \left[84,1\right]^T $ in $\Lambda_{V'}$. If we construct the third and the fourth columns to be  $\mathbf{v}_{6} = \left[1,-1,-1,1\right]^T$ and $\mathbf{v}_{12} = \left[1,1,-1,-1\right]^T$, equivalently $\left[20,65\right]^T$ and $ \left[80,5\right]^T $ in $\Lambda_{H}$ of $\mathbf{H}_4$ adding them we get $\left[2,0,-2,0\right]^T$, which is shown as a lattice point $\left[100,70\right]^T$ in Fig. \ref{2DLattice}.
Clearly, addition of $\left[16,69\right]^T$ and $ \left[84,1\right]^T $ lattice points also results $\left[100,70\right]^T$, hence, this is why intersection of two \textit{sublattices} is not an empty set. Therefore, it can be shown that all possible $\binom{8}{2}$ combination of two columns of $\mathbf{V}'$ will at least hit a lattice points of $\Lambda_H$, which means that $\Lambda_H \cap \Lambda_{V'} \neq \emptyset$ for $k=2.$ Hence, $|\mathbf{V}' |= 1$ for the case of $m=4$. Using this approach, we can construct \textit{sublattices} $\Lambda_H$ and $\Lambda_{V'}$ for the case of $m=8$ and find the upper bound of $k$ such that $\Lambda_H \cap \Lambda_{V'} = \emptyset$. This creates a framework to search the solution using geometric combinatorics (e.g. Minkowski sum, Minkowski geometry of numbers \cite{Olds2000}), partition and decomposition of $\Lambda$ into equivalence classes, formed by \textit{sublattices} and its' cosets to prove the $K_{\rm{max}}^a$ we can have for a given $m$ and how to generate those vectors in $\mathbf{V}'$.

Polynomials of $G_n^j(x)$ and $G_p^j(x)$ can also be represented by their exponents in the above example $G^{10}_n(x) \to E^{10}_n = \{1, 3\}$ and $G^{10}_p(x) \to E^{10}_p = \{0, 2\}$. Therefore, our problem of avoiding non-trivial combinations of $\mathbf{V}$ to hit any forbidden lattice points becomes
\begin{equation}
\label{poly02}
\biguplus_{\mathbf{h}_t \in \mathbf{H}_m} E_i^t  \neq \biguplus_{\mathbf{v}_t \in \mathbf{V}} E_j^t  \; \; i,j \in \{p,n\},
\end{equation}
\noindent where $\biguplus$ denotes the multiset union,  element-wise sums of each multisets and $1\geq t \geq 2^m-1$ integer value corresponds to antipodal vector. The objective is to avoid $\mathbf{V}$ $\biguplus_{\mathbf{v}_t \in \mathbf{V}} E_j^t$ sums of multisets $E_p$ or $E_n$ of $\mathbf{V}$ to hit any forbidden  $\biguplus_{\mathbf{h}_t \in \mathbf{H}_m} E_i^t $ sum of multisets $E_p$ or $E_n$ of $\mathbf{H}_m$. Obviously, there are $3^m$ forbidden multisets of $\mathbf{H}_m$  and we do not want $3^k-1$ multisets of $\mathbf{V}$ excluding trivial case hit any of the forbidden multisets. We can use multiset partition theories and study bipolar vectors by their $E_p$ and $E_n$ representations and prove the maximal number of $k$.

Note that in our matrix construction design of the $\mathbf{v}_i$'s that they are distinct and not equal to any of columns of $\mathbf{H}_m$ or $-\mathbf{H}_m$. If, however, any of $\mathbf{v}_i \in \pm \mathbf{H}_m$, then the multiset of $\mathbf{v}_i$ hits the forbidden multiset of $\mathbf{H}_m$. Such $\mathbf{v}_i$ vectors can never be included in vector set of UD codes. Additionally, the $\mathbf{v}_i$s can be replaced by  $-\mathbf{v}$ without violating the uniquly decodability property (\ref{null02}). Since all possible combinations of vectors including $-\mathbf{v}$ of $\mathbf{V}$ do not hit any of the forbidden multisets. In other words, if $[ \mathbf{H}_m \mathbf{V} ] \in \mathcal{C}$ so is any number of columns of $\mathbf{V}$ or $\mathbf{H}_m$ multiplied by $-1$.

One way to approach this problem is to classify all $\mathbf{v}$ bipolar vectors into groups then use inclusion-exclusion principle. There are $2^m$ number of $\mathbf{v}$'s that consist of $|E_n| = i$, where $i = \{0,1,...,m\}$ with all the combinations makes $\sum_{i=0}^m \binom{m}{i}= 2^m$.

Let us divide the total number of $m$-dimensional bipolar vectors into two classes; $\mathbf{v} \in \mathcal{B}_m^+$ if $0\leq | E_n| < m/2 $, $|E_n| = m/2 | E_p = \{1, 2, ..\}$ excluding $m$ $\mathbf{h}_i$'s columns from $\mathbf{H}_m$ and $\mathbf{v} \in \mathcal{B}_m^-$ if $m/2 < | E_n|\leq m $, $|E_n| = m/2 | E_n = \{1, 2,..\}$ excluding $m$ $-\mathbf{h}_i$'s columns from $-\mathbf{H}_m$. It's clear that if any $\mathbf{v} \in \mathcal{B}_m^+$ then $-\mathbf{v} \in \mathcal{B}_m^-$.  This narrows for our design to consider only distinct vector sets and the total number of such distinct $\mathbf{v}_i$'s to be considered in our $\mathbf{V}$ design is $\sum_{i=0}^{m/2} \binom{m}{i} - m = 2^{m-1} - m$. We need to construct from distinct vectors $\mathbf{v} \in \mathcal{B}_m^+$ such that they do not hit any forbidden multisets.

Therefore, the total possible number of $\mathbf{V}$ sets with $k$ columns is  $\binom{2^{m-1} - m}{k}$. Out of this total number of $\mathbf{V}$ sets with $k$ columns only some satisfy UD or one-to-one condition when appended to $\mathbf{H}_m$. If all the possible $\mathbf{V}$'s do not satisfy one-to-one condition that means $k$ exceeds the maximum number of columns that can be added to $\mathbf{H}_m$. Hence, we want to count how many combinations of $\mathbf{v} \in \mathcal{B}_m^+$ with $k$ columns hit the forbidden multisets. If the number of combination is equal the total number of $\mathbf{V}$ sets then we know that $k$ is not the maximum. Counting that number will help us to prove the maximum number of columns $k$.

We classify $\mathcal{B}_m^+$ into different groups, so that any combinations of vectors in similar group hits the forbidden multisets and therefore, such $\mathbf{v}$ that belong to the same group must be avoided in our design.

In our example, for $m=8$ and $k=2$, we classify $\mathcal{B}_8^+$ into groups and count how many $\mathbf{V}$'s do not satisfy the UD condition or one-to-one condition out of $\binom{2^{8-1} - 8}{2} = \binom{120}{2}=7140$ possible vector sets. By looking at all possible combinations of $\mathbf{V} =  [\mathbf{v}_{j_1} \mathbf{v}_{j_2} ]$ , $|E_n^{j_1}|, |E_n^{j_2}| \in \{1,...,8\}$ we only consider $(|E_n^{j_1}|, |E_n^{j_2}|) = \{ (1,3), (2,2), (3,3), (4,4)\}$. Since to hit $[\mathbf{h}_1 \pm \mathbf{h}_i]$, where $2\leq i \leq 8$, the vectors $[\mathbf{v}_{j_1} \mathbf{v}_{j_2} ]$ must have $(|E_n^{j_1}|, |E_n^{j_2}|) = (1,3)$ and $(|E_n^{j_1}|, |E_n^{j_2}|) = (2,2)$. We can easily show that if $[\mathbf{v}_{j_1} \mathbf{v}_{j_2} ]$ hit forbidden multisets with $(|E_n^{j_1}|, |E_n^{j_2}|) = (1,3)$ and $(|E_n^{j_1}|, |E_n^{j_2}|) = (2,2)$ then $[-\mathbf{v}_{j_1} -\mathbf{v}_{j_2} ]$ with $(|E_n^{j_1}|, |E_n^{j_2}|) = (7,5)$ and $(|E_n^{j_1}|, |E_n^{j_2}|) = (6,6)$ hit $[-\mathbf{h}_1 \mp \mathbf{h}_l]$ and $l \neq i$. Also, if $[\mathbf{v}_{j_1} \mathbf{v}_{j_2} ]$ hit forbidden multisets with $(|E_n^{j_1}|, |E_n^{j_2}|) = (2,2)$ then $[-\mathbf{v}_{j_1} \mathbf{v}_{j_2} ]$ and $[\mathbf{v}_{j_1} -\mathbf{v}_{j_2} ]$ with $(|E_n^{j_1}|, |E_n^{j_2}|) = (2,6)$ hit $[\pm \mathbf{h}_l \pm \mathbf{h}_p]$ and $i\neq l$, $p\neq i$. Hitting $[\pm \mathbf{h}_i \pm \mathbf{h}_j]$, where $2\leq i \neq j \leq 8$ multisets then vectors must have $(|E_n^{j_1}|, |E_n^{j_2}|) = (3,5)$ and $(|E_n^{j_1}|, |E_n^{j_2}|) = (4,4)$. Since we only work with distinct vectors $\mathcal{B}_8^+$ then $(|E_n^{j_1}|, |E_n^{j_2}|) = (3,5)$ that has $[\mathbf{v}_{j_1} -\mathbf{v}_{j_2}] $ is equivalent to $[\mathbf{v}_{j_1} \mathbf{v}_{j_2}] $ with $(|E_n^{j_1}|, |E_n^{j_2}|) = (3,3)$.

Now, let us classify $(|E_n^{j_1}|, |E_n^{j_2}|) = (1,3)$,  $(|E_n^{j_1}|, |E_n^{j_2}|) = (2,2)$, $(|E_n^{j_1}|, |E_n^{j_2}|) = (3,3)$, $(|E_n^{j_1}|, |E_n^{j_2}|) = (4,4)$ into $\mathcal{A}_i$'s, $\mathcal{D}_i$'s, $\mathcal{F}_i$'s and $\mathcal{G}_i$'s groups. We know that the total number of $\mathbf{v}$s with $|E_n| = 1$ is $\binom{8}{1} =8$, $|E_n| = 2$ is $\binom{8}{2} =28$,$|E_n| = 3$ is $\binom{8}{3} =56$ and $|E_n| = 4$ is $\binom{8}{4} =70$.

Here is how we divide $56$ of $|E_n| = 3$ and $8$ of $|E_n| = 1$ into $\mathcal{A}_i$'s,  $1\leq i \leq 8$ groups.

\vspace{0.3cm}
$\mathcal{A}_1 = 
\renewcommand{\baselinestretch}{0.9}
{\footnotesize \setcounter{MaxMatrixCols}{34}
	\begin{bmatrix} -1  \\ \phantom{-}1   \\ \phantom{-}1   \\  \phantom{-}1 \\ \phantom{-}1 \\ \phantom{-}1  \\ \phantom{-}1  \\ \phantom{-}1   \end{bmatrix} \begin{bmatrix} \phantom{-}1 & \phantom{-}1 & \phantom{-}1 & \phantom{-}1 & \phantom{-}1 & \phantom{-}1 & \phantom{-}1\\ -1 & -1 & -1 & \phantom{-}1 & \phantom{-}1 & \phantom{-}1 & \phantom{-}1  \\ -1 & \phantom{-}1 & \phantom{-}1 & -1 & -1 & \phantom{-}1 & \phantom{-}1   \\  -1 & \phantom{-}1 & \phantom{-}1 & \phantom{-}1 & \phantom{-}1 & -1 & -1 \\ \phantom{-}1 & -1 & \phantom{-}1 & -1 & \phantom{-}1 & -1 & \phantom{-}1 \\ \phantom{-}1 & -1 & \phantom{-}1 & \phantom{-}1 & -1 & \phantom{-}1 & -1  \\ \phantom{-}1 & \phantom{-}1 & -1 & -1 &  \phantom{-}1 &  \phantom{-}1 & -1 \\ \phantom{-}1 & \phantom{-}1 & -1 & \phantom{-}1 & -1 & -1 &  \phantom{-}1  \end{bmatrix}} $,\\
$\mathcal{A}_2 =  
\renewcommand{\baselinestretch}{0.9}
{\footnotesize \setcounter{MaxMatrixCols}{34}
	\begin{bmatrix} \phantom{-}1  \\ -1   \\ \phantom{-}1   \\  \phantom{-}1 \\ \phantom{-}1 \\ \phantom{-}1  \\ \phantom{-}1  \\ \phantom{-}1   \end{bmatrix} \begin{bmatrix} -1 & -1 & -1 & \phantom{-}1 & \phantom{-}1 & \phantom{-}1 & \phantom{-}1\\ \phantom{-}1 & \phantom{-}1 & \phantom{-}1 & \phantom{-}1 & \phantom{-}1 & \phantom{-}1 & \phantom{-}1  \\ -1 & \phantom{-}1 & \phantom{-}1 & -1 & -1 & \phantom{-}1 & \phantom{-}1   \\  -1 & \phantom{-}1 & \phantom{-}1 & \phantom{-}1 & \phantom{-}1 & -1 & -1 \\ \phantom{-}1 & -1 & \phantom{-}1 & -1 & \phantom{-}1 & -1 & \phantom{-}1 \\ \phantom{-}1 & -1 & \phantom{-}1 & \phantom{-}1 & -1 & \phantom{-}1 & -1  \\ \phantom{-}1 & \phantom{-}1 & -1 & \phantom{-}1 &  -1 & -1 & \phantom{-}1 \\ \phantom{-}1 & \phantom{-}1 & -1 & -1 & \phantom{-}1 & \phantom{-}1 &  -1  \end{bmatrix}} $,

$\mathcal{A}_3 = 
\renewcommand{\baselinestretch}{0.9}
{\footnotesize \setcounter{MaxMatrixCols}{34}
	\begin{bmatrix} \phantom{-}1  \\ \phantom{-}1   \\ -1   \\  \phantom{-}1 \\ \phantom{-}1 \\ \phantom{-}1  \\ \phantom{-}1  \\ \phantom{-}1   \end{bmatrix} \begin{bmatrix} -1 & -1 & -1 & \phantom{-}1 & \phantom{-}1 & \phantom{-}1 & \phantom{-}1\\ -1 & \phantom{-}1 & \phantom{-}1 & -1 & -1 & \phantom{-}1 & \phantom{-}1  \\ \phantom{-}1 & \phantom{-}1 & \phantom{-}1 & \phantom{-}1 & \phantom{-}1 & \phantom{-}1 & \phantom{-}1   \\  -1 & \phantom{-}1 & \phantom{-}1 & \phantom{-}1 & \phantom{-}1 & -1 & -1 \\ \phantom{-}1 & -1 & \phantom{-}1 & -1 & \phantom{-}1 & -1 & \phantom{-}1 \\ \phantom{-}1 & \phantom{-}1 & -1 & \phantom{-}1 & -1 & -1 & \phantom{-}1  \\ \phantom{-}1 & -1 & \phantom{-}1 & \phantom{-}1 &  -1 & \phantom{-}1 & -1 \\ \phantom{-}1 & \phantom{-}1 & -1 & -1 & \phantom{-}1 & \phantom{-}1 &  -1  \end{bmatrix}}$,\\
$\mathcal{A}_4 =
\renewcommand{\baselinestretch}{0.9}
{\footnotesize \setcounter{MaxMatrixCols}{34}
	\begin{bmatrix} \phantom{-}1  \\ \phantom{-}1   \\  \phantom{-}1  \\  -1 \\ \phantom{-}1 \\ \phantom{-}1  \\ \phantom{-}1  \\ \phantom{-}1   \end{bmatrix} \begin{bmatrix} -1 & -1 & -1 & \phantom{-}1 & \phantom{-}1 & \phantom{-}1 & \phantom{-}1\\ -1 & \phantom{-}1 & \phantom{-}1 & -1 & -1 & \phantom{-}1 & \phantom{-}1  \\ -1 & \phantom{-}1 & \phantom{-}1 & \phantom{-}1 & \phantom{-}1 & -1 & -1   \\  \phantom{-}1 & \phantom{-}1 & \phantom{-}1 & \phantom{-}1 & \phantom{-}1 & \phantom{-}1 & \phantom{-}1 \\ \phantom{-}1 & -1 & \phantom{-}1 & -1 & \phantom{-}1 & -1 & \phantom{-}1 \\ \phantom{-}1 & \phantom{-}1 & -1 & \phantom{-}1 & -1 & -1 & \phantom{-}1  \\ \phantom{-}1 & \phantom{-}1 & -1 & -1 &  \phantom{-}1 & \phantom{-}1 & -1 \\ \phantom{-}1 & -1 & \phantom{-}1 & \phantom{-}1 & -1 & \phantom{-}1 &  -1  \end{bmatrix}}$,

$\mathcal{A}_5 =
\renewcommand{\baselinestretch}{0.9}
{\footnotesize \setcounter{MaxMatrixCols}{34}
	\begin{bmatrix} \phantom{-}1  \\ \phantom{-}1   \\  \phantom{-}1  \\  \phantom{-}1 \\ -1 \\ \phantom{-}1  \\ \phantom{-}1  \\ \phantom{-}1   \end{bmatrix} \begin{bmatrix} -1 & -1 & -1 & \phantom{-}1 & \phantom{-}1 & \phantom{-}1 & \phantom{-}1\\ -1 & \phantom{-}1 & \phantom{-}1 & -1 & -1 & \phantom{-}1 & \phantom{-}1  \\ \phantom{-}1 & -1 & \phantom{-}1 & -1 & \phantom{-}1 & -1 & \phantom{-}1   \\  \phantom{-}1 & \phantom{-}1 & -1 & \phantom{-}1 & -1 & -1 & \phantom{-}1 \\ \phantom{-}1 & \phantom{-}1 & \phantom{-}1 & \phantom{-}1 & \phantom{-}1 & \phantom{-}1 & \phantom{-}1 \\ -1 & \phantom{-}1 & \phantom{-}1 & \phantom{-}1 & \phantom{-}1 & -1 & -1  \\ \phantom{-}1 & -1 & \phantom{-}1 & \phantom{-}1 &  -1 & \phantom{-}1 & -1 \\ \phantom{-}1 & \phantom{-}1 & -1 & -1 & \phantom{-}1 & \phantom{-}1 &  -1  \end{bmatrix}}$,
	
$\mathcal{A}_6 =  
\renewcommand{\baselinestretch}{0.9}
{\footnotesize \setcounter{MaxMatrixCols}{34}
	\begin{bmatrix} \phantom{-}1  \\ \phantom{-}1   \\  \phantom{-}1  \\  \phantom{-}1 \\ \phantom{-}1 \\ -1  \\ \phantom{-}1  \\ \phantom{-}1   \end{bmatrix} \begin{bmatrix} -1 & -1 & -1 & \phantom{-}1 & \phantom{-}1 & \phantom{-}1 & \phantom{-}1\\ -1 & \phantom{-}1 & \phantom{-}1 & -1 & -1 & \phantom{-}1 & \phantom{-}1  \\ \phantom{-}1 & -1 & \phantom{-}1 & -1 & \phantom{-}1 & -1 & \phantom{-}1   \\  \phantom{-}1 & \phantom{-}1 & -1 & \phantom{-}1 & -1 & -1 & \phantom{-}1 \\ -1 & \phantom{-}1 & \phantom{-}1 & \phantom{-}1 & \phantom{-}1 & -1 & -1 \\ \phantom{-}1& \phantom{-}1 & \phantom{-}1 & \phantom{-}1 & \phantom{-}1 & \phantom{-}1 & \phantom{-}1  \\ \phantom{-}1 & \phantom{-}1 & -1 & -1 &  \phantom{-}1 & \phantom{-}1 & -1 \\ \phantom{-}1 & -1 & \phantom{-}1 & \phantom{-}1 & -1 & \phantom{-}1 &  -1  \end{bmatrix}}$,
	
$\mathcal{A}_7 =  
\renewcommand{\baselinestretch}{0.9}
{\footnotesize \setcounter{MaxMatrixCols}{34}
	\begin{bmatrix} \phantom{-}1  \\ \phantom{-}1   \\  \phantom{-}1  \\  \phantom{-}1 \\ \phantom{-}1 \\ \phantom{-}1  \\ -1  \\ \phantom{-}1   \end{bmatrix} \begin{bmatrix} -1 & -1 & -1 & \phantom{-}1 & \phantom{-}1 & \phantom{-}1 & \phantom{-}1\\ -1 & \phantom{-}1 & \phantom{-}1 & -1 & -1 & \phantom{-}1 & \phantom{-}1  \\ \phantom{-}1 & -1 & \phantom{-}1 & -1 & \phantom{-}1 & -1 & \phantom{-}1   \\  \phantom{-}1 & \phantom{-}1 & -1 & \phantom{-}1 & -1 & -1 & \phantom{-}1 \\ \phantom{-}1 & -1 & \phantom{-}1 & \phantom{-}1 & -1 & \phantom{-}1 & -1 \\ \phantom{-}1& \phantom{-}1 & -1 & -1 & \phantom{-}1 & \phantom{-}1 & -1  \\ \phantom{-}1 & \phantom{-}1 & \phantom{-}1 & \phantom{-}1 &  \phantom{-}1 & \phantom{-}1 & \phantom{-}1 \\ -1 & \phantom{-}1 & \phantom{-}1 & \phantom{-}1 & \phantom{-}1 & -1 &  -1  \end{bmatrix}}$,
	
$\mathcal{A}_8 = 
\renewcommand{\baselinestretch}{0.9}
{\footnotesize \setcounter{MaxMatrixCols}{34}
	\begin{bmatrix} \phantom{-}1  \\ \phantom{-}1   \\  \phantom{-}1  \\  \phantom{-}1 \\ \phantom{-}1 \\ \phantom{-}1  \\ \phantom{-}1   \\ -1   \end{bmatrix} \begin{bmatrix} -1 & -1 & -1 & \phantom{-}1 & \phantom{-}1 & \phantom{-}1 & \phantom{-}1\\ -1 & \phantom{-}1 & \phantom{-}1 & -1 & -1 & \phantom{-}1 & \phantom{-}1  \\ \phantom{-}1 & -1 & \phantom{-}1 & -1 & \phantom{-}1 & -1 & \phantom{-}1   \\  \phantom{-}1 & \phantom{-}1 & -1 & \phantom{-}1 & -1 & -1 & \phantom{-}1 \\ \phantom{-}1 & \phantom{-}1  & -1 & -1 & \phantom{-}1  & \phantom{-}1 & -1 \\ \phantom{-}1& -1 & \phantom{-}1  & \phantom{-}1  & -1 & \phantom{-}1 & -1  \\ -1 & \phantom{-}1 & \phantom{-}1 & \phantom{-}1 &  \phantom{-}1 & -1 & -1 \\ \phantom{-}1  & \phantom{-}1 & \phantom{-}1 & \phantom{-}1 & \phantom{-}1 & \phantom{-}1  &  \phantom{-}1   \end{bmatrix}} $\\

\nid equivalently, we can then write them in multiset form as
\[ \mathcal{A}_1 = 
\renewcommand{\baselinestretch}{0.9}
{\footnotesize \setcounter{MaxMatrixCols}{34}
	\begin{array}{lcr}
	E_p= \{1, 2, 3, 4, 5, 6, 7\} & E_n=\{ 0\}  \\
	E_p=\{ 0,  4, 5, 6, 7\} & E_n=\{ 1,2, 3\} \\
	E_p=\{ 0,  2, 3, 6, 7\} & E_n=\{ 1,4, 5\} \\
	E_p=\{ 0,  2, 3, 4, 5\} & E_n=\{ 1,6, 7\} \\
	E_p=\{ 0,  1, 3, 5, 7\} & E_n=\{ 2,4, 6\} \\
	E_p=\{ 0,  1, 3, 4, 6\} & E_n=\{ 2,5, 7\} \\
	E_p=\{ 0,  1, 2, 5, 6\} & E_n=\{ 3,4, 7\} \\
	E_p=\{ 0,  1, 2, 4, 7\} & E_n=\{ 3,5, 6\}  \end{array}}\]
	
	\[\mathcal{A}_2 = 
    \renewcommand{\baselinestretch}{0.9}
    {\footnotesize \setcounter{MaxMatrixCols}{34}	
	\begin{array}{lcr}
	E_p= \{0, 2, 3, 4, 5, 6, 7\} & E_n=\{ 1\}  \\
	E_p=\{ 1,  4, 5, 6, 7\} & E_n=\{ 0,2, 3\} \\
	E_p=\{ 1,  2, 3, 6, 7\} & E_n=\{ 0,4, 5\} \\
	E_p=\{ 1,  2, 3, 4, 5\} & E_n=\{ 0,6, 7\} \\
	E_p=\{ 0,  1, 3, 5, 6\} & E_n=\{ 2,4, 7\} \\
	E_p=\{ 0,  1, 3, 4, 7\} & E_n=\{ 2,5, 6\} \\
	E_p=\{ 0,  1, 2, 5, 7\} & E_n=\{ 3,4, 6\} \\
	E_p=\{ 0,  1, 2, 4, 6\} & E_n=\{ 3,5, 7\}  \end{array}}\]

\[ \mathcal{A}_3 = 
\renewcommand{\baselinestretch}{0.9}
{\footnotesize \setcounter{MaxMatrixCols}{34}
	\begin{array}{lcr}
	E_p= \{0,1,  3, 4, 5, 6, 7\} & E_n=\{ 2\}  \\
	E_p=\{ 2,  4, 5, 6, 7\} & E_n=\{ 0,1, 3\} \\
	E_p=\{ 1,  2, 3, 5, 7\} & E_n=\{ 0,4, 6\} \\
	E_p=\{ 1,  2, 3, 4, 6\} & E_n=\{ 0,5, 7\} \\
	E_p=\{ 0,  2, 3, 5, 6\} & E_n=\{ 1,4, 7\} \\
	E_p=\{ 0,  2, 3, 4, 7\} & E_n=\{ 1,5, 6\} \\
	E_p=\{ 0,  1, 2, 6, 7\} & E_n=\{ 3,4, 5\} \\
	E_p=\{ 0,  1, 2, 4, 5\} & E_n=\{ 3,6, 7\}  \end{array}}\]
	
\[	\mathcal{A}_4 = 
\renewcommand{\baselinestretch}{0.9}
{\footnotesize \setcounter{MaxMatrixCols}{34}
    \begin{array}{lcr}
	E_p= \{0, 1, 2, 4, 5, 6, 7\} & E_n=\{ 3\}  \\
	E_p=\{ 1,  4, 5, 6, 7\} & E_n=\{ 0,1, 2\} \\
	E_p=\{ 1,  2, 3, 6, 7\} & E_n=\{ 0,4, 7\} \\
	E_p=\{ 1,  2, 3, 4, 5\} & E_n=\{ 0,5, 6\} \\
	E_p=\{ 0,  1, 3, 5, 6\} & E_n=\{ 1,4, 6\} \\
	E_p=\{ 0,  1, 3, 4, 7\} & E_n=\{ 1,5, 7\} \\
	E_p=\{ 0,  1, 2, 5, 7\} & E_n=\{ 2,4, 5\} \\
	E_p=\{ 0,  1, 2, 4, 6\} & E_n=\{ 2,6, 7\}  \end{array}}\]

\[ \mathcal{A}_5 = 
\renewcommand{\baselinestretch}{0.9}
{\footnotesize \setcounter{MaxMatrixCols}{34}
	\begin{array}{lcr}
	E_p= \{0,1,  2, 3, 5, 6, 7\} & E_n=\{ 4\}  \\
	E_p=\{ 2,  3, 4, 6, 7\} & E_n=\{ 0,1, 5\} \\
	E_p=\{ 1,  3, 4, 5, 7\} & E_n=\{ 0,2, 6\} \\
	E_p=\{ 1,  2, 4, 5, 6\} & E_n=\{ 0,3, 7\} \\
	E_p=\{ 0,  3, 4, 5, 6\} & E_n=\{ 1,2, 7\} \\
	E_p=\{ 0,  2, 4, 5, 7\} & E_n=\{ 1,3, 6\} \\
	E_p=\{ 0,  1, 4, 6, 7\} & E_n=\{ 2,3, 5\} \\
	E_p=\{ 0,  1, 2, 3, 4\} & E_n=\{ 5,6, 7\}  \end{array}}\]
	
\[	\mathcal{A}_6 =
	\renewcommand{\baselinestretch}{0.9}
{\footnotesize \setcounter{MaxMatrixCols}{34}
	\begin{array}{lcr}
	E_p= \{0, 1, 2, 3, 4, 6, 7\} & E_n=\{ 5\}  \\
	E_p=\{ 2,  3, 5, 6, 7\} & E_n=\{ 0,1, 4\} \\
	E_p=\{ 1,  3, 4, 5, 6\} & E_n=\{ 0,2, 7\} \\
	E_p=\{ 1,  2, 4, 5, 7\} & E_n=\{ 0,3, 6\} \\
	E_p=\{ 0,  3, 4, 5, 7\} & E_n=\{ 1,2, 6\} \\
	E_p=\{ 0,  2, 4, 5, 6\} & E_n=\{ 1,3, 7\} \\
	E_p=\{ 0,  1, 5, 6, 7\} & E_n=\{ 2,3, 4\} \\
	E_p=\{ 0,  1, 2, 3, 5\} & E_n=\{ 4,6, 7\}  \end{array}}\]

\[ \mathcal{A}_7 = 
\renewcommand{\baselinestretch}{0.9}
{\footnotesize \setcounter{MaxMatrixCols}{34}
	\begin{array}{lcr}
	E_p= \{0,1,  2, 3, 4, 5, 7\} & E_n=\{ 6\}  \\
	E_p=\{ 2,  3, 4, 5, 6\} & E_n=\{ 0,1, 7\} \\
	E_p=\{ 1,  3, 5, 6, 7\} & E_n=\{ 0,2, 4\} \\
	E_p=\{ 1,  2, 4, 6, 7\} & E_n=\{ 0,3, 5\} \\
	E_p=\{ 0,  3, 4, 6, 7\} & E_n=\{ 1,2, 5\} \\
	E_p=\{ 0,  2, 5, 6, 7\} & E_n=\{ 1,3, 4\} \\
	E_p=\{ 0,  1, 4, 5, 6\} & E_n=\{ 2,3, 7\} \\
	E_p=\{ 0,  1, 2, 3, 6\} & E_n=\{ 4,5, 7\}  \end{array}}\]
	
\[	\mathcal{A}_8 =
\renewcommand{\baselinestretch}{0.9}
{\footnotesize \setcounter{MaxMatrixCols}{34}
\begin{array}{lcr}
	E_p= \{0, 1, 2, 3, 4, 6, 7\} & E_n=\{ 7\}  \\
	E_p=\{ 2,  3, 4, 5, 7\} & E_n=\{ 0,1, 6\} \\
	E_p=\{ 1,  3, 4, 6, 7\} & E_n=\{ 0,2, 5\} \\
	E_p=\{ 1,  2, 5, 6, 7\} & E_n=\{ 0,3, 4\} \\
	E_p=\{ 0,  3, 5, 6, 7\} & E_n=\{ 1,2, 4\} \\
	E_p=\{ 0,  2, 4, 6, 7\} & E_n=\{ 1,3, 5\} \\
	E_p=\{ 0,  1, 4, 5, 7\} & E_n=\{ 2,3, 6\} \\
	E_p=\{ 0,  1, 2, 3, 7\} & E_n=\{ 4,5, 6\}  \end{array}}\]

We can prove that
\begin{equation}
\label{poly03}
E_i^1 + E_i^t  =  E_j^{t_1} + E_j^{t_2}  \; \; i,j \in \{p,n\},
\end{equation}

\noindent holds only if $t_1 = \{t_1 | \mathbf{a}_{t_1} \in \mathcal{A}_h\} $, $t_2 = \{t_2 | \mathbf{a}_{t_2} \in \mathcal{A}_h\} $, $t_1 \neq t_2$, $\mathbf{h}_t \in \mathbf{H}_8 | t \in \{2,...,8\}$ and $h \in \{1,...,8\}$ and does not hold if $t_1 = \{t_1 | \mathbf{a}_{t_1} \in \mathcal{A}_{h_1}\} $, and $t_2 = \{t_2 | \mathbf{a}_{t_2} \in \mathcal{A}_{h_2}\} $, where $h_1 \neq h_2$. Therefore, we can conclude that we can append two vectors $\mathbf{v}_1 \in A_j$ and $\mathbf{v}_2 \in \mathcal{A}_i$ to $\mathbf{H}_m$, which do not hit any forbidden multisets by choosing any two vectors from different groups $1\leq i \neq j \leq 8$.

Here is how we divide $28 $ of $|E_n| = 2$ into $\mathcal{D}_i$'s,  $1\leq i \leq 7$ groups.    \\

$\mathcal{D}_1 = 
\renewcommand{\baselinestretch}{0.9}
{\footnotesize \setcounter{MaxMatrixCols}{34}
	\begin{bmatrix} -1 & \phantom{-}1 & \phantom{-}1 & \phantom{-}1 \\ -1 & \phantom{-}1 & \phantom{-}1 & \phantom{-}1  \\ \phantom{-}1 & -1 & \phantom{-}1 & \phantom{-}1  \\  \phantom{-}1 & -1 & \phantom{-}1 & \phantom{-}1 \\ \phantom{-}1 & \phantom{-}1 & -1 & \phantom{-}1 \\ \phantom{-}1 & \phantom{-}1 & -1 & \phantom{-}1 \\ \phantom{-}1 & \phantom{-}1 & \phantom{-}1 & -1 \\ \phantom{-}1 & \phantom{-}1 & \phantom{-}1 & -1  \end{bmatrix}},
\mathcal{D}_2 = 
\renewcommand{\baselinestretch}{0.9}
{\footnotesize \setcounter{MaxMatrixCols}{34}
	\begin{bmatrix} -1 & \phantom{-}1 & \phantom{-}1 & \phantom{-}1 \\ \phantom{-}1 & -1 & \phantom{-}1 & \phantom{-}1  \\ -1 & \phantom{-}1 & \phantom{-}1 & \phantom{-}1  \\  \phantom{-}1 & -1 & \phantom{-}1 & \phantom{-}1 \\ \phantom{-}1 & \phantom{-}1 & -1 & \phantom{-}1 \\ \phantom{-}1 & \phantom{-}1 & \phantom{-}1 & -1 \\ \phantom{-}1 & \phantom{-}1 & -1 & \phantom{-}1 \\ \phantom{-}1 & \phantom{-}1 & \phantom{-}1 & -1 \\ \end{bmatrix}},\\
\mathcal{D}_3 = 
\renewcommand{\baselinestretch}{0.9}
{\footnotesize \setcounter{MaxMatrixCols}{34}
	\begin{bmatrix} -1 & \phantom{-}1 & \phantom{-}1 & \phantom{-}1 \\ \phantom{-}1 & -1 & \phantom{-}1 & \phantom{-}1  \\ \phantom{-}1 & -1 & \phantom{-}1 & \phantom{-}1  \\  -1 & \phantom{-}1 & \phantom{-}1 & \phantom{-}1 \\ \phantom{-}1 & \phantom{-}1 & -1 & \phantom{-}1 \\ \phantom{-}1 & \phantom{-}1 & \phantom{-}1 & -1 \\ \phantom{-}1 & \phantom{-}1 & \phantom{-}1 & -1 \\ \phantom{-}1 & \phantom{-}1 & -1 & \phantom{-}1 \\ \end{bmatrix}} $,
$\mathcal{D}_4 = 
\renewcommand{\baselinestretch}{0.9}
{\footnotesize \setcounter{MaxMatrixCols}{34}
	\begin{bmatrix} -1 & \phantom{-}1 & \phantom{-}1 & \phantom{-}1 \\ \phantom{-}1 & -1 & \phantom{-}1 & \phantom{-}1  \\ \phantom{-}1 & \phantom{-}1 & -1 & \phantom{-}1  \\  \phantom{-}1 & \phantom{-}1 & \phantom{-}1 & -1 \\ -1 & \phantom{-}1 & \phantom{-}1 & \phantom{-}1 \\ \phantom{-}1 & -1 & \phantom{-}1 & \phantom{-}1 \\ \phantom{-}1 & \phantom{-}1 & -1 & \phantom{-}1 \\ \phantom{-}1 & \phantom{-}1 & \phantom{-}1 & -1  \end{bmatrix}},\\
\mathcal{D}_5 = 
\renewcommand{\baselinestretch}{0.9}
{\footnotesize \setcounter{MaxMatrixCols}{34}
	\begin{bmatrix} -1 & \phantom{-}1 & \phantom{-}1 & \phantom{-}1 \\ \phantom{-}1 & -1 & \phantom{-}1 & \phantom{-}1  \\ \phantom{-}1 & \phantom{-}1 & -1 & \phantom{-}1  \\  \phantom{-}1 & \phantom{-}1 & \phantom{-}1 & -1 \\ \phantom{-}1 & -1 & \phantom{-}1 & \phantom{-}1 \\ -1 & \phantom{-}1 & \phantom{-}1 & \phantom{-}1 \\ \phantom{-}1 & \phantom{-}1 & \phantom{-}1 & -1 \\ \phantom{-}1 & \phantom{-}1 & -1 & \phantom{-}1 \\ \end{bmatrix}},
\mathcal{D}_6 = 
\renewcommand{\baselinestretch}{0.9}
{\footnotesize \setcounter{MaxMatrixCols}{34}
	\begin{bmatrix} -1 & \phantom{-}1 & \phantom{-}1 & \phantom{-}1 \\ \phantom{-}1 & -1 & \phantom{-}1 & \phantom{-}1  \\ \phantom{-}1 & \phantom{-}1 & -1 & \phantom{-}1  \\  \phantom{-}1 & \phantom{-}1 & \phantom{-}1 & -1 \\ \phantom{-}1 & \phantom{-}1 & -1 & \phantom{-}1 \\ \phantom{-}1 & \phantom{-}1 & \phantom{-}1 & -1 \\ -1 & \phantom{-}1 & \phantom{-}1 & \phantom{-}1 \\ \phantom{-}1 & -1 & \phantom{-}1 & \phantom{-}1 \\ \end{bmatrix}} $,
$\mathcal{D}_7 =
\renewcommand{\baselinestretch}{0.9}
{\footnotesize \setcounter{MaxMatrixCols}{34}
	\begin{bmatrix} -1 & \phantom{-}1 & \phantom{-}1 & \phantom{-}1 \\ \phantom{-}1 & -1 & \phantom{-}1 & \phantom{-}1  \\ \phantom{-}1 & \phantom{-}1 & -1 & \phantom{-}1  \\  \phantom{-}1 & \phantom{-}1 & \phantom{-}1 & -1 \\ \phantom{-}1 & \phantom{-}1 & \phantom{-}1 & -1 \\ \phantom{-}1 & \phantom{-}1 & -1 & \phantom{-}1 \\ \phantom{-}1 & -1 & \phantom{-}1 & \phantom{-}1 \\ -1 & \phantom{-}1 & \phantom{-}1 & \phantom{-}1   \\ \end{bmatrix}} $\\

\nid equivalently, we can then write in multiset form as

\[ \mathcal{D}_1 = 
\renewcommand{\baselinestretch}{0.9}
{\footnotesize \setcounter{MaxMatrixCols}{34}
	\begin{array}{lcr}
	E_p= \{ 2, 3, 4, 5, 6, 7\} & E_n=\{ 0, 1\}  \\
	E_p=\{ 0, 1, 4, 5, 6, 7\} & E_n=\{ 2, 3\} \\
	E_p=\{ 0, 1, 2, 3, 6, 7\} & E_n=\{ 4, 5\} \\
	E_p=\{ 0, 1, 2, 3, 4, 5\} & E_n=\{ 6, 7\} \end{array}}\]
	
\[	\mathcal{D}_2 =
\renewcommand{\baselinestretch}{0.9}
{\footnotesize \setcounter{MaxMatrixCols}{34}
\begin{array}{lcr}
	E_p= \{ 1, 3, 4, 5, 6, 7\} & E_n=\{ 0, 2\}  \\
	E_p=\{ 0, 2, 4, 5, 6, 7\} & E_n=\{ 1, 3\} \\
	E_p=\{ 0, 1, 2, 3, 5, 7\} & E_n=\{ 4, 6\} \\
	E_p=\{ 0, 1, 2, 3, 4, 6\} & E_n=\{ 5, 7\} \end{array}}\]
\[ \mathcal{D}_3 = 
\renewcommand{\baselinestretch}{0.9}
{\footnotesize \setcounter{MaxMatrixCols}{34}
	\begin{array}{lcr}
	E_p= \{ 1, 2, 4, 5, 6, 7\} & E_n=\{ 0,3 \}  \\
	E_p=\{ 0, 3, 4, 5, 6, 7\} & E_n=\{ 1, 2\} \\
	E_p=\{ 0, 1, 2, 3, 5, 6\} & E_n=\{ 4, 7\} \\
	E_p=\{ 0, 1, 2, 3, 4, 7\} & E_n=\{ 5, 6\} \end{array}}\]
	
\[	\mathcal{D}_4 = 
\renewcommand{\baselinestretch}{0.9}
{\footnotesize \setcounter{MaxMatrixCols}{34}
    \begin{array}{lcr}
	E_p= \{ 1, 2, 3, 5, 6, 7\} & E_n=\{ 0, 4\}  \\
	E_p=\{ 0, 2, 3, 4, 6, 7\} & E_n=\{ 1, 5\} \\
	E_p=\{ 0, 1, 3, 4, 5, 7\} & E_n=\{ 2, 6\} \\
	E_p=\{ 0, 1, 2, 4, 5, 6\} & E_n=\{ 3, 7\} \end{array}}\]
\[ \mathcal{D}_5 = 
\renewcommand{\baselinestretch}{0.9}
{\footnotesize \setcounter{MaxMatrixCols}{34}
	\begin{array}{lcr}
	E_p= \{ 1, 2, 3, 4, 6, 7\} & E_n=\{ 0,5 \}  \\
	E_p=\{ 0, 2, 3, 5, 6, 7\} & E_n=\{ 1, 4\} \\
	E_p=\{ 0, 1, 3, 4, 5, 6\} & E_n=\{ 2, 7\} \\
	E_p=\{ 0, 1, 2, 4, 5, 7\} & E_n=\{ 3, 6\} \end{array}}\]
	
\[	\mathcal{D}_6 = 
	\renewcommand{\baselinestretch}{0.9}
{\footnotesize \setcounter{MaxMatrixCols}{34}
	\begin{array}{lcr}
	E_p= \{ 1, 2, 3, 4, 5, 7\} & E_n=\{ 0, 6\}  \\
	E_p=\{ 0, 2, 3, 4, 5, 6\} & E_n=\{ 1, 7\} \\
	E_p=\{ 0, 1, 3, 5, ,6, 7\} & E_n=\{ 2, 4\} \\
	E_p=\{ 0, 1, 2, 4, 6, 7\} & E_n=\{ 3, 5\} \end{array}}\]
\[ \mathcal{D}_7 = 
\renewcommand{\baselinestretch}{0.9}
{\footnotesize \setcounter{MaxMatrixCols}{34}
	\begin{array}{lcr}
	E_p= \{ 1, 2, 3, 4, 5, 6\} & E_n=\{ 0,7 \}  \\
	E_p=\{ 0, 2, 3, 4, 5, 7\} & E_n=\{ 1, 6\} \\
	E_p=\{ 0, 1, 3, 4, 6, 7\} & E_n=\{ 2, 5\} \\
	E_p=\{ 0, 1, 2, 5, 6, 7\} & E_n=\{ 3, 4\} \end{array}}\]

It can be proved that
\begin{equation}
\label{poly04}
E_i^1 + E_i^t  =  E_j^{t_1} + E_j^{t_2}  \; \; i,j \in \{p,n\},
\end{equation}
and
\begin{equation}
\label{poly05}
E_i^{t'} + E_i^{t''}  =  E_j^{t_1} + E_j^{t_2}  \; \; i,j \in \{p,n\},
\end{equation}
holds only if $t_1 = \{t_1 | \mathbf{d}_{t_1} \in \mathcal{D}_h\} $, $t_2 = \{t_2 | \mathbf{d}_{t_2} \in \mathcal{D}_h\} $, $t_1 \neq t_2$, where $\mathbf{h}_{t}, \mathbf{h}_{t'}, \mathbf{h}_{t''},  \in \mathbf{H}_8 | t, t', t'' \in \{2,...,8\}$ and $h \in \{1,...,7\}$ and does not hold if if $t_1 = \{t_1 | \mathbf{d}_{t_1} \in \mathcal{D}_{h_1}\} $, and $t_2 = \{t_2 | \mathbf{d}_{t_2} \in \mathcal{D}_{h_2}\} $, where $h_1 \neq h_2$. Therefore, we can conclude that we can append two vectors $\mathbf{v}_1 \in \mathcal{D}_j$ and $\mathbf{v}_2 \in \mathcal{D}_i$ to $\mathbf{H}_m$, which do not hit any forbidden multisets by choosing any two vectors from different groups $1\leq i \neq j \leq 7$.

Dividing $56$ of $|E_n| = 3$  into $\mathcal{F}_i$, $1\leq i \leq 8$ groups is equivalent exactly $\mathcal{A}_i$'s. Similarly, we can prove that
\begin{equation}
\label{poly06}
E_i^{t'} + E_i^{t''}  =  E_j^{t_1} + E_j^{t_2}  \; \; i,j \in \{p,n\},
\end{equation}
holds only if $t_1 = \{t_1 | \mathbf{f}_{t_1} \in \mathcal{F}_h\} $, $t_2 = \{t_2 | \mathbf{f}_{t_2} \in \mathcal{F}_h\} $, $t_1 \neq t_2$, where $\mathbf{h}_{t'}, \mathbf{h}_{t''},  \in \mathcal{H}_8 | t', t'' \in \{2,...,8\}$ and $h \in \{1,...,8\}$ and does not hold if $t_1 = \{t_1 | \mathbf{f}_{t_1} \in \mathcal{F}_{h_1}\} $, and $t_2 = \{t_2 | \mathbf{f}_{t_2} \in \mathcal{F}_{h_2}\} $, where $h_1 \neq h_2$. Therefore, we can conclude that we can append two vectors $\mathbf{v}_1 \in \mathcal{F}_j$ and $\mathbf{v}_2 \in \mathcal{F}_i$ to $\mathbf{H}_m$, which do not hit any forbidden multisets by choosing any two vectors from different groups $1\leq i \neq j \leq 8$.

Here is how we divide $70/2 - 7 =  28 $ of $|E_n| = 4$ into $\mathcal{G}_i$'s,  $1\leq i \leq 7$ groups.  \\  

$\mathcal{G}_1 =
\renewcommand{\baselinestretch}{0.9}
{\footnotesize \setcounter{MaxMatrixCols}{34}
	\begin{bmatrix}  \phantom{-}1 & \phantom{-}1 & \phantom{-}1 & \phantom{-}1 \\ -1 & -1 & \phantom{-}1 & \phantom{-}1  \\ -1 &  \phantom{-}1 & -1 & \phantom{-}1  \\  -1 &  \phantom{-}1 & \phantom{-}1 & -1 \\ \phantom{-}1 & -1 & -1 & \phantom{-}1 \\ \phantom{-}1 & -1 &  \phantom{-}1 & -1 \\ \phantom{-}1 & \phantom{-}1 & -1 & -1 \\ -1 & -1 & -1 & -1  \end{bmatrix}},
\mathcal{G}_2 = 
\renewcommand{\baselinestretch}{0.9}
{\footnotesize \setcounter{MaxMatrixCols}{34}
	\begin{bmatrix}  \phantom{-}1 & \phantom{-}1 & \phantom{-}1 & \phantom{-}1 \\ -1 & -1 & \phantom{-}1 & \phantom{-}1  \\ -1 & \phantom{-}1 & -1 & \phantom{-}1  \\  -1 &  \phantom{-}1 & \phantom{-}1 & -1 \\ \phantom{-}1 & -1 &  \phantom{-}1 & -1 \\ \phantom{-}1 & -1 & -1 &  \phantom{-}1 \\ -1 & -1 & -1 & -1 \\ \phantom{-}1 & \phantom{-}1 & -1 & -1 \\ \end{bmatrix}},\\
\mathcal{G}_3 = 
\renewcommand{\baselinestretch}{0.9}
{\footnotesize \setcounter{MaxMatrixCols}{34}
	\begin{bmatrix}  \phantom{-}1 & \phantom{-}1 & \phantom{-}1 & \phantom{-}1 \\ -1 & -1 & \phantom{-}1 & \phantom{-}1  \\ -1 &  \phantom{-}1 & -1 & \phantom{-}1  \\  -1 & \phantom{-}1 & \phantom{-}1 & -1 \\ \phantom{-}1 & \phantom{-}1 & -1 & -1 \\ -1 & -1 & -1 &  -1 \\ \phantom{-}1 & -1 & -1 &  \phantom{-}1 \\ \phantom{-}1 & -1 &  \phantom{-}1 & -1 \\ \end{bmatrix}} $
$\mathcal{G}_4 =
\renewcommand{\baselinestretch}{0.9}
{\footnotesize \setcounter{MaxMatrixCols}{34}
	\begin{bmatrix}  \phantom{-}1 & \phantom{-}1 & \phantom{-}1 & \phantom{-}1 \\ -1 & -1 & \phantom{-}1 & \phantom{-}1  \\ -1 & \phantom{-}1 & -1 & \phantom{-}1  \\  -1 & \phantom{-}1 & \phantom{-}1 & -1 \\ -1 & -1 & -1 & -1 \\ \phantom{-}1 &  \phantom{-}1 & -1 & -1 \\ \phantom{-}1 & -1 &  \phantom{-}1 & -1 \\ \phantom{-}1 & -1 & -1 &  \phantom{-}1  \end{bmatrix}},\\
\mathcal{G}_5 = 
\renewcommand{\baselinestretch}{0.9}
{\footnotesize \setcounter{MaxMatrixCols}{34}
	\begin{bmatrix}  \phantom{-}1 & \phantom{-}1 & \phantom{-}1 & \phantom{-}1 \\ -1 & -1 & \phantom{-}1 & \phantom{-}1  \\ \phantom{-}1 & \phantom{-}1 & -1 & -1  \\  -1 & -1 & -1 & -1 \\ -1 &  \phantom{-}1 & -1 & \phantom{-}1 \\ -1 & \phantom{-}1 & \phantom{-}1 & -1 \\ \phantom{-}1 & -1 & -1 &  \phantom{-}1 \\ \phantom{-}1 & -1 &  \phantom{-}1 & -1 \\ \end{bmatrix}},
\mathcal{G}_6 = 
\renewcommand{\baselinestretch}{0.9}
{\footnotesize \setcounter{MaxMatrixCols}{34}
	\begin{bmatrix}  \phantom{-}1 & \phantom{-}1 & \phantom{-}1 & \phantom{-}1 \\ -1 & -1 & \phantom{-}1 & \phantom{-}1  \\ -1 & -1 & -1 & -1  \\  \phantom{-}1 & \phantom{-}1 & -1 & -1 \\ -1 & \phantom{-}1 &  -1 &  \phantom{-}1 \\ -1 &  \phantom{-}1 & \phantom{-}1 & -1 \\  \phantom{-}1 & -1 & \phantom{-}1 & -1 \\ \phantom{-}1 & -1 & -1 & \phantom{-}1 \\ \end{bmatrix}}, $\\
$\mathcal{G}_7 =
\renewcommand{\baselinestretch}{0.9}
{\footnotesize \setcounter{MaxMatrixCols}{34}
	\begin{bmatrix}  \phantom{-}1 & \phantom{-}1 & \phantom{-}1 & \phantom{-}1 \\ -1 & -1 & -1 & -1  \\ -1 & -1 &  \phantom{-}1 & \phantom{-}1  \\  \phantom{-}1 & \phantom{-}1 & -1 & -1 \\ -1 & \phantom{-}1 & -1 &  \phantom{-}1 \\ \phantom{-}1 & -1 &  \phantom{-}1 & -1 \\ -1 &  \phantom{-}1 & \phantom{-}1 & -1 \\  \phantom{-}1 & -1 & -1 & \phantom{-}1   \\ \end{bmatrix}} $

\nid equivalently, we can then write in multiset form as

\[ \mathcal{G}_1 = 
\renewcommand{\baselinestretch}{0.9}
{\footnotesize \setcounter{MaxMatrixCols}{34}
	\begin{array}{lcr}
	E_p= \{ 0, 4, 5, 6 \} & E_n=\{ 1, 2, 3, 7\}  \\
	E_p= \{ 0, 2, 3, 6 \} & E_n=\{ 1, 4, 5, 7\} \\
	E_p= \{ 0, 1, 3, 5 \} & E_n=\{ 2, 4, 6, 7\} \\
	E_p= \{ 0, 1, 2, 4 \} & E_n=\{ 3, 5, 6, 7\}  \end{array}}\]
	
\[	\mathcal{G}_2 = 
\renewcommand{\baselinestretch}{0.9}
{\footnotesize \setcounter{MaxMatrixCols}{34}
\begin{array}{lcr}
	E_p= \{ 0, 4, 5, 7 \} & E_n=\{ 1, 2, 3, 6\}  \\
	E_p= \{ 0, 2, 3, 7 \} & E_n=\{ 1, 4, 5, 6\} \\
	E_p= \{ 0, 1, 3, 4 \} & E_n=\{ 2, 5, 6, 7\} \\
	E_p= \{ 0, 1, 2, 5 \} & E_n=\{ 3, 4, 6, 7\} \end{array}}\]
\[ \mathcal{G}_3 =
\renewcommand{\baselinestretch}{0.9}
{\footnotesize \setcounter{MaxMatrixCols}{34}
	\begin{array}{lcr}
	E_p= \{ 0, 4, 6, 7 \} & E_n=\{ 1, 2, 3, 5\}  \\
	E_p= \{ 0, 2, 3, 4 \} & E_n=\{ 1, 5, 6, 7\} \\
	E_p= \{ 0, 1, 3, 7 \} & E_n=\{ 2, 4, 5, 6\} \\
	E_p= \{ 0, 1, 5, 6 \} & E_n=\{ 2, 3, 4, 7\}  \end{array}}\]
	
\[	\mathcal{G}_4 = 
\renewcommand{\baselinestretch}{0.9}
{\footnotesize \setcounter{MaxMatrixCols}{34}
    \begin{array}{lcr}
	E_p= \{ 0, 5, 6, 7 \} & E_n=\{ 1, 2, 3, 4\}  \\
	E_p= \{ 0, 2, 3, 5 \} & E_n=\{ 1, 4, 6, 7\} \\
	E_p= \{ 0, 1, 3, 6 \} & E_n=\{ 2, 4, 5, 7\} \\
	E_p= \{ 0, 1, 2, 7 \} & E_n=\{ 3, 4, 5, 6\} \end{array}}\]
\[ \mathcal{G}_5 =
\renewcommand{\baselinestretch}{0.9}
{\footnotesize \setcounter{MaxMatrixCols}{34}
	\begin{array}{lcr}
	E_p= \{ 0, 2, 6, 7 \} & E_n=\{ 1, 3, 4, 5\}  \\
	E_p= \{ 0, 2, 4, 5 \} & E_n=\{ 1, 3, 6, 7\} \\
	E_p= \{ 0, 1, 5, 7 \} & E_n=\{ 2, 3, 4, 6\} \\
	E_p= \{ 0, 1, 4, 6 \} & E_n=\{ 2, 3, 5, 7\}  \end{array}}\]
	
\[	\mathcal{G}_6 = 
\renewcommand{\baselinestretch}{0.9}
{\footnotesize \setcounter{MaxMatrixCols}{34}
\begin{array}{lcr}
	E_p= \{ 0, 3, 6, 7 \} & E_n=\{ 1, 2, 4, 5\}  \\
	E_p= \{ 0, 3, 4, 5 \} & E_n=\{ 1, 2, 6, 7\} \\
	E_p= \{ 0, 1, 5, 6 \} & E_n=\{ 2, 3, 4, 7\} \\
	E_p= \{ 0, 1, 4, 7 \} & E_n=\{ 2, 3, 5, 6\} \end{array}}\]
\[ \mathcal{G}_7 =
\renewcommand{\baselinestretch}{0.9}
{\footnotesize \setcounter{MaxMatrixCols}{34}
	\begin{array}{lcr}
	E_p= \{ 0, 3, 5, 7 \} & E_n=\{ 1, 2, 4, 6\}  \\
	E_p= \{ 0, 3, 4, 6 \} & E_n=\{ 1, 2, 5, 7\} \\
	E_p= \{ 0, 2, 5, 6 \} & E_n=\{ 1, 3, 4, 7\} \\
	E_p= \{ 0, 2, 4, 7 \} & E_n=\{ 1, 3, 5, 6\}  \end{array}} \]

It can be proved that
\begin{equation}
\label{poly07}
E_i^{t'} + E_i^{t''}  =  E_j^{t_1} + E_j^{t_2}  \; \; i,j \in \{p,n\},
\end{equation}
holds only if $t_1 = \{t_1 | \mathbf{g}_{t_1} \in \mathcal{G}_h\} $, $t_2 = \{t_2 | \mathbf{g}_{t_2} \in \mathcal{G}_h\} $, $t_1 \neq t_2$, where $\mathbf{h}_{t'}, \mathbf{h}_{t''},  \in \mathbf{H}_8 | t', t'' \in \{2,...,8\}$ and $h \in \{1,...,7\}$ and does not hold if $t_1 = \{t_1 | \mathbf{g}_{t_1} \in \mathcal{G}_{h_1}\} $, and $t_2 = \{t_2 | \mathbf{g}_{t_2} \in \mathcal{G}_{h_2}\} $, where $h_1 \neq h_2$. Therefore, we can conclude that we can append two vectors $\mathbf{v}_1 \in \mathcal{G}_j$ and $\mathbf{v}_2 \in \mathcal{G}_i$ to $\mathbf{H}_m$, which do not hit any forbidden multisets by choosing any two vectors from different groups $1\leq i \neq j \leq 7$.
Note that $\mathcal{A}_i$'s group can be constructed from $\mathcal{A}_i$'s and $\mathcal{D}_i$'s such as $\mathcal{A}_1 = \{ \mathcal{A}_2+\mathcal{D}_1, \mathcal{A}_3+\mathcal{D}_2, \mathcal{A}_4+\mathcal{D}_3, \mathcal{A}_5+\mathcal{D}_4, \mathcal{A}_6+\mathcal{D}_5, \mathcal{A}_7+\mathcal{D}_6, \mathcal{A}_8+\mathcal{D}_7  \}$, $\mathcal{A}_2 = \{ \mathcal{A}_1+\mathcal{D}_1, \mathcal{A}_3+\mathcal{D}_3, \mathcal{A}_4+\mathcal{D}_2, \mathcal{A}_5+\mathcal{D}_5, \mathcal{A}_6+\mathcal{D}_4, \mathcal{A}_7+\mathcal{D}_7, \mathcal{A}_8+\mathcal{D}_6  \}$, $\mathcal{A}_3 = \{ \mathcal{A}_1+\mathcal{D}_2, \mathcal{A}_2+\mathcal{D}_3, \mathcal{A}_4+\mathcal{D}_1, \mathcal{A}_5+\mathcal{D}_6, \mathcal{A}_6+\mathcal{D}_7, \mathcal{A}_7+\mathcal{D}_4, \mathcal{A}_8+\mathcal{D}_5, \}$,  $\mathcal{A}_4 = \{ \mathcal{A}_1+\mathcal{D}_3, \mathcal{A}_2+\mathcal{D}_2, \mathcal{A}_3+\mathcal{D}_1, \mathcal{A}_5+\mathcal{D}_7, \mathcal{A}_6+\mathcal{D}_6, \mathcal{A}_7+\mathcal{D}_5, \mathcal{A}_8+\mathcal{D}_4  \}$, $\mathcal{A}_5 = \{ \mathcal{A}_1+\mathcal{D}_4, \mathcal{A}_2+\mathcal{D}_5, \mathcal{A}_3+\mathcal{D}_6, \mathcal{A}_4+\mathcal{D}_7, \mathcal{A}_6+\mathcal{D}_1, \mathcal{A}_7+\mathcal{D}_2, \mathcal{A}_8+\mathcal{D}_3  \}$, $\mathcal{A}_6 = \{ \mathcal{A}_1+\mathcal{D}_5, \mathcal{A}_2+\mathcal{D}_4, \mathcal{A}_3+\mathcal{D}_7, \mathcal{A}_4+\mathcal{D}_6, \mathcal{A}_5+\mathcal{D}_1, \mathcal{A}_7+\mathcal{D}_3, \mathcal{A}_8+\mathcal{D}_2  \}$, $\mathcal{A}_7 = \{ \mathcal{A}_1+\mathcal{D}_6, \mathcal{A}_2+\mathcal{D}_7, \mathcal{A}_3+\mathcal{D}_4, \mathcal{A}_4+\mathcal{D}_5, \mathcal{A}_5+\mathcal{D}_2, \mathcal{A}_6+\mathcal{D}_3, \mathcal{A}_8+\mathcal{D}_1  \}$, $\mathcal{A}_8 = \{ \mathcal{A}_1+\mathcal{D}_7, \mathcal{A}_2+\mathcal{D}_6, \mathcal{A}_3+\mathcal{D}_5, \mathcal{A}_4+\mathcal{D}_4, \mathcal{A}_5+\mathcal{D}_3, \mathcal{A}_6+\mathcal{D}_2, \mathcal{A}_7+\mathcal{D}_1  \}$. Also the $\mathcal{G}_i$'s group can be constructed from different combinations of  $\mathcal{A}_i$'s or $\mathcal{D}_i$'s groups. For example, $\mathcal{G}_1 = \{ \mathcal{A}_1+\mathcal{A}_8, \mathcal{A}_2 +\mathcal{A}_7, \mathcal{A}_3 +\mathcal{A}_6, \mathcal{A}_4 +\mathcal{A}_5 \}$, $\mathcal{G}_2 = \{ \mathcal{A}_1+\mathcal{A}_7, \mathcal{A}_2 +\mathcal{A}_8, \mathcal{A}_3 +\mathcal{A}_5, \mathcal{A}_4 +\mathcal{A}_6 \}$, $\mathcal{G}_3 = \{ \mathcal{A}_1+\mathcal{A}_6, \mathcal{A}_2 +\mathcal{A}_5, \mathcal{A}_3 +\mathcal{A}_8, \mathcal{A}_4 +\mathcal{A}_7 \}$, $\mathcal{G}_4 = \{ \mathcal{A}_1+\mathcal{A}_5, \mathcal{A}_2 +\mathcal{A}_6, \mathcal{A}_3 +\mathcal{A}_7, \mathcal{A}_4 +\mathcal{A}_8 \}$, $\mathcal{G}_5 = \{ \mathcal{A}_1+\mathcal{A}_4, \mathcal{A}_2 +\mathcal{A}_3, \mathcal{A}_5 +\mathcal{A}_8, \mathcal{A}_6 +\mathcal{A}_7 \}$, $\mathcal{G}_6 = \{ \mathcal{A}_1+\mathcal{A}_3, \mathcal{A}_2 +\mathcal{A}_4, \mathcal{A}_5 +\mathcal{A}_7, \mathcal{A}_6 +\mathcal{A}_8 \}$ , $\mathcal{G}_7 = \{ \mathcal{A}_1+\mathcal{A}_2, \mathcal{A}_3 +\mathcal{A}_4, \mathcal{A}_5 +\mathcal{A}_6, \mathcal{A}_7 +\mathcal{A}_8 \}$ and $\mathcal{G}_1 = \{ \mathcal{D}_1+\mathcal{D}_6, \mathcal{D}_2 +\mathcal{D}_5, \mathcal{D}_3 +\mathcal{D}_4 \}$, $\mathcal{G}_2 = \{ \mathcal{D}_1+\mathcal{D}_7, \mathcal{D}_2 +\mathcal{D}_4, \mathcal{D}_3 +\mathcal{D}_5 \}$, $\mathcal{G}_3 = \{ \mathcal{D}_1+\mathcal{D}_4, \mathcal{D}_2 +\mathcal{D}_7, \mathcal{D}_3 +\mathcal{D}_6 \}$, $\mathcal{G}_4 = \{ \mathcal{D}_1+\mathcal{D}_5, \mathcal{D}_2 +\mathcal{D}_6, \mathcal{D}_3 +\mathcal{D}_7 \}$, $\mathcal{G}_5 = \{ \mathcal{D}_1+\mathcal{D}_2, \mathcal{D}_4 +\mathcal{D}_7, \mathcal{D}_5 +\mathcal{D}_6 \}$, $\mathcal{G}_6 = \{ \mathcal{D}_1+\mathcal{D}_3, \mathcal{D}_4 +\mathcal{D}_6, \mathcal{D}_5 +\mathcal{D}_7 \}$ and $\mathcal{G}_7 = \{ \mathcal{D}_2+\mathcal{D}_3, \mathcal{D}_4 +\mathcal{D}_5, \mathcal{D}_6 +\mathcal{D}_7 \}$.

Now, let us count how many vector sets $\mathbf{V} = [\mathbf{v}_{j_1} \mathbf{v}_{j_2} ]$ hit forbidden multisets. We can easily count after classifying $\mathcal{B}_8^+$ into groups as discussed above. Hence, the total number of vector sets that hit forbidden matrices are
\begin{equation}\nonumber
\label{total2}
\!\binom{8}{1}\!\times 7 +\!\!  \binom{4}{2}\!\times 7 +  \!\binom{7}{2}\!\times 8 + \!\!\binom{4}{2}\!\times 7\! = \!56 \!+ \!42 + \!168 + \!42 = 308.
\end{equation}

There are no other different combinations of two vectors in $\mathcal{B}_8^+$ that can hit the forbidden multisets. We find that our computed number of comination that hits the forbidden multiset $308 < 7140$ is less than the total number of two vectors combination sets, then we can claim that the maximum number of vectors that can be added to $\mathbf{H}_8$ is $(K_{\rm{max}}^a -L) \geq 2$. This method of classifying $\mathcal{B}_8^+$ into groups not only helps us to prove the maximum number of vectors but also on how to construct such vector sets that posses unique decodability property (\ref{null02}).

Similar computation can be carried out for the cases $k=3, 4, 5$ and still claim that the number of computation is less than the total number of $k=3, 4, 5$ vectors combination sets. As an example, we present two of such combinations below for the case of $k=5$,

\[\mathbf{V}_1
\!\!= \!\! \!
\renewcommand{\baselinestretch}{0.9}
{\footnotesize \setcounter{MaxMatrixCols}{34}
	\begin{bmatrix} -1 & \phantom{-}1 & \phantom{-}1 & -1 & \phantom{-}1 \\ \phantom{-}1 & -1 & \phantom{-}1 & \phantom{-}1 & -1 \\ \phantom{-}1 & \phantom{-}1  & -1 & \phantom{-}1 & -1  \\  \phantom{-}1 & \phantom{-}1 & \phantom{-}1 & \phantom{-}1  & -1 \\ \phantom{-}1 & \phantom{-}1 & \phantom{-}1 & -1 & -1 \\ \phantom{-}1 & \phantom{-}1 & \phantom{-}1 & \phantom{-}1 & \phantom{-}1\\ \phantom{-}1 & \phantom{-}1 & \phantom{-}1 & \phantom{-}1 & \phantom{-}1  \\ \phantom{-}1 & \phantom{-}1 & \phantom{-}1 & \phantom{-}1 & \phantom{-}1  \end{bmatrix}},\\
\!\!\mathbf{V}_2 \! 
\!= \!\! \!
\renewcommand{\baselinestretch}{0.9}
{\footnotesize \setcounter{MaxMatrixCols}{34}
	\begin{bmatrix} -1 & \phantom{-}1 & \phantom{-}1 & \phantom{-}1 & -1 \\ \phantom{-}1 & -1 & \phantom{-}1 & \phantom{-}1 & -1 \\ \phantom{-}1 & \phantom{-}1  & -1 & -1 & \phantom{-}1  \\  \phantom{-}1 & \phantom{-}1 & \phantom{-}1 & \phantom{-}1  & -1 \\ \phantom{-}1 & \phantom{-}1 & \phantom{-}1 & -1 & -1 \\ \phantom{-}1 & \phantom{-}1 & \phantom{-}1 & \phantom{-}1 & \phantom{-}1\\ \phantom{-}1 & \phantom{-}1 & \phantom{-}1 & \phantom{-}1 & \phantom{-}1  \\ \phantom{-}1 & \phantom{-}1 & \phantom{-}1 & \phantom{-}1 & \phantom{-}1  \end{bmatrix}}\!, \]
\[ \mathbf{V}_1 =
\renewcommand{\baselinestretch}{0.9}
{\footnotesize \setcounter{MaxMatrixCols}{34}
	\begin{array}{lcr}
	E_p= \{ 1, 2, 3, 4, 5, 6, 7\} & E_n=\{ 0\}  \\
	E_p=\{ 0, 2, 3, 4, 5, 6, 7\} & E_n=\{ 1\} \\
	E_p=\{ 0, 1, 3, 4, 5, 6, 7\} & E_n=\{ 2\} \\
	E_p=\{ 1, 2, 3, 4, 5, 6,7\} & E_n=\{ 0, 4\} \\
	E_p=\{ 0,  5, 6,7\} & E_n=\{ 1,2,3,4\} \end{array}},\]
	
\[ \mathbf{V}_2 = 
\renewcommand{\baselinestretch}{0.9}
{\footnotesize \setcounter{MaxMatrixCols}{34}
     \begin{array}{lcr}
	E_p= \{ 1, 2, 3, 4, 5, 6, 7\} & E_n=\{ 0\}  \\
	E_p=\{ 0, 2, 3, 4, 5, 6, 7\} & E_n=\{ 1\} \\
	E_p=\{ 0, 1, 3, 4, 5, 6, 7\} & E_n=\{ 2\} \\
	E_p=\{ 0,1,  3, 4,  6,7\} & E_n=\{ 2, 5\} \\
	E_p=\{ 2,  5, 6,7\} & E_n=\{ 0,1,3,4\}. \end{array}}.\]
	
In both construction examples $\mathbf{V}_1$ and $\mathbf{V}_2$, we cannot find any $2$ vector combinations that belong to the same $\mathcal{A}_i$'s, $\mathcal{D}_i$'s, $\mathcal{F}_i$'s and $\mathcal{G}_i$'s. Therefore, all possible combinations of $2$ vectors do not hit any of forbidden multisets.

However, once we add any other vector $\mathbf{v}_6 \in \mathcal{B}_8^+$ to the above sets some of the combinations of resulting vector sets hit one of the forbidden multisets. This means that if $k=6$ we compute the number of combinations that hit the forbidden multisets is exactly equal to $\binom{2^{8-1} - 8}{6} = \binom{120}{6}=3,652,745,460$. Hence, the maximum number of vectors $\mathbf{v}_i$ that can be appended to $\mathbf{H}_8$ is $f_2(8) = 5$. The proof is complete. Next, we also show the different groups of such combinations.

Among all possible constructions, here we show that we can only have these group combinations $(\mathcal{A}, \mathcal{A}, \mathcal{A}, \mathcal{D}, \mathcal{G})$, $(\mathcal{A}, \mathcal{A}, \mathcal{D}, \mathcal{D}, \mathcal{D})$, $(\mathcal{A}, \mathcal{A}, \mathcal{D}, \mathcal{D}, \mathcal{G})$, $(\mathcal{A}, \mathcal{A}, \mathcal{D}, \mathcal{G}, \mathcal{G})$,  $(\mathcal{A}, \mathcal{D}, \mathcal{D}, \allowbreak \mathcal{D}, \mathcal{D})$,  $(\mathcal{A}, \mathcal{D}, \mathcal{D}, \mathcal{D}, \mathcal{G})$, $(\mathcal{A}, \mathcal{D}, \mathcal{D}, \mathcal{G}, \mathcal{G})$ and $(\mathcal{A}, \mathcal{D}, \mathcal{G}, \mathcal{G}, \mathcal{G})$. We take each combinations and using the rules developed in \cite{michel2012}, we show in Apendix \ref{appendixB} that we cannot add any more columns from any groups.

\section{Minimum Distance of Code sets}
	\label{minDist}	
	The Manhattan distance \cite{craw2010} equivalently (${\ell}_1$)-norm of two $L$-dimensional vectors $\mathbf{y}_i$ and $\mathbf{y}_j$ for $i \neq j$ is defined as
	\begin{eqnarray}
	\label{dist}
	d_L (\mathbf{y}_i, \mathbf{y}_j) = \sum_{t=1}^L | y_{i,t} -y_{j,t}|,
	\end{eqnarray}
	\nid where $|\cdot|$ denotes complex amplitude. Then the general minimum Manhattan distance of received vectors for a given antipodal code set can be formulated by
	\begin{eqnarray}
	\label{MinDistCode}
	d_{\rm{min}} (\mathbf{C}) = \argmin_{\substack{\mathbf{x}_i,\mathbf{x}_j \in \{\pm 1\}^{K\times 1} \\ \mathbf{y}_i =\mathbf{C}\mathbf{x}_i,\mathbf{y}_j =\mathbf{C}\mathbf{x}_j}} d_L (\mathbf{y}_i, \mathbf{y}_j).
	\end{eqnarray} 
	\begin{theorem}
		\label{theor1}
		Let $\mathfrak{C} \in \{ \pm 1\}^{L \times K}$ represent the set of all antipodal matrices constructed by distinct\footnote{Not only the columns required to be distinct but we assume any column multiplication with minus one should result distinct columns as well.} columns, then the minimum distance of the code set, $\delta (\mathfrak{C})$, is equal to $4$, where
		\begin{equation}
		    \delta (\mathfrak{C}) =\argmin_{\substack{\mathbf{C}' \in \mathfrak{C}}} d_{min} (\mathbf{C}').
		\end{equation}
		
	\end{theorem}
	\begin{proof}
		Let assume that $\delta (\mathfrak{C}) = d_{\rm{min}} (\mathbf{C}) = d_L (\mathbf{y}_n, \mathbf{y}_m) $, where $\mathbf{y}_n = \mathbf{C}\mathbf{x}_{n}$, $\mathbf{y}_m = \mathbf{C}\mathbf{x}_{m}$, $\mathbf{y}_n, \mathbf{y}_m \in \mathcal{\mathcal{N}}^{L\times 1}$, $\mathcal{N} \in \{\pm K, \pm (K-2), ...\}$, and $\mathbf{x}_n, \mathbf{x}_m \in \{\pm 1\}^{K\times 1}$. The minimum value obtained when the difference vector $\mathbf{y} = \mathbf{y}_n- \mathbf{y}_m = \mathbf{C}(\mathbf{x}_{n}-\mathbf{x}_{m}) = \mathbf{C}\bar{\mathbf{x}}$ has only one non-zero element $y_c \neq 0$, $y_{n,c} \neq y_{m,c}$, and $L-1$ zeros $y_i = 0$, $y_{n,i} = y_{m,i}$ for $i \neq c$ where $\bar{\mathbf{x}} \in \{0, \pm 2\}^{K \times 1}$. In order to achieve only one non-zero element of $\mathbf{y}$ then $\mathbf{c}^i\bar{\mathbf{x}} = 0$ for $ \forall i \in \{1, \dots, L \}$ except $i \neq c$ where $\mathbf{c}^i$ is the $i$-th row of matrix $\mathbf{C}$. One obvious remark is that the number of $\pm 2$ elements in $\bar{\mathbf{x}}$ must be even for the inner product, $\mathbf{c}^i\bar{\mathbf{x}}$, to be zero. One possible option is to have all zeros but two non-zero elements in $\bar{\mathbf{x}}$. In this case the matrix $\mathbf{C}$ must have at least two columns $\mathbf{c}_j$ and $\mathbf{c}_k$ with $c_{j,i}= c_{k,i}$ or $c_{j,i}= -c_{k,i}$ for $ \forall i \in \{1, \dots, L \}$ except $i \neq c$ where $1 \leq j\neq k \leq K$ to result in a vector $\mathbf{y}$ having only one non-zero element. In other words, the difference of $\mathbf{c}_j$ and $\mathbf{c}_k$ is either in one or $L-1$ elements. Therefore, $y_{n,c} - y_{m,c} =\pm 2(c_{j,c}-c_{k,c})$ where $c_{j,c}-c_{k,c}$ results in $d_{\rm{min}} (\mathbf{C}) = |y_{n,c} - y_{m,c} | = 4$. 
	\end{proof}
	Now that we proved that $\delta (\mathfrak{C}) = 4$, we will try to find $d_{\rm{min}} (\mathbf{C})$ of our proposed UD code sets $\mathbf{C} \in \mathcal{C} \subset \mathfrak{C}$, where $\mathcal{C} \in \{ \pm 1\}^{L \times K}$ is the set of all the antipodal UD code sets. Let us follow the option of having all zeros but two non-zero elements in $\bar{\mathbf{x}}$. In the case of $L=4$ the columns $\mathbf{c}_1$ and $\mathbf{c}_5$ differ in the first element only. If $\bar{\mathbf{x}} = [2, 0,0,0,-2]^T$ then the difference vector is $\mathbf{y} = [4,0,0,0]^T$. Note that even if we substitute the $\mathbf{c}_1$ by $-\mathbf{c}_1$ we can still obtain same $\mathbf{y}$ with $\bar{\mathbf{x}} = [2, 0,0,0,2]^T$. Based on our construction in \cite{michel2012}, we look at the case of $L=8$. Observe that all the elements of the columns $9$-th and $12$-th, $[\alpha^{13}, 0]^T$ and $[\alpha^{13}, \alpha^{13}]^T$ of the $\mathbf{C}$ are equal except the $5$-th element in which they differ. If we select $x_{n,9} \neq x_{m,9}$, $x_{n,12} \neq x_{m,12}$, and $x_{n,i} = x_{m,i}$ for all $i \notin \{9,12\}$ then $y_{n,5} =2$ and $y_{m,5} =-2$ or $y_{n,5} =-2$ and $y_{m,5} =2$, which will result in $d_L (\mathbf{y}_n, \mathbf{y}_m) = 4$. With this specific observation together with the Theorem \ref{theor1}, we conclude that $d_{\rm{min}} (\mathbf{C}) = 4$. From this observation, we learn that if any two columns differ at one element in a UD code set we assure that $d_{\rm{min}} (\mathbf{C}) = 4$. Similarly, for $L=16$  columns $17$-th and $27$-th, $[\alpha^{13}, 0, \alpha^{13}, 0]^T$ and $[\alpha^{13}, 0, 0, 0]^T$  differ in one element only. Due to our recursive construction in \cite{michel2012} for $L = 2^p$, where $p = 5, 6, ...$ columns $p2^{(p-2)}+3$ and $(p-1)2^{(p-1)}+3$ differ in one element only. Therefore, all the UD code set generated in \cite{michel2012} has $d_{\rm{min}} (\mathbf{C}) = 4$.
	
	\section{Fast Decoding Algorithm in AWGN}
	\label{fastDecoder}
	The recursive linear NDA decoder, discussed in \cite{michel2012}, is not suitable for the noisy transmission channel. Let the received vector in the presence of noise be mathematically formulated as
	\begin{eqnarray}
	\label{systemModel}
	\mathbf{y} &=& A\mathbf{C}\mathbf{x} + \mathbf{n} \\
	&=& \sum_{k=1}^K A  \mathbf{c}_k x_k  + \mathbf{n},
	\end{eqnarray} 	
	\noindent where $A$ denotes the amplitude, $\mathbf{c}_k \in \{\pm 1\}^{L \times 1}$ are signatures for $1\leq k \leq K$, $\mathbf{x} \in \{\pm 1\}^{K\times 1}$ is user data and $\mathbf{n}$ is the Additive white Gaussian noise (\acs{AWGN}) channel noise vector with a variance of $\sigma^2$. 
	
		\vspace{-0.0cm}

	The objective of the receiver is as follows; recover the user data $\mathbf{\hat{x}}$ given the received vector $\mathbf{y}$ in (\ref{systemModel}) and $\mathbf{C}$, so that the mean square error $\mathbb{E}\{||\mathbf{x}-\mathbf{\hat{x}}||^2\} $ is minimized. The ML solution is given by
	 \begin{eqnarray}
	\label{MLDecoding}
	\widehat{\mathbf{x}} = \argmin_{\mathbf{x} \in \{\pm 1\}^{K\times 1}} ||\mathbf{y} - \mathbf{C}\mathbf{x} ||^2.
	\end{eqnarray}
	 It is widely recognized that obtaining the ML solution is generally NP-hard \cite{Lupas1989}. 
	
	Our detection problem, where the overloaded signature matrix has UD structure, can be solved efficiently if there is a function that maps $\mathbf{y} \mapsto \widehat{\mathbf{y}} \in \Lambda \subset \mathcal{N}^{L\times 1}$, where $\Lambda$ is a $\mathbb{Z}$-module with rank $L$. Therefore, it is equivalent to finding the closest vector point in a lattice $\Lambda$, such that 
	\begin{eqnarray}
	\label{MinDistY}
	\widehat{\mathbf{y}} = \argmin_{\mathbf{y}' \in \mathcal{N}^{L\times 1}} d_L (\mathbf{y}, \mathbf{y}').
	\end{eqnarray}
	Gaining the knowledge of $\widehat{\mathbf{y}}$, one of the points in $\Lambda$ generated by $\mathbf{C}$, we can obtain $\mathbf{\hat{x}}$ unambiguously (uniquely) by applying the NDA \cite{michel2012}, since $\mathbf{C}$ satisfies the unique decodability criteria. However, there is no known polynomial algorithm to solve for $\widehat{\mathbf{y}}$ from a given $\mathbf{y}$.
	
	Without loss of generality, we design our low-complexity decoding algorithm for code sets that are generated by the seed matrix $\mathbf{V}_8$ which is given by:
	\begin{eqnarray}
	\label{V8_2}
	{\mathbf{V}_8}=
	\renewcommand{\baselinestretch}{1.2}
	{\footnotesize \setcounter{MaxMatrixCols}{34}
		\begin{bmatrix}
		\alpha^{14}&\m1&\m\alpha&\m\alpha^{3}&\m\alpha^{3}\\
		0&\m0&\m0&\m\alpha^{3}&\m\alpha^{6}
		\end{bmatrix}.}
	\end{eqnarray}
 However, we do not necessarily imply that our proposed decoder cannot be applied to other recursive UD code sets such as in \cite{Lindstrom1964, Khachatrian1987, Khachatrian1995}. A slight modification may be required depending on a given $\mathbf{C}$ matrix. We present our proposed low-complexity FDA for $\mathbf{C}_L$, $L\in 2^i$, where $i \in \{2,3, ...\}$, which is portrayed in Table \ref{FDAalg}. To quickly summarize, the FDA estimates iteratively the number and positions of $-1$s in $\mathbf{\hat{x}}$ from the received values in $\mathbf{y}$. Those estimates are updated when the FDA compares each value of the vector $\mathbf{y}$ against quantized levels. Those quantized levels are computed based on the information from the previous rows and the current values of the vector $\mathbf{y}$.
 \vspace{-0.3cm}
	\begin{center}
		\begin{table}[H]
			\label{mergePartAlg}
			\begin{center}
				\begin{tabular}{l}
					\hline \hline \rule{0pt}{3ex} 
					\nid \textbf{Merge, $\boldsymbol{\mathsf{meP}}$ function}    \\
					\hline \rule{0pt}{3ex} 
					\nid \textbf{{Input}:} $dP', \mathbf{m}, n, K, r_c, m_{LR}, mP$ \\
					\hspace{0.3cm} 1: $dPart \gets \{\emptyset\}, \mathcal{B} \gets  mP(r_c), p' \gets 1$ \\
					\hspace{0.3cm} 2: \textbf{for} $p \gets 1$ to $len(dP')$   \\ 
					\hspace{0.3cm} 3:  \hspace{0.3cm} $A_L \gets dP'(p,3)$, $f \gets true$\\
					\hspace{0.3cm} 4:  \hspace{0.3cm} \textbf{if} $A_L \neq 0$\\
					\hspace{0.3cm} 5:  \hspace{0.6cm} $dPart(p', 1:2) \gets [ 0, A_L ]$, $\mathcal{L}_0 \gets dP'(p, 5:4+A_L)$\\
					\hspace{0.3cm} 6:  \hspace{0.6cm} $\mathcal{L}_1 \gets \mathcal{L}_0 \cap \mathcal{B}(1) $, $\mathcal{R}_1 \gets \mathcal{L}_0 \cap \mathcal{B}(2) $\\
					\hspace{0.3cm} 7:  \hspace{0.6cm} $dPart(p', 3:4+A_L) \gets [ |\mathcal{L}_1|, |\mathcal{R}_1|, \mathcal{L}_1, \mathcal{R}_1 ]$\\
					\hspace{0.3cm} 8:  \hspace{0.6cm} $r \gets r_c$, $f \gets false$\\
					\hspace{0.3cm} 9:  \hspace{0.6cm} \textbf{if} $\mathcal{L}_0 = mP(r-1,1)$\\
					\hspace{0.3cm}10:  \hspace{0.9cm} $dPart(p', 1) \gets m_{LR}(r-1,3)$, $f \gets true$\\
					\hspace{0.3cm}11:  \hspace{0.6cm} \textbf{else} \\
					\hspace{0.3cm}12:  \hspace{0.9cm} \textbf{while} $f = false$ AND $r >1$\\
					\hspace{0.3cm}13:  \hspace{1.2cm} $\mathcal{I}_L  \gets 
					\{i | \mathbf{m}(i) = -1, i \in \mathcal{L}_0\}$\\
					\hspace{0.3cm}14:  \hspace{1.2cm} \textbf{if} $\mathcal{I}_L = \emptyset$\\
					\hspace{0.3cm}15:  \hspace{1.5cm} $dPart(p', 1) \gets \sum_{i \in \mathcal{L}_0} \mathbf{m}(i)$, $f \gets true$\\
					\hspace{0.3cm}16:  \hspace{1.2cm} \textbf{else} \\
					\hspace{0.3cm}17:  \hspace{1.5cm} $\mathcal{L}_0'  \gets [1:K] \setminus \footnotemark  \mathcal{L}_0$ \\
					\hspace{0.3cm}18:  \hspace{1.5cm} $\mathcal{I}_{L'}  \gets 
					\{i | \mathbf{m}(i) = -1, i \in \mathcal{L}_0'\}$ \\
					\hspace{0.3cm}19:  \hspace{1.5cm} \textbf{if} $\mathcal{I}_{L'} = \emptyset$\\
					\hspace{0.3cm}20:  \hspace{1.8cm} $dPart(p', 1) \gets n - \sum_{i \in \mathcal{L}_0'} \mathbf{m}(i)$, $f \! \gets \! true$\\
					\hspace{0.3cm}21:  \hspace{1.5cm} \textbf{else} \\
					\hspace{0.3cm}22:  \hspace{1.8cm} $r \gets r - 1$\\
					\hspace{0.3cm}23:  \hspace{1.8cm} $\mathbf{m}  \gets \mathsf{uMissing}(m_{LR}(r, 3), mP(r,1) , \mathbf{m} )$ \\
					\hspace{0.3cm}24:  \hspace{0.6cm} \textbf{if} $f = false$ \\
					\hspace{0.3cm}25:  \hspace{0.9cm} break \\
					\hspace{0.3cm}26:  \hspace{0.6cm} $p' \gets p'+1$ \\
					\hspace{0.3cm}27:  \hspace{0.3cm} $A_R \gets dP'(p,4)$ \\
					\hspace{0.3cm}28:  \hspace{0.3cm} \textbf{if} $A_R \neq 0$ AND $f = true$\\
					\hspace{0.3cm}29:  \hspace{0.6cm} $dPart(p', 1:2) \gets [ 0, A_R ]$, \\
					\hspace{0.3cm}30: \hspace{0.6cm} $\mathcal{L}_0 \gets dP'(p, 5\!+\!A_L\!:\!4\!+\!A_L\!+\!A_R)$\\
					\hspace{0.3cm}31:  \hspace{0.6cm} $\mathcal{L}_2 \gets \mathcal{L}_0 \cap \mathcal{B}(1) $, $\mathcal{R}_2 \gets \mathcal{L}_0 \cap \mathcal{B}(2) $\\
					\hspace{0.3cm}32:  \hspace{0.6cm} $dPart(p', 3:4+A_R) \gets [ |\mathcal{L}_2|, |\mathcal{R}_2|, \mathcal{L}_2, \mathcal{R}_2 ]$\\
					\hspace{0.3cm}33:  \hspace{0.6cm} $r \gets r_c$, $f \gets false$\\
					\hspace{0.3cm}34:  \hspace{0.6cm} \textbf{if} $\mathcal{L}_0 = mP(r-1,2)$\\
					\hspace{0.3cm}35:  \hspace{0.9cm} $dPart(p', 1) \gets m_{LR}(r-1,4)$, $f \gets true$\\
					\hspace{0.3cm}36:  \hspace{0.6cm} \textbf{else} \\
					\hspace{0.3cm}37:  \hspace{0.9cm} \textbf{while} $f = false$ AND $r >1$\\
					\hspace{0.3cm}38:  \hspace{1.2cm} $\mathcal{I}_R  \gets 
					\{i | \mathbf{m}(i) = -1, i \in \mathcal{L}_0\}$\\
					\hspace{0.3cm}39:  \hspace{1.2cm} \textbf{if} $\mathcal{I}_R = \emptyset$\\
					\hspace{0.3cm}40:  \hspace{1.5cm} $dPart(p', 1) \gets \sum_{i \in \mathcal{L}_0} \mathbf{m}(i)$, $f \gets true$\\
					\hspace{0.3cm}41:  \hspace{1.2cm} \textbf{else} \\
					\hspace{0.3cm}42:  \hspace{1.5cm} $\mathcal{L}_0'  \gets [1:K] \setminus  \mathcal{L}_0$ \\
					\hspace{0.3cm}43:  \hspace{1.5cm} $\mathcal{I}_{L'}  \gets 
					\{i | \mathbf{m}(i) = -1, i \in \mathcal{L}_0'\}$ \\
					\hspace{0.3cm}44:  \hspace{1.5cm} \textbf{if} $\mathcal{I}_{L'} = \emptyset$\\
					\hspace{0.3cm}45:  \hspace{1.8cm} $dPart(p', 1) \gets n - \sum_{i \in \mathcal{L}_0'} \mathbf{m}(i)$, $f \!\gets \!true$\\
					\hspace{0.3cm}46:  \hspace{1.5cm} \textbf{else} \\
					\hspace{0.3cm}47:  \hspace{1.8cm} $r \gets r - 1$\\
					\hspace{0.3cm}48:  \hspace{1.8cm} $\mathbf{m}  \gets \mathsf{uMissing}(m_{LR}(r, 4), mP(r,2) , \mathbf{m} )$ \\
					\hspace{0.3cm}49:  \hspace{0.6cm} \textbf{if} $f = false$ \\
					\hspace{0.3cm}50:  \hspace{0.9cm} break \\
					\hspace{0.3cm}51:  \hspace{0.6cm} $p' \gets p'+1$ \\
					\hspace{0.3cm}52:  \hspace{0.0cm} \textbf{if} $f = false$ \\
					\hspace{0.3cm}53:  \hspace{0.25cm} $tP \!\! \gets \!\! [m_{LR}',\! K, m_{LR}(r_c\!-\!1,1),\! m_{LR}(r_c\!-\!1, 2),\! mP(r_c\!-\!1)] $\\
					\hspace{0.3cm}54:  \hspace{0.25cm} $[dPart,m] \gets \mathsf{meP}(tP, \mathbf{m}, n, K, r_c, m_{LR}, mP)$\\
					\nid \textbf{{Output}:} $dPart$, $\mathbf{m}$ \\
					\hline
				\end{tabular}\vspace{-0.0cm}
			\end{center}
		\end{table}
	\end{center}
	\footnotetext{$\mathcal{A} \setminus \mathcal{B} = \{x : x \in \mathcal{A} \text{   AND   } x \notin \mathcal{B} \}$.} 
 
 Therefore, the FDA attempts to find the closest lattice point that is generated by matrix $\mathbf{C}$ from the estimate $\hat{\mathbf{x}}$ to the received vector $\mathbf{y} \in \mathbb{R}^{L \times 1}$. This is achieved by performing quantization on each row of vector $\mathbf{y}$ to obtain $\mathbf{z} \in \mathcal{N}^{L\times 1}$ such that $\mathbf{z}$ is a valid lattice point. In order to demonstrate how the FDA works, it is beneficial to describe each function in detail. The quantizer $\mathsf{Q} : \mathbb{R} \mapsto \mathcal{N} $,  $z_1 = \mathsf{Q}(y, -K, K, 2)$ maps a received real value $y \in \mathbb{R}$ to one of the constellation in $\{\pm K, \pm (K-2), ...\}$ as follows,
	
	\begin{equation}
	\label{quantize}
 \mathsf{Q}(y, t^-, t^+,s) =  
  \begin{cases} 
   t^- & \!\!\!\text{if } y \leq t^-\!+\!s/2 \\
   t^-\!+\!s(i'-1) & \!\!\!\begin{array}{l}
\!\!\!\text{if }  t^-\!+\!s(i'-3/2) <y \dots\\
\!\!\!\dots  \leq t^-\!+\!s(i'-1/2)
\end{array}\\
   t^+       & \!\!\!\text{if } y > t^+\!-\!s/2
  \end{cases}
\end{equation}

\noindent where $t^-$, $t^+$ and $s$ are the input parameters for the minimum, maximum, step-size values and $i'\in \mathbb{N}$ is the internal value that decides the quantization level, respectively. Furthermore, let integer $n$ and vector $\mathbf{m}$ denote the number of $-1$s and locations, where the $-1$s occur in $\mathbf{\hat{x}}$, respectively. Mathematically, $n$ and $\mathbf{m}$ can be expressed as
	\begin{equation}
	n = (K-\mathbf{\hat{x}}^T\boldsymbol{1} )/2,
	\end{equation}
	and
	\begin{equation}
	\mathbf{m} = (\boldsymbol{1}-\mathbf{\hat{x}})/2.
	\end{equation}
\vspace{-0.1cm}
	\begin{center}
		\begin{table}[h]
			\vspace{-0.5cm} \caption{Fast Decoder Algorithm (FDA)}
			\label{FDAalg}
			\begin{center}
				\begin{tabular}{l}
					\hline \hline \rule{0pt}{3ex} 
					\nid \textbf{FDA}  \\
					\hline \rule{0pt}{3ex} 
					\nid \textbf{{Input}:} $\mathbf{y}$, $K$ \\
					\hspace{0.3cm} 1: $z_1 \gets \mathsf{Q}(y_1, -K,K,2)$ \\
					\hspace{0.3cm} 2: \textbf{if} $|z_1| = K$,   $\mathbf{\hat{x}}\gets \mathsf{sgn}(z_1)\boldsymbol{1}_K$\\ 
					\hspace{0.3cm} 3:  \textbf{else}\\
					\hspace{0.3cm} 4:  \hspace{0.3cm} $\mathbf{m} \gets -\boldsymbol{1}_K$, $r_c \gets 1$, $n \gets (K-z_1)/2$\\
					\hspace{0.3cm} 5:  \hspace{0.3cm} $m_{LR}(r_c,3) \gets n$\\
					\hspace{0.3cm} 6:  \hspace{0.3cm} $dP(r_c) \gets [n, K, m_{LR}(r_c,1), m_{LR}(r_c, 2), mP(r_c)]$\\
					\hspace{0.3cm} 7:  \hspace{0.3cm} $c_{AL} \gets \boldsymbol{0}$, $\mathbf{z} \gets \boldsymbol{0}$, $s_{I}\gets 1$, $c_{T} \gets 1$\\
					\hspace{0.3cm} 8:  \hspace{0.3cm} \textbf{while}  $s_{I}=1$ AND $c_{T}< N_c$\\
					\hspace{0.3cm} 9:  \hspace{0.6cm} $s_{I}\gets 0$\\
					\hspace{0.3cm}10:  \hspace{0.6cm} \textbf{while} $r_c < L$, $r_c \gets r_c +1$\\
					\hspace{0.3cm}11:  \hspace{0.9cm} $[dP(r_c),m] \!\! \gets \!\! \mathsf{meP}(dP(r_c\!\!-\!\!1), \mathbf{m}, \! n,\! K, r_c, m_{LR}, mP)$\\
					\hspace{0.3cm}12:  \hspace{0.9cm} $A^- \gets \mathsf{minT}(dP(r_c))$, $A^+ \gets \mathsf{maxT}(dP(r_c))$\\
					\hspace{0.3cm}13:  \hspace{0.9cm} $\mathbf{z}(r_c) \gets \mathsf{Q}(y', A^-, A^+,4)$\\
					\hspace{0.3cm}14:  \hspace{0.9cm} $m_{LR}(r_c,3) \gets (2n-\mathbf{z}(r_c)-$\\
					\hspace{0.3cm}14:  \hspace{0.89cm} $(m_{LR}(r_c,2) - m_{LR}(r_c,1)))/4$\\
					\hspace{0.3cm}15:  \hspace{0.9cm} $m_{LR}(r_c,4) \gets n- m_{LR}(r_c,3)$ \\
					\hspace{0.3cm}16:  \hspace{0.9cm} $\mathbf{m} \gets \mathsf{uM}(\mathbf{m}, m_{LR}, r_c, mP)$ \\
					\hspace{0.3cm}17:  \hspace{0.6cm} $\mathbf{m} \gets \mathsf{f_c}(\mathbf{m}, m_{LR}, n)$, $\mathbf{t}_d \gets \mathbf{z} - \mathbf{C}(-2\mathbf{m} + \boldsymbol{1})$\\
					\hspace{0.3cm}18:  \hspace{0.6cm}  \textbf{if} $\mathbf{t}_d \notin \boldsymbol{0}$,   $s_{I}\gets 1$, $r_c \gets i_d$\footnotemark   \\
					\hspace{0.3cm}19:  \hspace{0.9cm} $c_{AL}(r_c+1) \gets  c_{AL}(r_c+1)+1$ \\ 
					\hspace{0.3cm}20:  \hspace{0.6cm} $c_{T} \gets c_T +1$ \\
					\hspace{0.3cm}21:  \hspace{0.3cm} $\mathbf{\hat{x}} \gets -2\mathbf{m} + \boldsymbol{1}$  \\
					\nid \textbf{{Output}:} $\mathbf{\hat{x}}$ \\
					\hline
				\end{tabular}\vspace{-0.1cm}
			\end{center}
		\end{table}
	\end{center}
	\footnotetext{$i_d$ is the lowest index that $\mathbf{t}_d(i_d) \neq 0$.} 
	For the trivial case when $z_1 = K$ or $z_1 = -K$ the algorithm outputs the decision vector $\mathbf{\hat{x}}$ without proceeding to the next steps. Otherwise, it will initialize the vector $\mathbf{m}\gets \mathbf{1}_K$, index $r_c \gets 1$ and $n \gets (K-z_1)/2$. Then the value of $n$ is recorded in the table $m_{LR}$ at the row $r_c = 1$ and the column $n_L$. The table $m_{LR}$ keeps track of $l_{L}$, $l_{R}$, $n_L$ and $n_R$ values for each row and $l_{L}$ and $l_{R}$ are defined as the number of $+1$s and $-1$s of the row of code set $\mathbf{C}$, $n_L$ and $n_R$ are the number of $-1$s of the estimated $\mathbf{\hat{x}}$ that corresponds to locations of $+1$ and $-1$ of each row of the code set.

	For example, $m_{LR}$ table for the code set in Fig. \ref{C_4_5_01} can be constructed as
	\begin{table}[H]
  \begin{center}
    \label{tab:tableMLR}
    \begin{tabular}{c|c|c|c}
      ${l_{L}}$ &  $l_{R}$ & $n_{L}$ & $n_{R}$\\  
       \hline \rule{0pt}{2ex}
      5 & 0 & $n$ & .\\ 
      2 & 3 & . & .\\ 
      3 & 2 & . & .\\ 
      3 & 2 & . & .\\ 
    \end{tabular}
  \end{center}
\end{table}
	\noindent where the last two columns presented as dots are filled up in the further steps of the algorithm. Whereas $mP(r_c)$ is the actual column indices of $+1$s and $-1$s in the row $r_c$. In the case of code set in Fig. \ref{C_4_5_01}, the matrix $mP$ is defined as
	
	\begin{table}[h!]
  \begin{center}
    \label{tab:tableMPART}
    \begin{tabular}{c|c|c}
      ${r_{c}}$ &  $index_{L}$ & $index_{R}$   \\  
       \hline \rule{0pt}{2ex}
      1 & \{$\emptyset$\} & \{1,2,3,4,5\}   \\ 
      2 &  \{2,4,5\} &  \{1,3\} \\ 
      3 & \{3,4\} & \{1,2,5\}  \\ 
      4 & \{2,3\} & \{1,4,5\}  \\ 
    \end{tabular}
  \end{center}
\end{table}

	The algorithm proceeds by partitioning each row $r_c$ and saving the estimated number of $-1$s of vector $\mathbf{\hat{x}}$, $n'$, the partition size, $K'$, the number of $+1$s and $-1$s, $l_{L'}$ and $l_{R'}$ of the specific partition in $dP(r_c)$. In Step $7$, the adaptive parameter $c_{AL}$, stopping iteration flag $s_I$, and the number of repetition count $c_T$ (e.g., algorithm can repeat steps from $10$ to $16$ up to a maximum of $N_c$ times, since it directly depends on the variance of the noise) are initialized. At each row $r_c$, $dP(r_c)$ gets updated by calling $\mathsf{meP}()$ function. The function $\mathsf{meP}(dP(r_c-1), \mathbf{m}, n, K, r_c, m_{LR}, mP)$ scans each partition of the row $r_c$ with updated values and if it finds one or more partitions completely identified the exact locations of $-1$s, hence will skip partitioning further. $A^-$ and $A^+$ are the minimum and maximum values calculated for a given partitions at each row.
	\vspace{-0.5cm}
	\begin{center}
		\begin{table}[H]
			\label{minT}
			\begin{center}
				\begin{tabular}{l}
					\hline \hline \rule{0pt}{3ex} 
					\nid \textsf{\textbf{MinT} function}  \\
					\hline \rule{0pt}{3ex} 
					\nid \textbf{{Input}:} $dP$ \\
					\hspace{0.3cm} 1: $A^- \gets 0$ \\
					\hspace{0.3cm} 2: \textbf{for} $i \gets 1$ to $\mathsf{len}(dP)$\\
					\hspace{0.3cm} 3:  \hspace{0.3cm} $A^- \gets A^- + \mathsf{minF}(dP(i,1), dP(i,3), dP(i,2))$\\
					\nid \textbf{{Output}:} $A^-$ \\
					\hline
				\end{tabular}\vspace{-0.0cm}
			\end{center}
		\end{table}
	\end{center}
	\vspace{-0.8cm}
	\begin{center}
		\begin{table}[H]
			\label{maxT}
			\begin{center}
				\begin{tabular}{l}
					\hline \hline \rule{0pt}{3ex} 
					\nid \textsf{\textbf{MaxT} function}  \\
					\hline \rule{0pt}{3ex} 
					\nid \textbf{{Input}:} $dP$ \\
					\hspace{0.3cm} 1: $A^+ \gets 0$ \\
					\hspace{0.3cm} 2: \textbf{for} $i \gets 1$ to $\mathsf{len}(dP)$\\ 
					\hspace{0.3cm} 3:  \hspace{0.3cm} $A^+ \gets A^+ + \mathsf{maxF}(dP(i,1), dP(i,4), dP(i,2))$\\
					\nid \textbf{{Output}:} $A^+$ \\
					\hline
				\end{tabular}\vspace{-0.0cm}
			\end{center}
		\end{table}
	\end{center}
	\noindent where $\mathsf{minF}(n, L, K) = 2|n-L| -K$ and $\mathsf{maxF}(n, R, K) = -2|n-R| +K$. In line $13$ of FDA, we define $y'= y(r_c)+2\mathsf{sgn}(y(r_c)-z(r_c))c_{AL}(r_c)$, where $c_{AL}$ is an integer vector that is incremented in Step $19$ by one if estimated $\mathbf{z}$ is not one of the lattice vertices generated by $\mathbf{C}$. In line $18$, we scan the rows from $1$ to $L$ to find the first $r_c$, where $\mathbf{z}$ differ from estimated lattice vertex. The function $\mathsf{uM}(\mathbf{m}, m_{LR}, r_c, mP)$ updates $\mathbf{m}$ with the given updated parameters as follows
	
	\vspace{-0.0cm}
	\begin{center}
		\begin{table}[H]
			\label{updateM}
			\begin{center}
				\begin{tabular}{l}
					\hline \hline \rule{0pt}{3ex} 
					\nid \textsf{\textbf{uM} function}  \\
					\hline \rule{0pt}{3ex} 
					\nid \textbf{{Input}:} $\mathbf{m}, m_{LR}, r, mP$ \\
					\hspace{0.3cm} 1: \textbf{if} $m_{LR}(r, 3) = 0$ \\
					\hspace{0.3cm} 2: \hspace{0.3cm} $\mathbf{m}(mP(r,1)) \gets 0$ \\
					\hspace{0.3cm} 3: \textbf{elseif} $m_{LR}(r, 3) = m_{LR}(r, 1)$ \\
					\hspace{0.3cm} 4: \hspace{0.3cm} $\mathbf{m}(mP(r,1)) \gets 1$ \\
					\hspace{0.3cm} 5: \textbf{else} \\
					\hspace{0.3cm} 6: \hspace{0.3cm} $\mathbf{m} \gets \mathsf{uMissing}(m_{LR}(r, 3), mP(r,1), \mathbf{m}) $ \\
					\hspace{0.3cm} 7: \textbf{if} $m_{LR}(r, 4) = 0$ \\
					\hspace{0.3cm} 8: \hspace{0.3cm} $\mathbf{m}(mP(r,2)) \gets 0$ \\
					\hspace{0.3cm} 9: \textbf{elseif} $m_{LR}(r, 4) = m_{LR}(r, 2)$ \\
					\hspace{0.3cm}10: \hspace{0.3cm} $\mathbf{m}(mP(r,2)) \gets 1$ \\
					\hspace{0.3cm}11: \textbf{else} \\
					\hspace{0.3cm}12: \hspace{0.3cm} $\mathbf{m} \gets \mathsf{uMissing}(m_{LR}(r, 4), mP(r,2), \mathbf{m}) $ \\
					\nid \textbf{{Output}:} $\mathbf{m}$ \\
					\hline
				\end{tabular}\vspace{-0.0cm}
			\end{center}
		\end{table}
	\end{center}
	\noindent where the function $\mathsf{uMissing}$ in lines $6$ and $12$ of $\mathsf{uM}$ are defined as
	\vspace{-0.0cm}
	\begin{center}
		\begin{table}[H]
			\label{updateMissing}
			\begin{center}
				\begin{tabular}{l}
					\hline \hline \rule{0pt}{3ex} 
					\nid \textsf{\textbf{uMissing function}}  \\
					\hline \rule{0pt}{3ex} 
					\nid \textbf{{Input}:} $n'$, $\mathcal{A}$, $\mathbf{m}$ \\
					\hspace{0.3cm} 1: $\mathcal{I}_L \gets \{i | \mathbf{m}(i) = -1, i \in \mathcal{A} \}$ \\
					\hspace{0.3cm} 2: \textbf{if} $|\mathcal{I}_L | = 1$ \\ 
					\hspace{0.3cm} 3:  \hspace{0.3cm} $\mathbf{m}(\mathcal{I}_L) \gets n' - \sum_{i \neq \mathcal{I}_L}\mathbf{m}(i) $\\
					\nid \textbf{{Output}:} $\mathbf{m}$ \\
					\hline
				\end{tabular}\vspace{-0.0cm}
			\end{center}
		\end{table}
	\end{center}
	In line $17$ of the FDA, the function $\mathsf{f_c}(\mathbf{m}, m_{LR}, n)$ finds all the locations of $-1$s in $\mathbf{m}$ based on the UD structure of $\mathbf{C}$ codes. As an example, in case of code set presented in Fig. \ref{C_4_5_01}, the function $\mathsf{f_c}$ is expressed as
	\vspace{-0.0cm}
	\begin{center}
		\begin{table}[H]
			\label{fcAlgorithm}
			\begin{center}
				\begin{tabular}{l}
					\hline \hline \rule{0pt}{3ex} 
					\nid $\boldsymbol{\mathsf{f_c}}$ \textsf{\textbf{function,}} $\boldsymbol{\mathsf{L=4}}$   \\
					\hline \rule{0pt}{3ex} 
					\nid \textbf{{Input}:} $\mathbf{m}$, $m_{LR}$, $n$\\
					\hspace{0.3cm} 1: $n' \gets \sum_{i \in  \mathcal{A}_5}m_{LR}(i, 4)$ \\
					\hspace{0.3cm} 2: $\mathbf{m}(5) \gets \mathsf{mod}(n',2) $ \\ 
					\hspace{0.3cm} 3: $m' \gets (n'-\mathbf{m}(5))/2 $ \\
					\hspace{0.3cm} 4: $\mathbf{m}(1) \gets n-m'-\mathbf{m}(5) $ \\
					\hspace{0.3cm} 5: $\mathbf{m}(2) \gets m_{LR}(3,3)-\mathbf{m}(1)-\mathbf{m}(5) $ \\
					\hspace{0.3cm} 6: $\mathbf{m}(3) \gets m_{LR}(2,3)-\mathbf{m}(1) $ \\
					\hspace{0.3cm} 7: $\mathbf{m}(4) \gets m_{LR}(2,4)-\mathbf{m}(2)-\mathbf{m}(5) $ \\
					\nid \textbf{{Output}:} $\mathbf{m}$ \\
					\hline
				\end{tabular}\vspace{-0.0cm}
			\end{center}
		\end{table}
	\end{center}
	\noindent where $\mathcal{A}_5 = \{2, 3, 4\}$. In case the information in $\mathbf{m}$ does not correspond to $\mathbf{z}$, which is verified in line $18$ of FDA it then sets  $r_c$ to the row, where the discrepancy happened and re-runs from Step $10$ again until it finds $\mathbf{m}$ that correspondence to $\mathbf{z}$. 
	
	{ For the case of Rayleigh fading channel instead of AWGN in (\ref{systemModel}) the proposed FDA is still applicable to perform detection. However, the channel coefficients for each user $k$ should be known at the receiver side. For the frequency-selective fading channels, we can employ transmitter precoding scheme to overcome multipath channel effect as proposed by Fantuz and D'Amours, which is detailed in \cite{Fantuz2019}. Briefly, this transmit precoding scheme exploits the knowledge of the channel impulse response for transforming the multipath channel into a single-path non-dispersive channel, which is equivalent over non-dispersive Rayleigh fading channel model.}
	
	\section{{Complexity Analysis}}
	\label{performanceAnalysis}
	
	In this section, we discuss the complexity analysis of the proposed NDA and FDA algorithms. The NDA decoder for the noiseless transmission channels, discussed in \cite{michel2012}, deciphers the data of all users at the receiver side in a recursive manner. At each step, it performs additions, comparisons and multiplications to decipher the bits of the users. 
	\vspace{-0.5cm}
		\begin{figure}[H]
			\hspace{-0.4cm}
			\includegraphics[width=3.8 in]{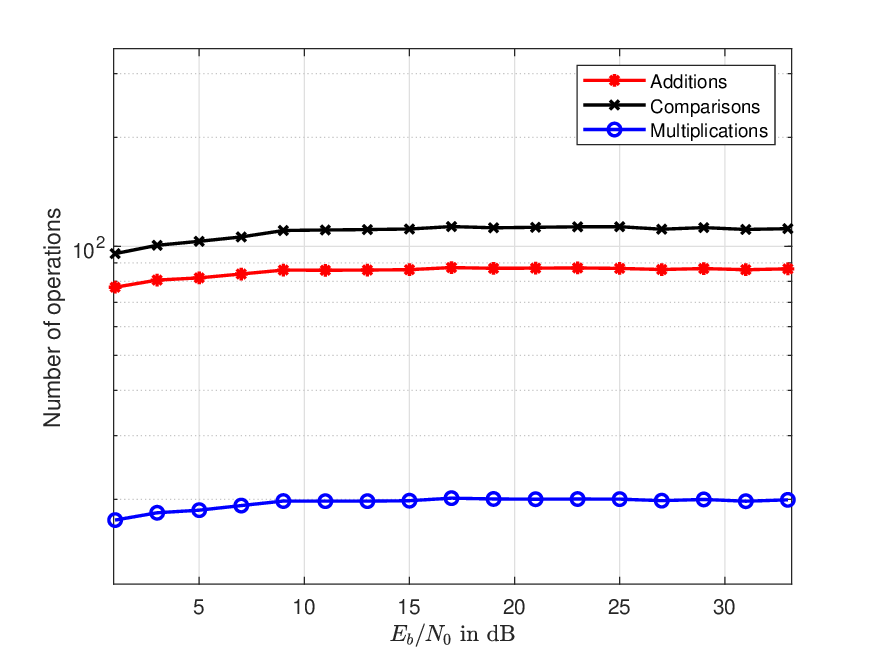}
           \caption{Complexity for $L=4$.} \label{complexity_4by5}
		\end{figure}
	\vspace{-0.5cm}
	The NDA is deterministic with an exact number of execution steps. After $L/8$ number of execution steps it calls the NDA algorithm recursively, using two smaller vectors composed of the upper and lower $L/2$ elements of the received vector. For that reason to compute the complexity of the algorithm we first break the algorithm into two blocks $B_1$ and $B_2$ from steps $1$ to $8$ and $9$ to $12$. 
	\vspace{-0.5cm}
		\begin{figure}[H]
		\hspace{-0.4cm}
			\includegraphics[width=3.8 in]{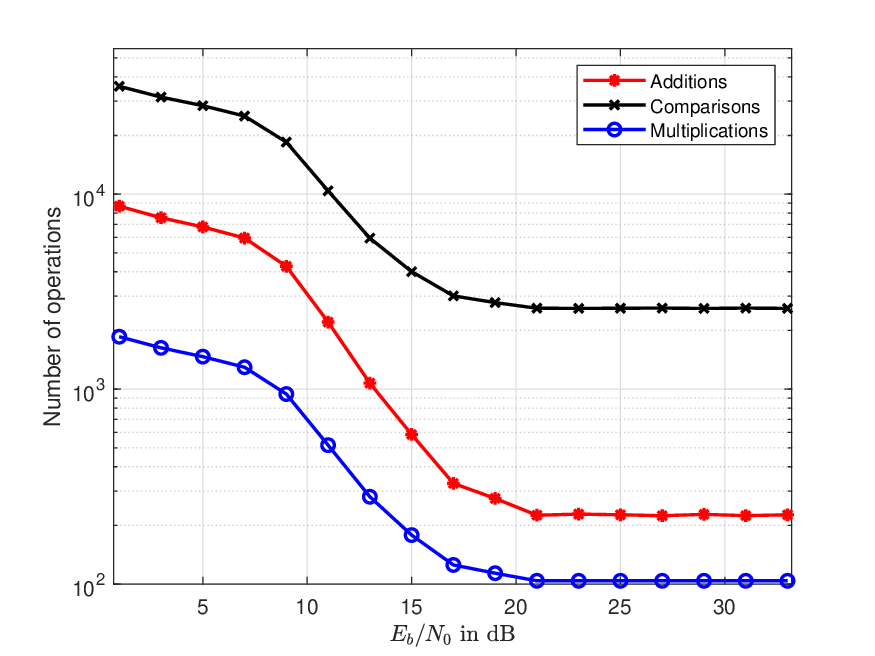}
           \caption{Complexity for $L=8$.}  \label{complexity_8by13}
		\end{figure}
	\vspace{-0.5cm}
	We denote $N_1$ and $N_2$ to represent the total number of additions, comparisons and multiplications ($N^i_{add}$, $N^i_{comp}$, $N^i_{mult}$) of block $B_i$ at each recursive calls. As an example $N_1 = (N^1_{add}, N^1_{comp}, N^1_{mult})$, where $N^1_{add} = 10$, $N^1_{comp} = 3$ and $N^1_{mult} = 0$ since there is $3$, $3$ and $4$ additions in steps $2$, $4$ and $6$, and $1$ comparison in steps $3$, $5$ and $8$, respectively. Similarly, when $L=16$ the $N_2 = (N^2_{add}, N^2_{comp}, N^2_{mult})$, where $N^2_{add} = (3+L)L/2+4=156$, $N^2_{comp} = 1$ and $N^2_{mult} = L^2/2=128$, respectively. For the trivial case when $L =4$ we have $N_4 = (N^4_{add}, N^4_{comp}, N^4_{mult})$, where $N^4_{add} = 4$, $N^4_{comp} = 4$ and $N^4_{mult} = 4$.
	\vspace{-0.5cm}
		\begin{figure}[H]
			\hspace{-0.4cm}
	\includegraphics[width=3.8in]{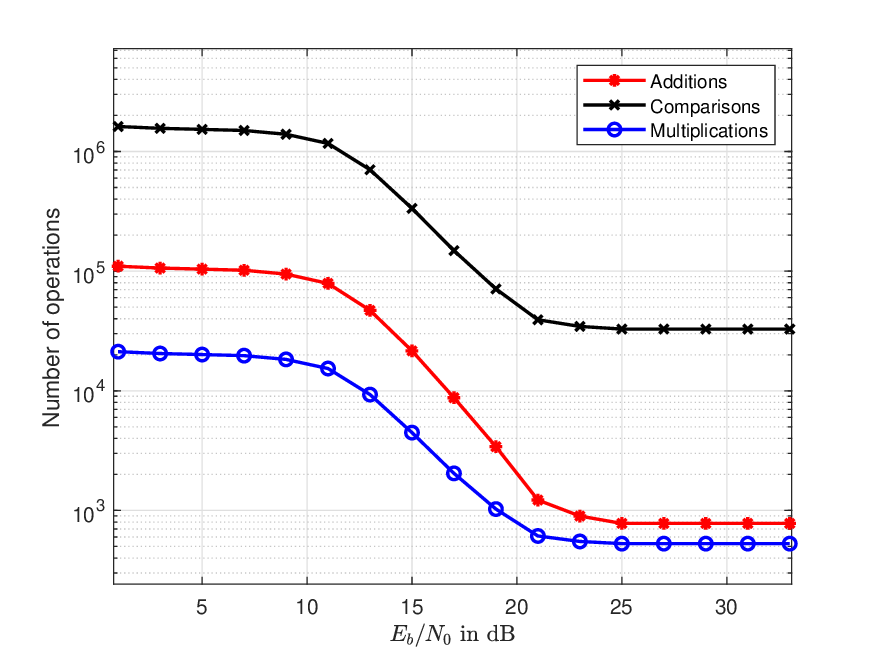}
           \caption{Complexity for $L=16$.} 		\label{complexity_16by33}
		\end{figure}
	\vspace{-0.5cm}
	The complexity when $L=8$ is $N_1 + 2N_4$ since the NDA will not execute block $B_2$. In the case when $L=16$, the complexity of NDA is $4N_1 + N_2+ 4N_4$. It follows that when $L>8$ the complexity of the recursive NDA can be represented as such
 \begin{equation}
     (i-2)2^{(i-3)}N_1 + ((i-4)2^{(i-3)}+1)N_2+ 2^{(i-2)}N_4,
 \end{equation}
 \noindent where $L = 2^i$ for $i \in \{4, 5, ... \}$. Based on this calculation, we can conclude that the complexity of NDA algorithm is $\mathcal{O}(L\:\mathsf{log}_2(L))$.
 \begin{table}[H]
	\caption{Complexity comparison of the Detectors} 
	\centering 
	\begin{tabular}{l c c c c} 
		\hline\hline \rule{0pt}{3ex}  
		\bf{Algorithms} & \bf{Complexity}  & \bf{Main procedures}\\ [0.5ex]
		\hline \rule{0pt}{3ex}  
		 NDA & $\mathcal{O}(L\:\mathsf{log}_2(L))$ &Comparisons \\[1ex]
		\:\:FDA & $\mathcal{O}(LK\:\mathsf{log}_2(K))$ &Comparisons \\[1ex]
		\:\:MMSE-PIC & $\mathcal{O}(LK^2)$ &Inversion + multiplication \\[1ex]
		\:\:Slab-sphere & $\mathcal{O}(LK^2)$ & Multiplication + comparisons \\[1ex]
		\:\:PDA &$\mathcal{O}(L^2K^2)$ & Inversion + multiplication \\[1ex]
		\:\:ML &$\mathcal{O}(2^K)$ & Multiplication + addition \\[1ex]
		\hline 
	\end{tabular}
	\label{table:complexity_bigO}
\end{table}
 As one would normally expect, the complexity of the decoder in noisy channels is much higher than in the noiseless channels. However, the complexity of the proposed FDA decoder in Section \ref{fastDecoder} is not more complex than the NDA in terms of big $\mathcal{O}$ notation. It is important to state that the proposed FDA requires neither matrix inversion nor decomposition, instead, only additions, comparisons and multiplications. The algorithm goes through each row of the received vector to decode one or more users. The best case scenario of our FDA would be to satisfy the condition in Step $2$ with the complexity of maximum of $K$ comparisons, $N_{comp} = K$. We should note that unlike the deterministic NDA, FDA has an element of probability in the execution steps. Therefore, to compute the overall complexity of the FDA we need to consider the worst case scenario. In Step $11$ the complexity of $\mathsf{meP}$ can be shown to have $6K$ additions and $(1+ \mathsf{log}_2(K))6K +2$ comparisons. The complexity of steps $12$ and $13$ is $6$ additions and $K$ comparisons. Similarly, the complexity of steps $14$ and $15$ is $4$ additions. The complexity of $\mathsf{uM}$ method is $2K$ additions and $2K+6$ comparisons. Those execution steps from $11$ to $16$ are repeated $L$ times. Finally, in Step $17$ the complexity is $(2+K)L + \mathcal{O}(K)$ additions and $LK$ multiplications. Considering all the components, the total complexity can be shown to be $N_{add} = 9LK+12L + \mathcal{O}(LK)$, $N_{comp} = 6LK\mathsf{log}_2(K)+9LK +8L$ and $N_{mult} = LK$. Therefore, we can conclude by looking at the higher order terms of comparisons, $N_{comp}$, since it is higher than additions and multiplications, hence the overall complexity is $\mathcal{O}(LK \:\mathsf{log}_2(K))$. 
We note that in both NDA and FDA algorithms we do not consider assignments in our complexity computations. We can see that the complexity of FDA is comparably larger than NDA but much lower than the MMSE-PIC, Slab-sphere, PDA and ML decoders, which have complexities of $\mathcal{O}(LK^2)$, $\mathcal{O}(LK^2)$, $\mathcal{O}(L^2K^2)$, and $\mathcal{O}(2^K)$, respectively, as shown in Table \ref{table:complexity_bigO}. The exact computation for the cases $L=4$, $L = 8$ and $L=16$ are shown in Table \ref{table:complexity_numbers}.
 \begin{table}[H]
	\caption{Complexity Of Decoders} 
	\centering 
	\begin{tabular}{l l c c c} 
		\hline\hline \rule{0pt}{3ex}  
		\bf{Decoder} & \bf{Complexity} & ($\mathbf{4\times5}$) & ($\mathbf{8\times13}$) & ($\mathbf{16\times33}$) \\ [0.5ex]
		\hline \rule{0pt}{3ex}  
		\multirow{2}*{FDA} & Additions & $248$ &  $1,136$  & $5,472$ \\
        & Comparisons & $491$ &  $3,309$  & $20,861$ \\
        & Multiplications & $20$ &  $104$  & $528$ \\ [1ex]
		\:\:PDA & Inver. + add. & $400$  & $10,816$  & $278,784$\\[1ex]
		\:\:ML & Mult. + add. & $2^5$ & $2^{13}$ & $2^{33}$ \\[1ex]
		\hline 
	\end{tabular}
	\label{table:complexity_numbers}
\end{table}
\begin{figure*}[!b]
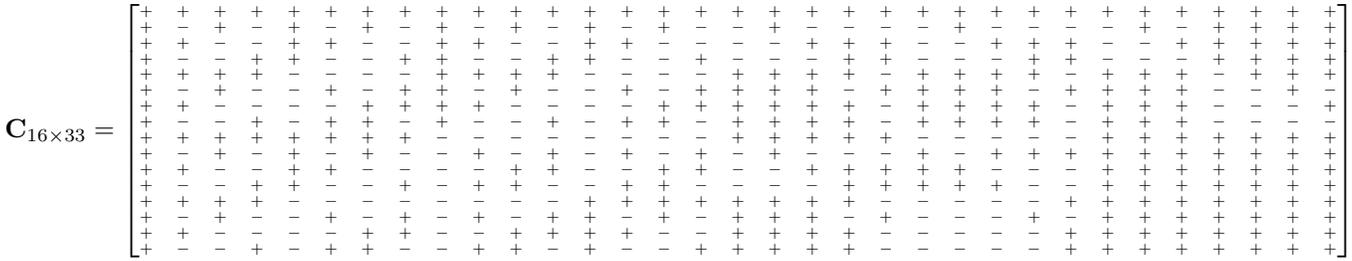

		\begin{eqnarray}
		{\bf C}_{16\times 33}=
		\renewcommand{\baselinestretch}{1}
		{\tiny \setcounter{MaxMatrixCols}{34}
			\begin{bmatrix}
			+&\m+&\m+&\m+&\m+&\m+ &\m+&\m+&\m+&\m+ &\m+&\m+&\m+&\m+&\m+&\m+&\m+&\m+&\m+&\m+&\m+&\m+&\m+&\m+&\m+&\m+&\m+&\m+&\m+&\m+&\m+&\m+&\m+\\
			+&\m-&\m+&\m-&\m+&\m- &\m+&\m-&\m+&\m- &\m+&\m-&\m+&\m-&\m+&\m-&\m-&\m+&\m-&\m-&\m-&\m-&\m+&\m-&\m-&\m-&\m-&\m+&\m-&\m+&\m+&\m+&\m+\\
			+&\m+&\m-&\m-&\m+&\m+ &\m-&\m-&\m+&\m+ &\m-&\m-&\m+&\m+&\m-&\m-&\m-&\m-&\m+&\m+&\m+&\m-&\m-&\m+&\m+&\m+&\m-&\m-&\m+&\m+&\m+&\m+&\m+\\
			+&\m-&\m-&\m+&\m+&\m- &\m-&\m+&\m+&\m- &\m-&\m+&\m+&\m-&\m-&\m+&\m-&\m-&\m-&\m+&\m+&\m-&\m-&\m-&\m+&\m+&\m-&\m-&\m-&\m+&\m+&\m+&\m+\\
			+&\m+&\m+&\m+&\m-&\m- &\m-&\m-&\m+&\m+ &\m+&\m+&\m-&\m-&\m-&\m-&\m+&\m+&\m+&\m+&\m-&\m+&\m+&\m+&\m+&\m-&\m+&\m+&\m+&\m-&\m+&\m+&\m+\\
			+&\m-&\m+&\m-&\m-&\m+ &\m-&\m+&\m+&\m- &\m+&\m-&\m-&\m+&\m-&\m+&\m+&\m+&\m+&\m-&\m+&\m+&\m+&\m+&\m-&\m+&\m+&\m+&\m+&\m-&\m-&\m+&\m-\\
			+&\m+&\m-&\m-&\m-&\m- &\m+&\m+&\m+&\m+ &\m-&\m-&\m-&\m-&\m+&\m+&\m+&\m+&\m+&\m+&\m-&\m+&\m+&\m+&\m+&\m-&\m+&\m+&\m+&\m-&\m-&\m-&\m+\\
			+&\m-&\m-&\m+&\m-&\m+ &\m+&\m-&\m+&\m- &\m-&\m+&\m-&\m+&\m+&\m-&\m+&\m+&\m+&\m+&\m-&\m+&\m+&\m+&\m+&\m-&\m+&\m+&\m+&\m-&\m-&\m-&\m-\\
			+&\m+&\m+&\m+&\m+&\m+ &\m+&\m+&\m-&\m- &\m-&\m-&\m-&\m-&\m-&\m-&\m+&\m+&\m+&\m+&\m+&\m-&\m-&\m-&\m-&\m-&\m+&\m+&\m+&\m+&\m+&\m+&\m+\\
			+&\m-&\m+&\m-&\m+&\m- &\m+&\m-&\m-&\m+ &\m-&\m+&\m-&\m+&\m-&\m+&\m-&\m+&\m-&\m-&\m-&\m+&\m-&\m+&\m+&\m+&\m+&\m+&\m+&\m+&\m+&\m+&\m+\\
			+&\m+&\m-&\m-&\m+&\m+ &\m-&\m-&\m-&\m- &\m+&\m+&\m-&\m-&\m+&\m+&\m-&\m-&\m+&\m+&\m+&\m+&\m+&\m-&\m-&\m-&\m+&\m+&\m+&\m+&\m+&\m+&\m+\\
			+&\m-&\m-&\m+&\m+&\m- &\m-&\m+&\m-&\m+ &\m+&\m-&\m-&\m+&\m+&\m-&\m-&\m-&\m-&\m+&\m+&\m+&\m+&\m+&\m-&\m-&\m+&\m+&\m+&\m+&\m+&\m+&\m+\\
			+&\m+&\m+&\m+&\m-&\m- &\m-&\m-&\m-&\m- &\m-&\m-&\m+&\m+&\m+&\m+&\m+&\m+&\m+&\m+&\m-&\m-&\m-&\m-&\m-&\m+&\m+&\m+&\m+&\m+&\m+&\m+&\m+\\
			+&\m-&\m+&\m-&\m-&\m+ &\m-&\m+&\m-&\m+ &\m-&\m+&\m+&\m-&\m+&\m-&\m+&\m+&\m+&\m-&\m+&\m-&\m-&\m-&\m+&\m-&\m+&\m+&\m+&\m+&\m+&\m+&\m+\\
			+&\m+&\m-&\m-&\m-&\m- &\m+&\m+&\m-&\m- &\m+&\m+&\m+&\m+&\m-&\m-&\m+&\m+&\m+&\m+&\m-&\m-&\m-&\m-&\m-&\m+&\m+&\m+&\m+&\m+&\m+&\m+&\m+\\
			+&\m-&\m-&\m+&\m-&\m+ &\m+&\m-&\m-&\m+ &\m+&\m-&\m+&\m-&\m-&\m+&\m+&\m+&\m+&\m+&\m-&\m-&\m-&\m-&\m-&\m+&\m+&\m+&\m+&\m+&\m+&\m+&\m+\\
			\end{bmatrix}}
		\nonumber
		\end{eqnarray}
		\caption{UD code set $\mathbf{C}$ with $L=16$ and $K=33$.} \label{C_16_33_02}
	\end{figure*}
 The FDA verifies in Step $18$ if the $L$-dimensional lattice point that is generated by the estimated vector $\mathbf{m}$ is the same as the vector $\mathbf{z}$ or not. If that condition is not true the algorithm will adjust $c_{AL}$ value by increasing it by one and starts repeating steps from $11$ to $16$.
 It is obvious that the FDA will make an exact number of execution steps if no noise vector, $\mathbf{n}$, is present and the number of times it will repeat the steps from $11$ to $16$ will only depend on the variance of noise, $\sigma^2$. 

 Therefore, we demonstrate the number of additions, comparisons and multiplications in Figs. \ref{complexity_4by5}, \ref{complexity_8by13} and \ref{complexity_16by33} by varying $\sigma^2$ in terms of $E_b/N_0$ in our simulations. For the case of $L=4$ in Fig. \ref{complexity_4by5}, the number of additions, comparisons and multiplications does not depend on variance due to small overload factor $5/4$. However, for $L=8$ and $L=16$ it stays high for up to $10$dB in $E_b/N_0$. Then it drops to constant numbers at around $20$dB in $E_b/N_0$, which is shown in Figs. \ref{complexity_8by13} and \ref{complexity_16by33}, respectively.
 \section{Simulation results}
	\label{simulation}
	 In this section, we evaluate the performance of the proposed antipodal UD code sequences generated by the seed matrix (\ref{V8_2}), which are portrayed in Figs. \ref{C_4_5_01}, \ref{C_8_13_02} and \ref{C_16_33_02}.
	 In our simulations, we compare the FDA with the MMSE-PIC \cite{MichelHanzo2021}, slab-sphere \cite{KaiKit01}, PDA \cite{Romano2005} and ML detectors for binary phase-shift keying (\acs{BPSK}) modulation, which are characterized in Figs. \ref{C4by5}, \ref{C8by13} and \ref{C16by33}.
	\begin{figure}[H]
		\centering
		\begin{center}
			\begin{eqnarray}
			{\bf C}_{4\times 5}=
			\renewcommand{\baselinestretch}{1}
			{\tiny \setcounter{MaxMatrixCols}{34}
				\begin{bmatrix}
				+&\m+&\m+&\m+&\m+\\
				+&\m-&\m+&\m-&\m-\\
				+&\m+&\m-&\m-&\m+\\
				+&\m-&\m-&\m+&\m+\\
				\end{bmatrix}}
			\nonumber
			\end{eqnarray}
			\caption{UD code set $\mathbf{C}$ with $L=4$ and $K=5$.} \label{C_4_5_01}
		\end{center}
	\end{figure}
	\vspace{-0.0cm}
		\vspace{-0.0cm}
	\begin{figure}[H]
		\begin{eqnarray}
		{\bf C}_{8\times 13}=
		\renewcommand{\baselinestretch}{1}
		{\tiny \setcounter{MaxMatrixCols}{34}
			\begin{bmatrix}
			+&\m+&\m+&\m+&\m+&\m+&\m+&\m+&\m+&\m+&\m+&\m+&\m+\\
			+&\m-&\m+&\m-&\m+&\m-&\m+&\m-&\m-&\m+&\m-&\m-&\m-\\
			+&\m+&\m-&\m-&\m+&\m+&\m-&\m-&\m-&\m-&\m+&\m+&\m+\\
			+&\m-&\m-&\m+&\m+&\m-&\m-&\m+&\m-&\m-&\m-&\m+&\m+\\
			+&\m+&\m+&\m+&\m-&\m-&\m-&\m-&\m+&\m+&\m+&\m+&\m-\\
			+&\m-&\m+&\m-&\m-&\m+&\m-&\m+&\m+&\m+&\m+&\m-&\m+\\
			+&\m+&\m-&\m-&\m-&\m-&\m+&\m+&\m+&\m+&\m+&\m+&\m-\\
			+&\m-&\m-&\m+&\m-&\m+&\m+&\m-&\m+&\m+&\m+&\m+&\m-\\
			\end{bmatrix}}
		\nonumber
		\end{eqnarray}
		\caption{UD code set $\mathbf{C}$ with $L=8$ and $K=13$.} \label{C_8_13_02}
	\end{figure}
	\vspace{-0.3cm}
	 The BER performance of UD code sets are averaged over the different users for $\mathbf{C}_{4 \times 5}$, $\mathbf{C}_{8\times 13}$ and $\mathbf{C}_{16\times 33}$, respectively.
	The performance of the proposed FDA is comparable to that of ML, as shown in Fig. \ref{C4by5}.
	\vspace{-0.5cm}
	\hspace{-0.6cm}
		\begin{figure}[H]
			
			\includegraphics[width=3.8 in]{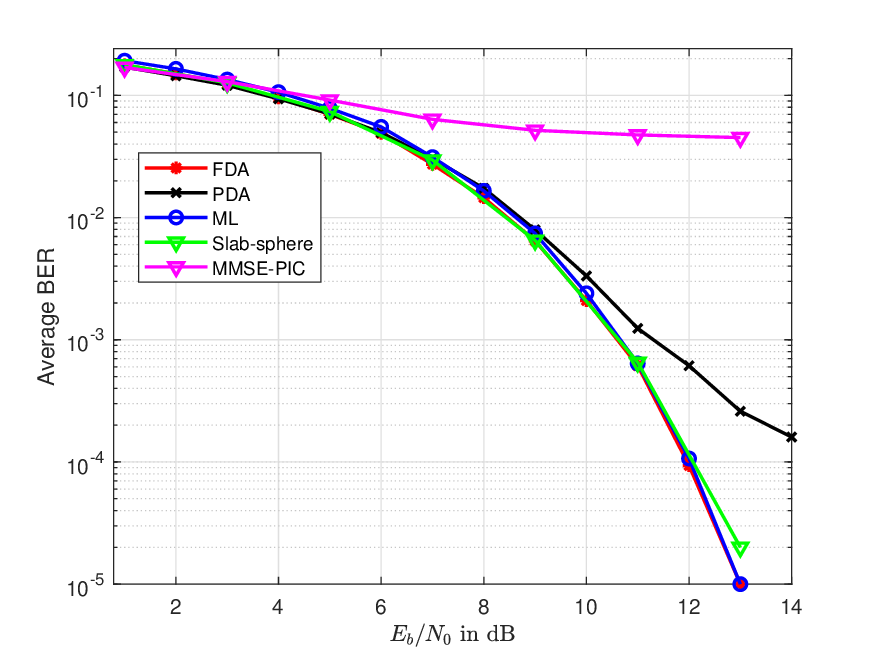}
			\centering \caption{UD code set $\mathbf{C}_{4 \times 5}$.} \label{C4by5}
		\end{figure}
	\vspace{-0.5cm}
	For the larger values of $L$s, the FDA has slightly inferior performance in terms of BER compared to ML. However, in practice the $E_b/N_o$ at the BER of $10^{-3}$ is considered to be the operating threshold since we can apply channel encoding to achieve even as low as $10^{-6}$ BER. At the BER of $10^{-3}$ the FDA achieves $1$ dB and $4$ dB gain compared to the Slab-sphere, $4$ dB and $15$ dB gain compared the PDA, as shown in Figs. \ref{C8by13} and \ref{C16by33}. The significant increase of the BER performance gap between the FDA and Slab-sphere, PDA detectors for greater values of $L$s is due to error floor experienced by the Slab-sphere and the PDA detectors. The reason for this is since for larger values of $L$ the overload factor increases exponentially. This in turn results in degradation of the linear separabilty of the overall system.
	\vspace{-0.5cm}
		\begin{figure}[H]
				\hspace{-0.4cm}
			\includegraphics[width=3.8 in]{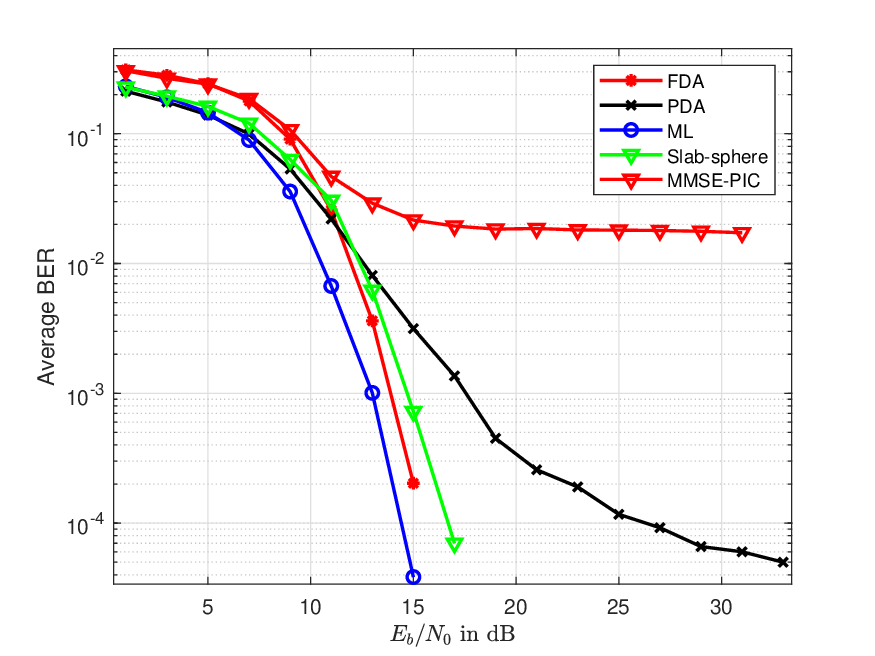}
			\centering \caption{UD code set $\mathbf{C}_{8 \times 13}$.} \label{C8by13}
		\end{figure}
	\vspace{-0.5cm}
	The linear separability criterion for linear detectors such as PDA have the property of vanishing BER as the channel noise goes to zero.
   Notice that we omitted the discussion of MMSE-PIC detector due to the overall poor BER performance of MMSE-PIC detector in comparison to other detectors. The Table \ref{table:complexity_bigO} shows the computational complexity of all the detectors. Even though the performance of the proposed FDA is slightly worse than the ML detector, it has much lower complexity compared to ML.
	\section{Conclusion}
\label{conclusion}
In this paper, we introduced a novel fast (low-complexity) decoder algorithm (FDA) for antipodal uniquely decodable (UD) code sets. The proposed algorithm has a much lower computational complexity compared to the maximum likelihood (ML) decoder whose complexity may be prohibitive for even moderate code lengths. Simulation results show that the performance of the proposed decoder is almost as good as that of the ML decoder with only a $1-2$ dB SNR degradation at a BER of $10^{-3}$. Moreover, we proved the minimum Manhattan distance of UD codes proposed in \cite{michel2012} and a number of propositions which collectively served to be the foundation of the formal proof for the maximum number of users, $K_{\rm{max}}^a$, for the case of $L=8$.

{In our future research, we will conceive multiuser detection for higher-order constellations for transmission over dispersive fading channels.}
\vspace{-0.5cm}
		\begin{figure}[H]
			\includegraphics[width=3.8 in]{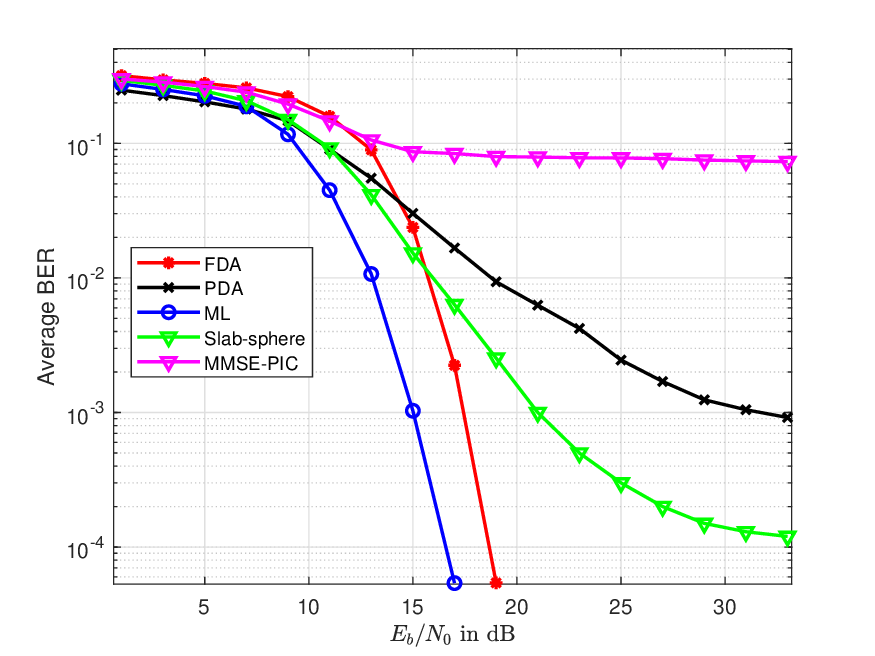}
			\centering \caption{UD code set $\mathbf{C}_{16 \times 33}$.}\label{C16by33}
		\end{figure}
	\vspace{-1.0cm}
\appendices
\section{Proof of the conversion from $\mathbf{C}^a$ to $\mathbf{C}^b$ }
\label{appendixA}
The proof of the conversion from the antipodal overloaded UD code sets to binary UD code sets, coined as optical CDMA code sets in \cite{Marvasti2009}, are presented next.

\begin{theorem}
	If there is an antipodal UD code set $\mathbf{C}^a \in \{\pm 1\}^{L \times K}$, then there is an equivalent binary UD code set $\mathbf{C}^b \in \{0, 1\}^{L \times K}$.
\end{theorem}
\begin{proof}
	Suppose there is an antipodal UD code set $\mathbf{C}_{L\times K}^a$. By corollary, if multiplying each row or column by $-1$ we can assume that the entries of the first row of $\mathbf{C}_{L\times K}^a$ are all $1$s. Let the conversion to the binary matrix $\mathbf{C}_{L \times K}^b = (\mathbf{C}_{L \times K}^a + \mathbf{J})/2$, where $\mathbf{J}$ is the $L \times K$ all-one matrix. It is clear that $\mathbf{C}_{L \times K}^b \in \{0,1 \}^{L \times K}$, therefore, we now need to prove
	the following $\mathsf{Null}\footnote{$\mathsf{Null}()$ represents the nullspace of a matrix.}(\mathbf{C}_{L \times K}^b) \cap \{0,\pm 1\}^{K\times 1} = \{0\}^{K\times 1}$ statement. Assume that $\mathbf{C}_{L \times K}^b \mathbf{z} = \boldsymbol{0}_L$, which yields to $ (\mathbf{C}_{L \times K}^a + \mathbf{J}) \mathbf{z} = \boldsymbol{0}_L$ and thus, $\mathbf{C}_{L \times K}^a \mathbf{z} =-\mathbf{J} \mathbf{z}$, where $\mathbf{z} \in \{0,\pm 1 \}^{K \times 1}$. Since the entries of the first row of $\mathbf{C}_{L \times K}^a$ as well as the matrix $\mathbf{J}$ are all $1$s, the first entry of $\mathbf{C}_{L \times K}^a \mathbf{z}$ must be equal to the first entry of -$\mathbf{J} \mathbf{z}$. It is only possible if the first entry of -$\mathbf{J} \mathbf{z}$ is $0$. Thus, -$\mathbf{J} \mathbf{z} = \boldsymbol{0}_L$ which leads to $\mathbf{z} = \boldsymbol{0}_L$. 
	As a consequence of UD code set of $\mathbf{C}_{L \times K}^a$ that satisfies unique decodability condition (\ref{null02}) the expression $\mathbf{C}_{L \times K}^a \mathbf{z} = \boldsymbol{0}_L$ implies that $\mathbf{z} = \boldsymbol{0}_L$.
\end{proof}
However, in general it is not necessarily true that any binary (optical) code sets $\mathbf{C}^b$ of any size can be converted to UD code sets $\mathbf{C}^a$. The reason is because for a given UD code set $\mathbf{C}^b$ with $K_{\rm{max}}^b$, the following is true $K_{\rm{max}}^b \geq K_{\rm{max}}^a$ since $K_{\rm{max}}^b = \upgamma(L+1)$ and $K_{\rm{max}}^a = \upgamma(L)+1$ and obviously $\upgamma(L+1) \geq \upgamma(L)+1$. It is equality, $K_{\rm{max}}^b = K_{\rm{max}}^a$, if $L = 2^i$ where $i \in \{1, 2, \dots  \}$. Therefore, for a given UD code set $\mathbf{C}^b$ with $K_{\rm{max}}^b$, if $L \neq  2^i$ we have $K_{\rm{max}}^b > K_{\rm{max}}^a$ and there is no equivalent UD code set $\mathbf{C}^a$.
\section{Proof of All Possible Combinations}
\label{appendixB}

\textit{Combination 1} $(\mathcal{A}, \mathcal{A}, \mathcal{A}, \mathcal{D}, \mathcal{G})$: This group combination is such that distinct $(\mathcal{A}, \mathcal{A}, \mathcal{A})$ produce $\binom{3}{2}=3$ distinct $\mathcal{D}$'s different from existing $\mathcal{D}$, and $\mathcal{G}$'s by $\mathcal{A} + \mathcal{A} = \mathcal{D}/\mathcal{G}$ rule. Different $\mathcal{A}$ is also produced by $\mathcal{A} + \mathcal{D} = \mathcal{A}$ from existing ones. With this new $\mathcal{A}$ it can produce another distinct $\binom{3}{2}=3$ $\mathcal{G}$'s. Total of six $\mathcal{G}$'s are produced with $(\mathcal{A}, \mathcal{A}, \mathcal{A}, \mathcal{D})$ and existing $\mathcal{G}$ we have seven $\mathcal{G}$'s. This completes all seven different $\mathcal{G}$'s from each seven groups, hence we cannot add any other $\mathcal{G}$ to the existing combination. So three $\mathcal{A}$'s with $\mathcal{D}$ produce six distinct $\mathcal{A}$'s and the $\mathcal{A} + \mathcal{A} +\mathcal{A} = \mathcal{A}$, $\mathcal{A} + \mathcal{A} +\mathcal{A} + \mathcal{D} = \mathcal{A}$ produce another two distinct $\mathcal{A}$'s. Therefore, we cannot add any more $\mathcal{A}$ to the exisiting cominations. Now three $\mathcal{A}$'s will produce three distinct $\mathcal{D}$'s, and three $\mathcal{A}$'s with $\mathcal{D}$, $\mathcal{A} + \mathcal{D} = \mathcal{A}$, produce another three disctinct $\mathcal{D}$'s, plus the $\mathcal{D}$ that is in the combination that makes a total of seven $\mathcal{D}$'s. Therefore, we cannot add any more $\mathcal{D}$'s to the combination. Therefore, we proved that we cannot add any more $\mathcal{A}$, $\mathcal{D}$ and $\mathcal{G}$ to the $(\mathcal{A}, \mathcal{A}, \mathcal{A}, \mathcal{D}, \mathcal{G})$.

\textit{Combination 2} $(\mathcal{A}, \mathcal{A}, \mathcal{D}, \mathcal{D}, \mathcal{D})$: This group combination is such that
distinct $(\mathcal{D}, \mathcal{D}, \mathcal{D})$ produce $\binom{3}{2}=3$ distinct $\mathcal{G}$'s by $\mathcal{D} + \mathcal{D} = \mathcal{G}$ rule, and those $\mathcal{G}$'s must be different from created by $(\mathcal{A}, \mathcal{A})$ $\mathcal{A} + \mathcal{A} = \mathcal{G}$ rule. Also each $\mathcal{A}$ must be different from three distinct $\mathcal{A}$'s created by $\mathcal{D} + \mathcal{A} = \mathcal{A}$ rule. Each existing $\mathcal{A}$'s with $\binom{3}{2}=3$ distinct $\mathcal{D} + \mathcal{D} +  \mathcal{A} = \mathcal{A}$ produce six distinct $\mathcal{A}$'s and plus two existing $\mathcal{A}$ makes a total of eight $\mathcal{A}$'s. This completes all distinct $\mathcal{A}$'s and we cannot add any more $\mathcal{A}$ to the combination. From existing combination there are four distict $\mathcal{D}$'s produced by $\mathcal{A} + \mathcal{A} + \mathcal{D} =  \mathcal{D}$ and $\mathcal{D} + \mathcal{D} + \mathcal{D} =  \mathcal{D}$ rules as well as three existing $\mathcal{D}$'s this make a total of seven distinct $\mathcal{D}$'s. Hence, we cannot add any more $\mathcal{D}$ to the combination. We have already seen that with existing combinations we create six distint $\mathcal{G}$'s and with $(\mathcal{A},\mathcal{A}, \mathcal{D})$ we can create more $\mathcal{G}$'s by $\mathcal{A} + \mathcal{A} + \mathcal{D} = \mathcal{G}$ rule. This tells us that we cannot add any more $\mathcal{G}$'s. Therefore, we proved that we cannot add any more $\mathcal{A}$, $\mathcal{D}$ and $\mathcal{G}$ to the $(\mathcal{A}, \mathcal{A}, \mathcal{D}, \mathcal{D}, \mathcal{D})$.

\textit{Combination 3} $(\mathcal{A}, \mathcal{A}, \mathcal{D}, \mathcal{D}, \mathcal{G})$: This group combination is such that different $\mathcal{A}$'s created by $\mathcal{A} + \mathcal{D} =\mathcal{A}$, $\mathcal{A} + \mathcal{D} + \mathcal{D} =\mathcal{A}$, $\mathcal{A} + \mathcal{G} =\mathcal{A}$, $\mathcal{A} + \mathcal{D}+ \mathcal{G} =\mathcal{A}$, $\mathcal{A} + \mathcal{D}+ \mathcal{D}+\mathcal{G} =\mathcal{A}$ rules are distinct from existing $\mathcal{A}$'s. Similarly, different $\mathcal{D}$'s created by $\mathcal{A} + \mathcal{A} =\mathcal{D}$, $\mathcal{A} + \mathcal{A} + \mathcal{G} =\mathcal{D}$, $\mathcal{D} + \mathcal{G} =\mathcal{D}$ rules are distinct from existing $\mathcal{D}$'s and different $\mathcal{G}$'s created by $\mathcal{A} + \mathcal{A} =\mathcal{G}$, $\mathcal{A} + \mathcal{D} + \mathcal{A} =\mathcal{G}$, $\mathcal{A} + \mathcal{A}+\mathcal{D} + \mathcal{D} =\mathcal{G}$ rules are distinct from existing $\mathcal{G}$'s. The combination produces six distinct $\mathcal{A}$'s by $\mathcal{A} + \mathcal{D} + \mathcal{D} = \mathcal{A}$, $\mathcal{A} + \mathcal{G} = \mathcal{A}$, $\mathcal{A} + \mathcal{D} = \mathcal{A}$ rules and with two existing $\mathcal{A}$'s that makes a total of eight $\mathcal{A}$'s, hence we cannot add any more $\mathcal{A}$ to the combination. There are five distinct $\mathcal{D}$'s produced by $\mathcal{A} + \mathcal{A} + \mathcal{D} = \mathcal{D}$, $\mathcal{G} + \mathcal{D} = \mathcal{D}$, $\mathcal{A} + \mathcal{A} = \mathcal{D}$ rules and with two existing $\mathcal{D}$'s that makes a total of seven $\mathcal{D}$'s, hence we cannot add any more $\mathcal{D}$ to the combination. Similarly, for the case of $\mathcal{G}$ the combination produces six distinct $\mathcal{G}$'s by $\mathcal{D} + \mathcal{D}  = \mathcal{G}$, $\mathcal{A} + \mathcal{A} = \mathcal{G}$, $\mathcal{G} + \mathcal{A} + \mathcal{A}= \mathcal{G}$, $\mathcal{G} + \mathcal{D} + \mathcal{D}= \mathcal{G}$, $\mathcal{A} + \mathcal{A}+ \mathcal{D} + \mathcal{D}= \mathcal{G}$, $\mathcal{G} + \mathcal{A} + \mathcal{A}+ \mathcal{D} + \mathcal{D}= \mathcal{G}$ rules and with the existing $\mathcal{G}$'s that makes a total of seven $\mathcal{G}$'s, hence we cannot add any more $\mathcal{G}$ to the combination. Therefore, we proved that we cannot add any more $\mathcal{A}$, $\mathcal{D}$ and $\mathcal{G}$ to the $(\mathcal{A}, \mathcal{A}, \mathcal{D}, \mathcal{D}, \mathcal{G})$.

\textit{Combination 4} $(\mathcal{A}, \mathcal{A}, \mathcal{D}, \mathcal{G}, \mathcal{G})$: This group combination is such that different $\mathcal{A}$'s created by $\mathcal{A} + \mathcal{D} =\mathcal{A}$, $\mathcal{A} + \mathcal{G}  =\mathcal{A}$, $\mathcal{A} + \mathcal{G}+ \mathcal{G} =\mathcal{A}$, $\mathcal{A} + \mathcal{D}+ \mathcal{G} =\mathcal{A}$, $\mathcal{A} + \mathcal{D}+ \mathcal{G}+\mathcal{G} =\mathcal{A}$ rules are distinct from existing $\mathcal{A}$'s. Similarly, different $\mathcal{D}$'s created by $\mathcal{A} + \mathcal{A} =\mathcal{D}$, $\mathcal{A} + \mathcal{A} + \mathcal{G} =\mathcal{D}$, $\mathcal{A} + \mathcal{A} + \mathcal{G} + \mathcal{G} =\mathcal{D}$ rules are distinct from existing $\mathcal{D}$'s and different $\mathcal{G}$'s created by $\mathcal{A} + \mathcal{A} =\mathcal{G}$, $\mathcal{A} + \mathcal{D} + \mathcal{A} =\mathcal{G}$ rules are distinct from existing $\mathcal{G}$'s. The combination produces six distinct $\mathcal{A}$'s by $\mathcal{A} + \mathcal{D}  = \mathcal{A}$, $\mathcal{A} + \mathcal{D} + \mathcal{G} = \mathcal{A}$ rules and with two existing $\mathcal{A}$'s that makes a total of eight $\mathcal{A}$'s, therefore we cannot add any more $\mathcal{A}$ to the combination. There are six distinct $\mathcal{D}$'s produced by $\mathcal{A} + \mathcal{A} = \mathcal{D}$, $\mathcal{G} + \mathcal{D} = \mathcal{D}$, $\mathcal{D} + \mathcal{G} +  \mathcal{G}  = \mathcal{D}$, $\mathcal{A} + \mathcal{A} +  \mathcal{G}  = \mathcal{D}$ rules and with the existing $\mathcal{D}$'s that makes a total of seven $\mathcal{D}$'s, hence we cannot add any more $\mathcal{D}$ to the combination. Similarly, for the case of $\mathcal{G}$ the combination produces five distinct $\mathcal{G}$'s by $\mathcal{A} + \mathcal{A}  = \mathcal{G}$, $\mathcal{G} + \mathcal{G} = \mathcal{G}$, $\mathcal{G} + \mathcal{A} + \mathcal{A}= \mathcal{G}$ rules and with the existing $\mathcal{G}$'s that makes a total of seven $\mathcal{G}$'s, hence we cannot add any more $\mathcal{G}$ to the combination. Therefore, we proved that we cannot add any more $\mathcal{A}$, $\mathcal{D}$ and $\mathcal{G}$ to the $(\mathcal{A}, \mathcal{A}, \mathcal{D}, \mathcal{G}, \mathcal{G})$.

\textit{Combination 5} $(\mathcal{A}, \mathcal{D}, \mathcal{D}, \mathcal{D}, \mathcal{D})$: This group combination is such that all $\mathcal{D}$'s are distinct and no $\mathcal{D} +\mathcal{D} = \mathcal{D} + \mathcal{D} $ is satisfied. The combination produces seven distinct $\mathcal{A}$'s by $\mathcal{A} + \mathcal{D}  = \mathcal{A}$, $\mathcal{A} + \mathcal{D} + \mathcal{D} = \mathcal{A}$ rules and with the existing $\mathcal{A}$ that makes a total of eight $\mathcal{A}$'s, hence we cannot add any more $\mathcal{A}$ to the combination. There are three distinct $\mathcal{D}$'s produced by $\mathcal{D} + \mathcal{D} + \mathcal{D} = \mathcal{D}$ rules and with the existing $\mathcal{D}$'s that makes a total of seven $\mathcal{D}$'s, hence we cannot add any more $\mathcal{D}$ to the combination. Similarly, for the case of $\mathcal{G}$ the combination produces seven distinct $\mathcal{G}$'s by $\mathcal{D} + \mathcal{D}  = \mathcal{G}$, $\mathcal{D} + \mathcal{D} +\mathcal{D} + \mathcal{D} = \mathcal{G}$ rules and with the existing $\mathcal{G}$'s that makes a total of seven $\mathcal{G}$'s, hence we cannot add any more $\mathcal{G}$ to the combination. Therefore, we proved that we cannot add any more $\mathcal{A}$, $\mathcal{D}$ and $\mathcal{G}$ to the $(\mathcal{A}, \mathcal{D}, \mathcal{D}, \mathcal{D}, \mathcal{D})$.

\textit{Combination 6} $(\mathcal{A}, \mathcal{D}, \mathcal{D}, \mathcal{D}, \mathcal{G})$: This group combination is such that different $\mathcal{D}$'s created by $\mathcal{D} + \mathcal{G} =\mathcal{D}$, $\mathcal{D} + \mathcal{D} + \mathcal{D} =\mathcal{D}$ rules are distinct from existing $\mathcal{D}$'s and different $\mathcal{G}$'s created by $\mathcal{D} + \mathcal{D} =\mathcal{G}$, $\mathcal{A} + \mathcal{D} + \mathcal{D} =\mathcal{G}$ rules are distinct from existing $\mathcal{G}$'s. The combination produces seven distinct $\mathcal{A}$'s by $\mathcal{A} + \mathcal{D}  = \mathcal{A}$, $\mathcal{A} + \mathcal{G} = \mathcal{A}$, $\mathcal{A} + \mathcal{D}+ \mathcal{G} = \mathcal{A}$ rules and with the existing $\mathcal{A}$ that makes a total of eight $\mathcal{A}$'s, hence we cannot add any more $\mathcal{A}$ to the combination. There are four distinct $\mathcal{D}$'s produced by $\mathcal{D} + \mathcal{G} = \mathcal{D}$, $\mathcal{D} + \mathcal{D} + \mathcal{D} + \mathcal{G} = \mathcal{D}$ rules and with three existing $\mathcal{D}$'s that makes a total of seven $\mathcal{D}$'s, hence we cannot add any more $\mathcal{D}$ to the combination. Similarly, for the case of $\mathcal{G}$ the combination produces five distinct $\mathcal{G}$'s by $\mathcal{D} + \mathcal{D}  = \mathcal{G}$, $\mathcal{D}+ \mathcal{D} + \mathcal{G} = \mathcal{G}$ rules and with the existing $\mathcal{G}$ that makes a total of seven $\mathcal{G}$'s, hence we cannot add any more $\mathcal{G}$ to the combination. Therefore, we proved that we cannot add any more $\mathcal{A}$, $\mathcal{D}$ and $\mathcal{G}$ to the $(\mathcal{A}, \mathcal{D}, \mathcal{D}, \mathcal{D}, \mathcal{G})$.

\textit{Combination 7} $(\mathcal{A}, \mathcal{D}, \mathcal{D}, \mathcal{G}, \mathcal{G})$: This group combination is such that different $\mathcal{D}$'s created by $\mathcal{D} + \mathcal{G} =\mathcal{D}$ rules are distinct from existing $\mathcal{D}$'s and different $\mathcal{G}$'s created by $\mathcal{D} + \mathcal{D} =\mathcal{G}$ rules are distinct from existing $\mathcal{G}$'s. The combination produces seven distinct $\mathcal{A}$'s by $\mathcal{A} + \mathcal{D}  = \mathcal{A}$, $\mathcal{A} + \mathcal{G} = \mathcal{A}$, $\mathcal{A} + \mathcal{D}+  \mathcal{D}+ \mathcal{G} = \mathcal{A}$ rules and with the existing $\mathcal{A}$ that makes a total of eight $\mathcal{A}$'s, hence we cannot add any more $\mathcal{A}$ to the combination. There are five distinct $\mathcal{D}$'s produced by $\mathcal{D} + \mathcal{G} = \mathcal{D}$, $\mathcal{D} + \mathcal{G} + \mathcal{G} = \mathcal{D}$ rules and with two existing $\mathcal{D}$'s that makes a total of seven $\mathcal{D}$'s, hence we cannot add any more $\mathcal{D}$ to the combination. Similarly, for the case of $\mathcal{G}$ the combination produces five distinct $\mathcal{G}$'s by $\mathcal{D} + \mathcal{D}  = \mathcal{G}$, $\mathcal{G} + \mathcal{G}  = \mathcal{G}$, $\mathcal{D}+ \mathcal{D} + \mathcal{G} = \mathcal{G}$, $\mathcal{D}+ \mathcal{G} + \mathcal{G} = \mathcal{G}$ rules and with the existing $\mathcal{G}$ that makes a total of seven $\mathcal{G}$'s, hence we cannot add any more $\mathcal{G}$ to the combination. Therefore, we proved that we cannot add any more $\mathcal{A}$, $\mathcal{D}$ and $\mathcal{G}$ to the $(\mathcal{A}, \mathcal{D}, \mathcal{D}, \mathcal{G}, \mathcal{G})$.

\textit{Combination 8} $(\mathcal{A}, \mathcal{D}, \mathcal{G}, \mathcal{G}, \mathcal{G})$: This group combination is such that different $\mathcal{D}$'s created by $\mathcal{D} + \mathcal{G} =\mathcal{D}$ rules are distinct from existing $\mathcal{D}$'s and different $\mathcal{G}$'s created by $\mathcal{G} + \mathcal{G} =\mathcal{G}$ rules are distinct from existing $\mathcal{G}$'s. The combination produces seven distinct $\mathcal{A}$'s by $\mathcal{A} + \mathcal{D}  = \mathcal{A}$, $\mathcal{A} + \mathcal{G} = \mathcal{A}$, $\mathcal{A} + \mathcal{G}+ \mathcal{G} = \mathcal{A}$, $\mathcal{A} + \mathcal{D}+ \mathcal{G} = \mathcal{A}$ rules and with the existing $\mathcal{A}$ that makes a total of eight $\mathcal{A}$'s, hence we cannot add any more $\mathcal{A}$ to the combination. There are six distinct $\mathcal{D}$'s produced by $\mathcal{D} + \mathcal{G} = \mathcal{D}$, $\mathcal{D} + \mathcal{G} + \mathcal{G} = \mathcal{D}$, $\mathcal{D} + \mathcal{G} + \mathcal{G} + \mathcal{G} = \mathcal{D}$ rules and with two existing $\mathcal{D}$'s that makes a total of seven $\mathcal{D}$'s, hence we cannot add any more $\mathcal{D}$ to the combination. Similarly, for the case of $\mathcal{G}$ the combination produces five distinct $\mathcal{G}$'s by $\mathcal{G} + \mathcal{G}  = \mathcal{G}$, $\mathcal{G} + \mathcal{G} + \mathcal{G}  = \mathcal{G}$ rules and with the existing $\mathcal{G}$ that makes a total of seven $\mathcal{G}$'s, hence we cannot add any more $\mathcal{G}$ to the combination. Therefore, we proved that we cannot add any more $\mathcal{A}$, $\mathcal{D}$ and $\mathcal{G}$ to the $(\mathcal{A}, \mathcal{D}, \mathcal{G}, \mathcal{G}, \mathcal{G})$.



	 
\bibliographystyle{IEEEtran}
\bibliography{michelbib_file}

\vspace{-0.3cm}
\begin{IEEEbiography}[{\includegraphics[width=1in,height=1.25in,clip,keepaspectratio]{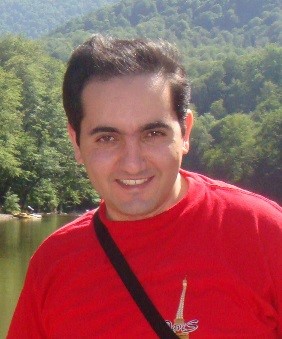}}]{\bf Michel Kulhandjian} (M'18-SM'20) received his M.S. and Ph.D. degrees in Electrical Engineer from the State University of New York at Buffalo in 2007 and 2012, respectively. He had previously received his B.S. degree in Electronics Engineering and Computer Science (Minor), with ``Summa Cum Laude'' from the American University in Cairo (AUC) in 2005. He was employed at Alcatel-Lucent, in Ottawa, Ontario, in 2012. In the same year he was appointed as a Research Associate at EION Inc. He received Natural Science and Engineering Research Council of Canada (NSERC) Industrial R\&D Fellowship (IRDF). 
	He is currently a Research Scientist at the School of Electrical Engineering and Computer Science at the University of Ottawa. He is also employed as a senior embedded software engineer at L3Harris Technologies. 
	
His research interests include wireless multiple access communications, adaptive coded modulation, waveform design for overloaded code-division multiplexing applications, channel coding, space-time coding, adaptive multiuser detection, statistical signal processing, covert communications, spread-spectrum steganography and steganalysis. He has served as a guest editor for Journal of Sensor and Actuator Networks (JSON). He actively serves as member of Technical Program Committee (TPC) of IEEE WCNC, IEEE GLOBECOM, IEEE ICC, and IEEE VTC.
\end{IEEEbiography}
\vspace{-0.6cm}
\begin{IEEEbiography}[{\includegraphics[width=1in,clip,keepaspectratio]{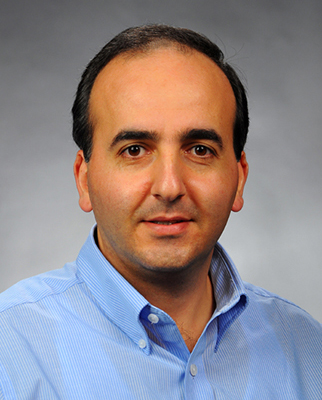}}]{\bf Hovannes Kulhandjian}(S'14-M'15-SM'20) received the
	B.S. degree (magna cum laude) in electronics engineering from The American University in Cairo, Cairo, Egypt, in 2008, and the M.S. and Ph.D. degrees in electrical engineering from the State University of New York at Buffalo, Buffalo, NY, USA,
	in 2010 and 2014, respectively.
	From December 2014 to July 2015, he was an
	Associate Research Engineer with the Department of
	Electrical and Computer Engineering, Northeastern
	University, Boston, MA, USA. He is currently an
	Associate Professor with the Department of Electrical and Computer Engineering, California State University, Fresno, Fresno, CA, USA. His current
	research interests include  wireless communications
	and networking, with applications to underwater acoustic communications, visible light communications and applied machine learning. He has served as a guest editor for IEEE Access - Special Section Journal on Underwater Wireless Communications and Networking. He has also served as a Session Co-Chair for IEEE UComms 2020, Session Chair for ACM WUWNet 2019.
	He actively serves as a member of the Technical Program Committee for ACM and IEEE conferences such as IEEE GLOBECOM, ICC, UComms, PIMRC, WD, ACM WUWNet, among others. 
\end{IEEEbiography}
\vspace{-0.6cm}
\begin{IEEEbiography}[{\includegraphics[width=1in,height=1.32in, clip,keepaspectratio]{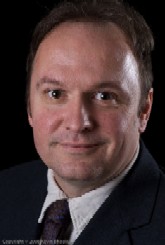}}]{\bf Claude D'Amours} received the degrees of B.A.Sc, M.A.Sc. and Ph.D. in Electrical Engineering from the University of Ottawa in 1990, 1992 and 1995 respectively. In 1992 he was employed as a Systems Engineer at Calian Communications Ltd. In 1995 he joined the Communications Research Centre in Ottawa, Ontario, Canada, as a Systems Engineer.  Later in 1995, he joined the Department of Electrical and Computer Engineering at the Royal Military College of Canada in Kingston, Ontario, Canada, as an Assistant Professor.  He joined the School of Information Technology and Engineering (SITE), which has since been renamed as the School of Electrical Engineering and Computer Science (EECS), at the University of Ottawa as an Assistant Professor in 1999.  From 2007-2011, he served as Vice Dean of Undergraduate Studies for the Faculty of Engineering and has been serving as the Director of the School of EECS at the University of Ottawa since 2013.  His research interests are in physical layer technologies for wireless communications systems, notably in multiple access techniques and interference cancellation.
\end{IEEEbiography}	
\vspace{-0.6cm}
\begin{IEEEbiography}[{\includegraphics[width=1.25in, height=1.33in, clip,keepaspectratio]{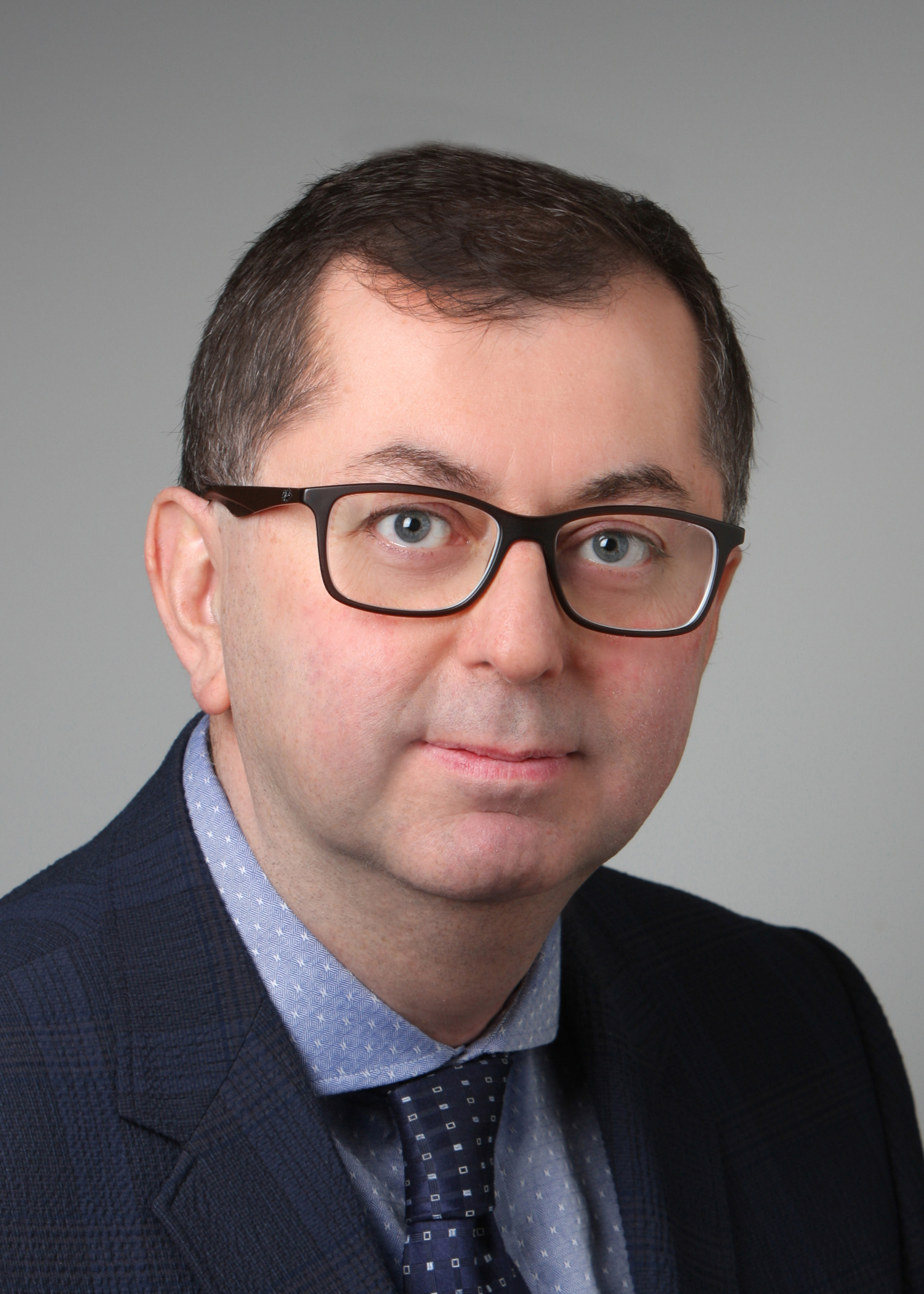}}]{\bf Halim Yanikomeroglu}(Fellow, IEEE) is a Professor in the Department of Systems and Computer Engineering at Carleton University, Ottawa, Canada. His research group has made substantial contributions to 4G and 5G wireless technologies. His group's current focus is the aerial and satellite networks for the 6G and beyond-6G era. His extensive collaboration with industry resulted in 39 granted patents. He is a Fellow of IEEE, EIC (Engineering Institute of Canada), and CAE (Canadian Academy of Engineering), and a Distinguished Speaker for both IEEE Communications Society and IEEE Vehicular Technology Society. Dr. Yanikomeroglu received several awards for his research, teaching, and service.
\end{IEEEbiography}
\vspace{-0.6cm}
\begin{IEEEbiography}[{\includegraphics[width=1.6in, height=1.36in, clip,keepaspectratio]{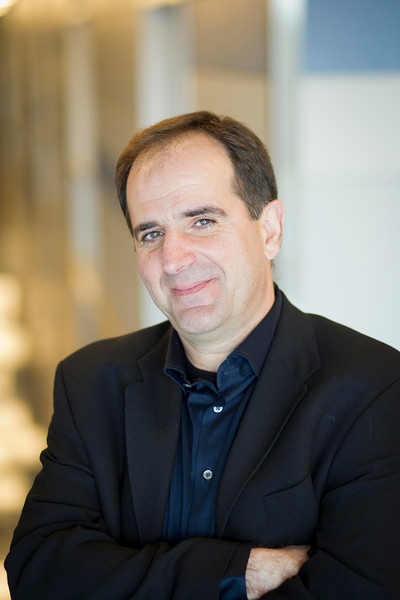}}]{\bf Dimitris A. Pados}(M'95-SM'15) received the Diploma degree in computer science and engineering (five-year program) from the University of Patras, Greece, and the Ph.D. degree in electrical engineering from the University of Virginia, Charlottesville, VA. From 1997 to 2017, he was with the Department of Electrical Engineering, The State University of New York at Buffalo, as Assistant Professor, Associate Professor, Professor, and Clifford C. Furnas Chair Professor of Electrical Engineering. He also served as Associate Chair and was appointed Chair of the Department of Electrical Engineering. He was elected University Faculty Senator four times and served on the Faculty Senate Executive Committee for two terms. In 2017, he joined Florida Atlantic University, Boca Raton, FL, as the Schmidt Eminent Scholar Professor of Engineering and Computer Science and Fellow of the Institute for Sensing and Embedded Network Systems Engineering (I-SENSE). Dr. Pados is the Founding Director of the FAU Center for Connected Autonomy and Artificial Intelligence: https://ca-ai.fau.edu/

Dr. Pados is a member of the IEEE Communications, IEEE Signal Processing, IEEE Information Theory, and IEEE Computational Intelligence Societies. He served as Associate Editor for the IEEE SIGNAL PROCESSING LETTERS and the IEEE TRANSACTIONS ON NEURAL NETWORKS. Articles that he co-authored with his students received the 2001 IEEE International Conference on Telecommunications Best Paper Award, the 2003 IEEE TRANSACTIONS ON NEURAL NETWORKS Outstanding Paper Award, the 2010 IEEE International Communications Conference (ICC) Best Paper Award in signal processing for communications, the 2013 International Symposium on Wireless Communication Systems Best Paper Award in physical layer communications and signal processing, Best of IEEE GLOBECOM 2014-Top 50 Papers Distinction, Best Paper in the 2016 IEEE International Conference on Multimedia Big Data, and paper distinctions at iWAT (International Workshop on Antenna Technology) 2019, and IEEE/MTS Oceans 2020. Dr. Pados is a recipient of the 2009 SUNY-wide Chancellor's Award for Excellence in Teaching and the 2011 University at Buffalo Exceptional Scholar-Sustained Achievement Award. He was presented with the 2021 Florida Atlantic Research \& Development Authority Distinguished Researcher Award, Boca Raton, FL.

Dr. Pados has served as Principal Investigator on federal grants (NSF and DoD) of about \$17M and has been author/co-author of 230 journal and conference proceedings articles in predominantly IEEE venues. Notable technical contributions from his team include small-sample-support adaptive filtering (auxiliary-vector filters), optimal total-squared-correlation multiple-access code sets (Karystinos-Pados bounds and designs), optimal spread-spectrum data hiding, L1-norm principal-component analysis (optimal algorithms for exact L1-norm PCA), and robust localization in extreme environments (L1-norm feature extraction from complex-valued data).
\end{IEEEbiography}
\vspace{-0.8cm}
\begin{IEEEbiography}[{\includegraphics[width=1.1in, height=1.46in, clip,keepaspectratio]{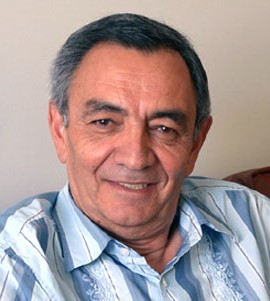}}]{\bf Gurgen Khatchatrian} is a well known and highly regarded applied scientist, scholar and technologist in the field of Error Control Coding and Cryptography. For his career achievements in these fields, he was elected as a full member of Armenian National Academy of Sciences in 1996 which is the highest honor awarded to scientists in Armenia. Dr. Khachatrian worked from 1999 through 2001 as Chief Cryptographer for Cylink Corporation (Nasdaq: CYLK) located in Sunnyvale, CA. Dr. Khachatrian then worked from 2001 to 2007 as a Chief Scientist for Quantum Digital Solutions corporation located in Santa Monica, CA. He was a Professor at American University of Armenia (AUA) from 2011 to 2018. Since 2019 he is a Chief Cryptography Officer at Quantum Digital Solutions Corporation (QDSC) in Marina Del Rey, CA, USA and Professor on leave from AUA.

\end{IEEEbiography}

\EOD

\end{document}

%% file: nomenclature.tex
\usepackage{acro}

\newlength\myitemwidth
\newlist{listabbrev}{description}{1}
\setlist[listabbrev]{
    labelindent = 0pt,
    labelsep    = 0pt,
    leftmargin  = \myitemwidth,
    labelwidth  = \myitemwidth,
    format      = \normalfont
    }

\NewAcroTemplate[list]{styleabbrev}{%

  \UseAcroTemplate[list]{description}[0]%
}

\acsetup{
  list/template  = styleabbrev,
}

\DeclareAcronym{IOT}{
  short = \normalfont{IoT} ,
  long  = Internet of Things,
  sort  = i,
  tag = nomencl
}
\DeclareAcronym{4G}{
  short = \normalfont{4G} ,
  long  =  The fourth generation,
  sort  = 4,
  tag = nomencl
}
\DeclareAcronym{5G}{
  short = \normalfont{5G} ,
  long  =  The fifth generation,
  sort  = 5,
  tag = nomencl
}

\DeclareAcronym{6G}{
  short = \normalfont{6G} ,
  long  =  The sixth generation,
  sort  = 6,
  tag = nomencl
}

\DeclareAcronym{LDS}{
  short = \normalfont{LDS} ,
  long  =  Low-density spreading,
  sort  = lds,
  tag = nomencl
}

\DeclareAcronym{NOMA}{
  short = \normalfont{NOMA} ,
  long  =  Non-orthogonal multiple access,
  sort  = noma,
  tag = nomencl
}

\DeclareAcronym{MPA}{
  short = \normalfont{MPA} ,
  long  =  Message passing algorithm,
  sort  = mpa,
  tag = nomencl
}

\DeclareAcronym{MAP}{
  short = \normalfont{MAP} ,
  long  = Maximum \textit{a posteriori},
  sort  = map,
  tag = nomencl
}

\DeclareAcronym{MMSE}{
  short = \normalfont{MMSE} ,
  long  = Minimum mean-square estimation,
  sort  = mmse,
  tag = nomencl
}

\DeclareAcronym{BER}{
  short = \normalfont{BER} ,
  long  = Bit-error-rate,
  sort  = ber,
  tag = nomencl
}

\DeclareAcronym{MUD}{
  short = \normalfont{MUD} ,
  long  = Multiuser detection,
  sort  = mud,
  tag = nomencl
}

\DeclareAcronym{MIMO}{
  short = \normalfont{MIMO} ,
  long  = Multiple-input multiple-output,
  sort  = mimo,
  tag = nomencl
}

\DeclareAcronym{M2M}{
  short = \normalfont{M2M} ,
  long  = Machine-to-machine,
  sort  = m2m,
  tag = nomencl
}

\DeclareAcronym{D2D}{
  short = \normalfont{D2D} ,
  long  = Device-to-device,
  sort  = d2d,
  tag = nomencl
}

\DeclareAcronym{D2E}{
  short = \normalfont{D2E} ,
  long  = Device-to-everything ,
  sort  = d2e,
  tag = nomencl
}

\DeclareAcronym{IoV}{
  short = \normalfont{IoV} ,
  long  = Internet of Vehicles,
  sort  = iov,
  tag = nomencl
}

\DeclareAcronym{MTC}{
  short = \normalfont{MTC} ,
  long  = Machine-type communication,
  sort  = mtc,
  tag = nomencl
}

\DeclareAcronym{ITS}{
  short = \normalfont{ITS} ,
  long  = Intelligent transportation system,
  sort  = its,
  tag = nomencl
}

\DeclareAcronym{OMA}{
  short = \normalfont{OMA} ,
  long  = Orthogonal multiple access,
  sort  = oma,
  tag = nomencl
}

\DeclareAcronym{TDMA}{
  short = \normalfont{TDMA} ,
  long  = Time-division multiple access,
  sort  = tdma,
  tag = nomencl
}

\DeclareAcronym{FDMA}{
  short = \normalfont{FDMA} ,
  long  = Frequency division multiple access,
  sort  = fdma,
  tag = nomencl
}

\DeclareAcronym{CDM}{
  short = \normalfont{CDM} ,
  long  = Code-division multiplexing,
  sort  = cdm,
  tag = nomencl
}

\DeclareAcronym{PN}{
  short = \normalfont{PN} ,
  long  = Pseudo-noise,
  sort  = pn,
  tag = nomencl
}

\DeclareAcronym{MO}{
  short = \normalfont{MO} ,
  long  = Multiple-orthogonal,
  sort  = mo,
  tag = nomencl
}

\DeclareAcronym{CDMA}{
  short = \normalfont{CDMA} ,
  long  = Code-division multiple access,
  sort  = cdma,
  tag = nomencl
}

\DeclareAcronym{SCMA}{
  short = \normalfont{SCMA} ,
  long  = Sparse code multiple access,
  sort  = scma,
  tag = nomencl
}

\DeclareAcronym{PDMA}{
  short = \normalfont{PDMA} ,
  long  = Pattern division multiple access,
  sort  = pdma,
  tag = nomencl
}

\DeclareAcronym{MUSA}{
  short = \normalfont{MUSA} ,
  long  = Multiuser shared access,
  sort  = musa,
  tag = nomencl
}

\DeclareAcronym{MMTC}{
  short = \normalfont{mMTC} ,
  long  = Massive machine-type-communications,
  sort  = mmtc,
  tag = nomencl
}

\DeclareAcronym{LDPC}{
  short = \normalfont{LDPC} ,
  long  = Low-density parity-check,
  sort  = ldpc,
  tag = nomencl
}

\DeclareAcronym{AWGN}{
  short = \normalfont{AWGN} ,
  long  = Additive white Gaussian noise,
  sort  = awgn,
  tag = nomencl
}

\DeclareAcronym{FIR}{
  short = \normalfont{FIR} ,
  long  = Finite impulse response,
  sort  = fir,
  tag = nomencl
}

\DeclareAcronym{CIR}{
  short = \normalfont{CIR} ,
  long  = Channel impulse response,
  sort  = cir,
  tag = nomencl
}

\DeclareAcronym{MF}{
  short = \normalfont{MF} ,
  long  = Matched filter,
  sort  = mf,
  tag = nomencl
}

\DeclareAcronym{ZF}{
  short = \normalfont{ZF} ,
  long  = Zero-forcing,
  sort  = zf,
  tag = nomencl
}

\DeclareAcronym{SIC}{
  short = \normalfont{SIC} ,
  long  = Successive interference cancellation,
  sort  = sic,
  tag = nomencl
}

\DeclareAcronym{PIC}{
  short = \normalfont{PIC} ,
  long  = Parallel interference cancellation,
  sort  = pic,
  tag = nomencl
}

\DeclareAcronym{PDA}{
  short = \normalfont{PDA} ,
  long  = Probabilistic data association,
  sort  = pda,
  tag = nomencl
}

\DeclareAcronym{BP}{
  short = \normalfont{BP} ,
  long  = Believe propagation,
  sort  = bp,
  tag = nomencl
}

\DeclareAcronym{JSG}{
  short = \normalfont{JSG} ,
  long  = Joint sparse graph ,
  sort  = jsg,
  tag = nomencl
}

\DeclareAcronym{BPSK}{
  short = \normalfont{BPSK} ,
  long  = Binary phase-shift keying  ,
  sort  = bpsk,
  tag = nomencl
}

\DeclareAcronym{QPSK}{
  short = \normalfont{QPSK} ,
  long  = Quadrature phase-shift keying  ,
  sort  = qpsk,
  tag = nomencl
}

\DeclareAcronym{QAM}{
  short = \normalfont{QAM} ,
  long  = Quadrature amplitude modulation  ,
  sort  = qam,
  tag = nomencl
}

\DeclareAcronym{APP}{
  short = \normalfont{APP} ,
  long  = \textit{A posteriori} probability  ,
  sort  = app,
  tag = nomencl
}

\DeclareAcronym{LLR}{
  short = \normalfont{LLR} ,
  long  = Log-likelihood ratio  ,
  sort  = llr,
  tag = nomencl
}

\DeclareAcronym{DE}{
  short = \normalfont{DE} ,
  long  = Domain equalization  ,
  sort  = de,
  tag = nomencl
}

\DeclareAcronym{DFT}{
  short = \normalfont{DFT} ,
  long  = Discrete Fourier transform  ,
  sort  = dft,
  tag = nomencl
}

\DeclareAcronym{IDFT}{
  short = \normalfont{IDFT} ,
  long  = Inverse discrete Fourier transform  ,
  sort  = idft,
  tag = nomencl
}

\DeclareAcronym{CP}{
  short = \normalfont{CP} ,
  long  = Cyclic prefix  ,
  sort  = cp,
  tag = nomencl
}

\DeclareAcronym{MAI}{
  short = \normalfont{MAI} ,
  long  = Multiple-access interference  ,
  sort  = mai,
  tag = nomencl
}

\DeclareAcronym{RE}{
  short = \normalfont{RE} ,
  long  = Resource element  ,
  sort  = re,
  tag = nomencl
}

\DeclareAcronym{ML}{
  short = \normalfont{ML} ,
  long  =Maximum-likelihood   ,
  sort  = ml,
  tag = nomencl
}

\DeclareAcronym{FDE}{
  short = \normalfont{FDE} ,
  long  = Frequency domain equalization,
  sort  = fde,
  tag = nomencl
}

\DeclareAcronym{OFDM}{
  short = \normalfont{OFDM} ,
  long  = Orthogonal frequency-division multiplexing,
  sort  = ofdm,
  tag = nomencl
}

\DeclareAcronym{UD}{
  short = \normalfont{UD} ,
  long  = Uniquely decodable,
  sort  = ud,
  tag = nomencl
}

\DeclareAcronym{IrLDS}{
  short = \normalfont{IrLDS} ,
  long  = Irregular low-density spreading,
  sort  = irlds,
  tag = nomencl
}

\DeclareAcronym{VA}{
  short = \normalfont{VA} ,
  long  = Viterbi algorithm,
  sort  = va,
  tag = nomencl
}

\DeclareAcronym{SD}{
  short = \normalfont{SD} ,
  long  = Sphere-decoding,
  sort  = sd,
  tag = nomencl
}

\DeclareAcronym{SISO}{
  short = \normalfont{SISO} ,
  long  = Soft-input soft-output,
  sort  = siso,
  tag = nomencl
}

\DeclareAcronym{GML}{
  short = \normalfont{GML} ,
  long  = Global maximum likelihood detector,
  sort  = gml,
  tag = nomencl
}

\DeclareAcronym{SNR}{
  short = \normalfont{SNR} ,
  long  = Signal-to-noise ratio,
  sort  = snr,
  tag = nomencl
}

\DeclareAcronym{LRS}{
  short = \normalfont{LRS} ,
  long  = Large random spreading,
  sort  = lrs,
  tag = nomencl
}

\DeclareAcronym{NDA}{
  short = \normalfont{NDA} ,
  long  = Noiseless detection algorithm,
  sort  = nda,
  tag = nomencl
}

\DeclareAcronym{FDA}{
  short = \normalfont{FDA} ,
  long  = Fast detection algorithm,
  sort  = fda,
  tag = nomencl
}

\DeclareAcronym{QLSS}{
  short = \normalfont{QLSS} ,
  long  = Quasi-large sparse sequence,
  sort  = qlss,
  tag = nomencl
}

\DeclareAcronym{SER}{
  short = \normalfont{SER} ,
  long  = Symbol error rate,
  sort  = ser,
  tag = nomencl
}

\DeclareAcronym{MC}{
  short = \normalfont{MC} ,
  long  = Multicarrier,
  sort  = mc,
  tag = nomencl
}

\DeclareAcronym{JSFG}{
  short = \normalfont{JSFG} ,
  long  = Joint sparse factor graph,
  sort  = jsfg,
  tag = nomencl
}

\DeclareAcronym{SMI}{
  short = \normalfont{SMI} ,
  long  = Single user mutual information,
  sort  = smi,
  tag = nomencl
}

\DeclareAcronym{SVE}{
  short = \normalfont{SVE} ,
  long  = Spreading vector extension,
  sort  = sve,
  tag = nomencl
}

\DeclareAcronym{LDSM}{
  short = \normalfont{LDSM} ,
  long  = Low-density superposition modulation,
  sort  = ldsm,
  tag = nomencl
}

\DeclareAcronym{BBPSO}{
  short = \normalfont{BBPSO} ,
  long  = Bare-bone particle swarm optimization,
  sort  = bbpso,
  tag = nomencl
}

\DeclareAcronym{STS}{
  short = \normalfont{STS} ,
  long  = Steiner triple system,
  sort  = sts,
  tag = nomencl
}

\DeclareAcronym{BIBD}{
  short = \normalfont{BIBD} ,
  long  = Balanced incomplete block design,
  sort  = bibd,
  tag = nomencl
}

\DeclareAcronym{LDSMA}{
  short = \normalfont{LDSMA} ,
  long  = Low-density spreading multiple access,
  sort  = ldsma,
  tag = nomencl
}

\DeclareAcronym{GMAC}{
  short = \normalfont{GMAC} ,
  long  = Gaussian multiple access channel,
  sort  = gmac,
  tag = nomencl
}

\DeclareAcronym{BLER}{
  short = \normalfont{BLER} ,
  long  = Block error rate,
  sort  = bler,
  tag = nomencl
}

\DeclareAcronym{MAC}{
  short = \normalfont{MAC} ,
  long  = Multiple access channel,
  sort  = mac,
  tag = nomencl
}

\DeclareAcronym{EXIT}{
  short = \normalfont{EXIT} ,
  long  = Extrinsic information transfer,
  sort  = exit,
  tag = nomencl
}

\DeclareAcronym{SEP}{
  short = \normalfont{SEP} ,
  long  = Symbol error probability,
  sort  = sep,
  tag = nomencl
}

\DeclareAcronym{MS}{
  short = \normalfont{MS} ,
  long  = Mobile station,
  sort  = ms,
  tag = nomencl
}

\DeclareAcronym{SFBC}{
  short = \normalfont{SFBC} ,
  long  = Space-frequency block codes,
  sort  = sfbc,
  tag = nomencl
}

\DeclareAcronym{SM}{
  short = \normalfont{SM} ,
  long  = Spatial modulation,
  sort  = sm,
  tag = nomencl
}

\DeclareAcronym{SCDMA}{
  short = \normalfont{SCDMA} ,
  long  = Sparse code-division multiple-access,
  sort  = scdma,
  tag = nomencl
}

\DeclareAcronym{FBMC}{
  short = \normalfont{FBMC} ,
  long  = Filter-bank multicarrier,
  sort  = fbmc,
  tag = nomencl
}

\DeclareAcronym{LDWM}{
  short = \normalfont{LDWM} ,
  long  = Low-density weight matrix,
  sort  = ldwm,
  tag = nomencl
}

\DeclareAcronym{ULDS}{
  short = \normalfont{ULDS} ,
  long  = Ultra low-density spread,
  sort  = ulds,
  tag = nomencl
}

\DeclareAcronym{SSM}{
  short = \normalfont{SSM} ,
  long  = Sparse superposition matrix,
  sort  = ssm,
  tag = nomencl
}

\DeclareAcronym{SINR}{
  short = \normalfont{SINR} ,
  long  = Signal-to-interference-plus-noise ratio,
  sort  = sinr,
  tag = nomencl
}

\DeclareAcronym{TSC}{
  short = \normalfont{TSC} ,
  long  = Total squared correlation,
  sort  = tsc,
  tag = nomencl
}

\DeclareAcronym{WBE}{
  short = \normalfont{WBE} ,
  long  = Welch-bound-equality,
  sort  = wbe,
  tag = nomencl
}

\DeclareAcronym{MWBE}{
  short = \normalfont{MWBE} ,
  long  = Maximum-Welch-bound-equality,
  sort  = mwbe,
  tag = nomencl
}

\DeclareAcronym{CM}{
  short = \normalfont{CM} ,
  long  = Coded modulation,
  sort  = cm,
  tag = nomencl
}

\DeclareAcronym{MLCM}{
  short = \normalfont{MLCM} ,
  long  = Multilevel coded modulation,
  sort  = mlcl,
  tag = nomencl
}

\DeclareAcronym{TCM}{
  short = \normalfont{TCM} ,
  long  = Trellis-coded modulation,
  sort  = tcm,
  tag = nomencl
}

\DeclareAcronym{TTCM}{
  short = \normalfont{TTCM} ,
  long  = Turbo trellis-coded modulation,
  sort  = ttcm,
  tag = nomencl
}

\DeclareAcronym{ED}{
  short = \normalfont{ED} ,
  long  = Euclidean distance,
  sort  = ed,
  tag = nomencl
}

\DeclareAcronym{BICM}{
  short = \normalfont{BICM} ,
  long  = Bit-interleaved coded modulation,
  sort  = bicm,
  tag = nomencl
}

\DeclareAcronym{LTE}{
  short = \normalfont{LTE} ,
  long  = Long-term evolution,
  sort  = lte,
  tag = nomencl
}

\DeclareAcronym{QPP}{
  short = \normalfont{QPP} ,
  long  = Quadratic permutation polynomial,
  sort  = qpp,
  tag = nomencl
}

\DeclareAcronym{PEG}{
  short = \normalfont{PEG} ,
  long  = Progressive edge-growth,
  sort  = peg,
  tag = nomencl
}

\DeclareAcronym{eMBB}{
  short = \normalfont{eMBB} ,
  long  = Enhanced mobile broadband,
  sort  = embb,
  tag = nomencl
}

\DeclareAcronym{uRLLC}{
  short = \normalfont{uRLLC} ,
  long  = Ultra-reliable low-latency communications,
  sort  = urllc,
  tag = nomencl
}

\DeclareAcronym{mmWave}{
  short = \normalfont{mmWave} ,
  long  = Millimeter-wave,
  sort  = mmwave,
  tag = nomencl
}

\DeclareAcronym{RMS}{
  short = \normalfont{RMS} ,
  long  = Root-mean-square,
  sort  = rms,
  tag = nomencl
}

%% file: timeline.tex
\usepackage{tikz}
\usetikzlibrary{arrows, calc, decorations.markings, positioning,trees}

\makeatother